\documentclass[11pt,a4paper]{article}
\usepackage{graphicx}
\usepackage{amssymb,amsmath,color,mathtools,comment}
\usepackage{url}
\usepackage{listings,enumitem}
\usepackage[mathcal]{eucal}

\DeclareMathAlphabet\mathbfcal{OMS}{cmsy}{b}{n}

\usepackage[colorlinks=true, allcolors=blue]{hyperref}

\RequirePackage{natbib}

\usepackage[a4paper,margin=1in]{geometry}
\newcommand{\BEAS}{\begin{eqnarray*}}
\newcommand{\EEAS}{\end{eqnarray*}}
\newcommand{\BEA}{\begin{eqnarray}}
\newcommand{\EEA}{\end{eqnarray}}
\newcommand{\BEQ}{\begin{equation}}
\newcommand{\EEQ}{\end{equation}}
\newcommand{\BIT}{\begin{itemize}}
\newcommand{\EIT}{\end{itemize}}
\newcommand{\BNUM}{\begin{enumerate}}
\newcommand{\ENUM}{\end{enumerate}}
\newcommand{\BA}{\begin{array}}
\newcommand{\EA}{\end{array}}

\newcommand{\corr}{{\mathop { \rm corr{}}}}
\newcommand{\Var}{\mathop{ \rm var{}}}

\newcommand{\argmax}{\mathop{\rm argmax}}

\newcommand{\sign}{\mathop{ \rm sign}}

\newcommand{\rb}{\mathbb{R}}

\newcommand{\BlackBox}{\rule{1.5ex}{1.5ex}}

\newenvironment{proof}[1][Proof]{\par\noindent{\bf #1\ }}{\hfill\BlackBox\par\medskip}
\newtheorem{Lemma}{Lemma}
\newtheorem{Theorem}{Theorem}
\newtheorem{Condition}{Condition}
\newtheorem{Definition}{Definition}

\newtheorem{Proposition}{Proposition}

\def \E{{\mathbb E}}
\def \P{{\mathbb P}}

\def \X{{\mathcal X}}
\def \Y{{\mathcal Y}}

\def \P{{\mathbb P}}

\usepackage{stackengine}
\usepackage{scalerel}
\usepackage{xcolor}
\usepackage{booktabs,array,tabularx,threeparttable}

\usepackage[linesnumbered,ruled,vlined]{algorithm2e}
\SetKwInput{KwInput}{Input}
\SetKwInput{KwReturn}{Return}

\title{Generalized Engression Models}

\author{Xinwei Shen$^{1}$, Zijian Guo$^{2}$, and Francis Bach$^{3}$\\[1ex]
{\normalsize $^{1}$Department of Statistics, University of Washington}\\
{\normalsize $^{2}$Center for Data Science, Zhejiang University}\\
{\normalsize $^{3}$Inria, \smash{\'E}cole Normale Sup\'erieure, PSL Research University}}

\date{}

\begin{document}
\maketitle

\begin{abstract}
We consider estimating the conditional distribution of a multivariate outcome given covariates when its coordinates may be continuous, binary, categorical, ordinal or rankings, and are conditionally dependent on one another. Different statistical methods have been developed for each outcome type, and most of them target a summary of the conditional distribution, such as the mean of each coordinate, rather than the joint distribution of the outcome vector. We develop generalized engression models, a unified nonparametric distributional regression framework for outcomes of any type. The proposed method builds upon engression, a scoring-rule-based deep generative model, and introduces a data-type-specific link function and a stochastic perturbation that smooths the loss, enabling gradient-based training even with discontinuous links. We establish universal representation results for continuous, discrete and mixed outcomes. In simulations and in two applications, 242 species in a community ecology benchmark and a 17-dimensional mixed-type health outcome, the method matches type-specific models on marginal scores, improves on them on the joint distribution, and matches or exceeds purpose-built state-of-the-art joint species distribution models. Software is available in Python.

\end{abstract}

\section{Introduction}\label{sec:intro}

Statistics typically has a different regression model for every type of outcome. A continuous response is fitted by least squares or by quantile regression; a binary response by logistic or probit regression; a categorical one by the multinomial logit; an ordinal one by the proportional-odds model; a ranking by Plackett--Luce. Generalized linear models (GLMs) \citep{nelder1972generalized,mccullagh1989generalized} showed that several of these are one model,
\begin{equation}\label{eq:glm}
\E[Y\mid X]=h(\beta^\top X),
\end{equation}
in which the response $Y$ given the covariates $X$ follows a distribution in an exponential family, continuous, binary or categorical, $\beta^\top X$ is a linear predictor and $h$ a fixed function chosen for the type, the inverse of the link. Outcomes of different types share the linear predictor and the maximum likelihood estimation principle. However, outcomes outside the exponential family are beyond the scope of classical GLMs. The linear predictor, in turn, has been extended by generalized additive models \citep{hastie1990generalized}, tree-based models \citep{breiman2001random}, and neural networks \citep{tran2020bayesian}.

Whatever the type, the estimand is often the conditional mean or the class probabilities of $Y$ given $X$. Such summaries can miss much of what matters. For a single discrete response the class probabilities do characterize the distribution, but for multiple ones they say nothing about how the coordinates depend on one another; for a continuous response the mean can be a poor summary, for instance when the distribution is bimodal and the mean falls between the modes. The joint conditional distribution of the outcome vector carries all of this information, the uncertainty of each coordinate as well as the dependence between them. Distributional regression estimates the whole conditional distribution rather than a summary of it, through its quantiles \citep{koenker1978regression}, through all parameters of a response family \citep{rigby2005generalized} or through a transformation of the response \citep{hothorn2014conditional}; see \citet{kneib2023rage} for a review of the field. It has been developed mainly for a single response, usually a continuous one.

For a multivariate response, the dependence between the coordinates is either confined to a fixed family, as in the standard joint models and in copula models, or ignored altogether when each coordinate is fitted on its own. The standard joint models, the multivariate probit \citep{chib1998analysis} and the Plackett--Luce model \citep{luce1959individual,plackett1975analysis} among them, read the outcome off a latent vector with a linear index and either correlated Gaussian noise, whose dependence is determined by pairwise correlations, or independent Gumbel noise across items, which implies the independence of irrelevant alternatives; in their standard form, this latent dependence does not vary with the covariates.
For multivariate discrete outcomes, one might instead treat the whole outcome as a single categorical variable, but the number of categories grows too fast: $k$ binary labels give $2^k$ patterns, so that the $242$ species of Section~\ref{sec:jsdm} give $2^{242}$ patterns to be learned from $600$ training plots, and a ranking of $k$ items gives $k!$ orderings, far beyond what is computationally and statistically feasible with a moderate $k$.

An alternative to distributional regression is generative modeling.  
Generative models that generate images, text and molecules have become everyday technology. 
At their core, all of them solve a statistical problem, that of learning a probability distribution, usually a conditional one, from a sample, and they solve it by a route that classical statistics has rarely taken. Instead of specifying a density or a likelihood, one fits a map $g$ that turns covariates $x$ and a noise variable $\varepsilon$ into an outcome, so that the conditional distribution is represented by the draws $g(x,\varepsilon)$, and any functional of it, a mean, a quantile, the probability of an event defined by $Y$, can be estimated by simple Monte Carlo \citep{shen2024engression}. The representation is nonparametric, scales to high-dimensional outcomes and large samples, and needs no analytic form for the distribution. What has kept it from discrete data is that $g$ has to be differentiable, so that its output is usually continuous. Generative models have been extended to discrete data by methods specific to each type \citep{jang2017categorical,austin2021structured,kotelnikov2023tabddpm}, so that the type-specific habit of statistics reappears there. This paper handles outcomes of any type in one model, and in doing so puts the classical type-specific models and generative models under one roof. Section~\ref{sec:background} reviews both bodies of work.

Engression \citep{shen2024engression} is the generative model this paper builds on. It fits the generator by minimizing the energy score \citep{gneiting2007strictly}, a proper scoring rule by which probabilistic forecasts are judged \citep{gneiting2014probabilistic}, computed from draws of the model alone: no adversary to balance, no density to evaluate, no likelihood to integrate, only one loss minimized by gradient descent, which makes the fit cheap and stable. Engression has since been applied across the sciences, from climate emulation \citep{meinshausen2026gcmagicc,schillinger2026enscale} to epidemic forecasting \citep{pathak2026spatiotemporal}. These applications make the practical case for engression as a statistical tool, yet all of them model continuous outcomes; this paper extends the method to outcomes of any type.

This paper proposes \emph{generalized engression models} (GEMs), a unified framework for estimating the conditional distribution of an outcome of any type, continuous, binary, categorical, ordinal, a ranking, or a vector mixing these. A GEM is a conditional generative model of the form
\begin{equation*}
	Y=h\big(g(X,\varepsilon)+\sigma(X)\eta\big),
\end{equation*}
where $g$ is a neural network generating a continuous vector, $h$ a fixed link mapping that vector into the outcome space, $\varepsilon$ a standard Gaussian noise, and $\eta$ a perturbation with a learned scale $\sigma(X)$. As a GLM adapts linear regression to the type of the response through its link, a GEM extends engression through the link $h$. The link is the only ingredient that depends on the type: the identity for a continuous coordinate, the sign for a binary one, the indicator of the largest coordinate for a categorical one, rounding for an ordinal one, the rank vector for a ranking, and these block by block for a mixed outcome (Table~\ref{tab:links}). Everything else, the generator, the loss and the fitting algorithm, is the same for every type; with the identity link the model is engression. 

Set against GLMs in \eqref{eq:glm}, $g(X,\varepsilon)+\sigma(X)\eta$ in GEMs takes the place of the linear predictor $\beta^\top X$, and the outcome $Y$ itself that of its conditional mean. GEMs are thus generative models that produce samples of $Y$ given $X$ from random noise, by which the full conditional distribution of $Y$ is estimated. A neural network $g$ also makes the model flexible: it imposes no parametric form on the response distribution, which, unlike in GLMs, need not belong to an exponential family. The links also differ: the link of a GLM is a smooth, invertible transformation of the mean, and $h$ in \eqref{eq:glm} is its inverse; ours acts on a continuous latent vector, as in the latent-variable form of the probit model, and is often neither smooth nor invertible. 

With a linear generator and a constant scale, GEMs reduce to classical latent-variable models such as the multivariate probit (see Section~\ref{sec:model}). Beyond these special cases, Section~\ref{sec:universality} establishes universal representation of GEMs: the conditional distribution of an outcome of any type can be written as a GEM with a measurable generator $g$, and, under positivity and continuity conditions, with a continuous $g$, which a neural network can approximate.

For fitting GEMs, we use the energy score as the loss, as in engression, and the perturbation $\eta$ is what makes a discrete link tractable. For discontinuous links, the gradient of the empirical risk is zero almost everywhere. Averaging over $\eta$ restores it: the outcome still takes discrete values, but the probability of each value varies smoothly with $g$, so the expected loss is differentiable even though the link is not, in the spirit of perturbed optimizers \citep{berthet2020learning}. The gradient is an expectation over the perturbation that involves only the values of the loss, not its derivatives (Proposition~\ref{prop:grad}), so it can be estimated by Monte Carlo.

GEM is implemented in the Python package \texttt{engression}\footnote{\url{https://github.com/xwshen51/engression}}. A fit takes one call, and the same call serves every outcome type; only the \texttt{data\_type} argument changes:
\makeatletter\lst@AddToHook{OnEmptyLine}{\vspace{-.6\baselineskip}}\makeatother
\begin{lstlisting}[basicstyle=\ttfamily\small]
from engression import engression

fit = engression(x, y)                                   # continuous outcome
fit = engression(x, y, data_type="multilabel")           # several binary labels
fit = engression(x, y, data_type="ranking")              # ranking of columns of y
fit = engression(x, y, data_type={"continuous": 0,       # mixed outcome: one link
                                  "multilabel": 7,       # per block, each given by
                                  "multiclass": 12})     # its first column

y_draws = fit.sample(x_new, sample_size=1000)            # draws of Y given x_new
\end{lstlisting}
The draws in \texttt{y\_draws} approximate the conditional distribution of $Y$ given $X=x_{\mathrm{new}}$ and can be used to estimate any functional of it. 

\begin{figure}[!t]
\centering
\begin{tabular}{@{}c@{}}
	(a) Multi-label outcome: three binary labels with strong co-occurrence\\
	\includegraphics[width=.38\textwidth]{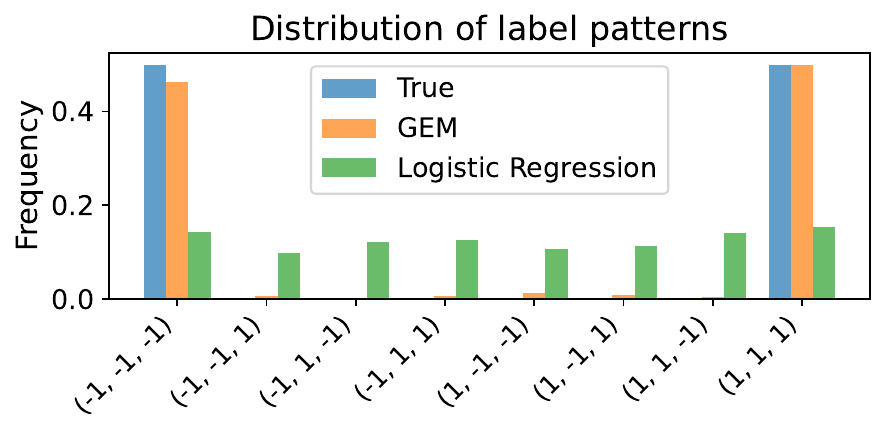}~~~~
	\includegraphics[width=.52\textwidth]{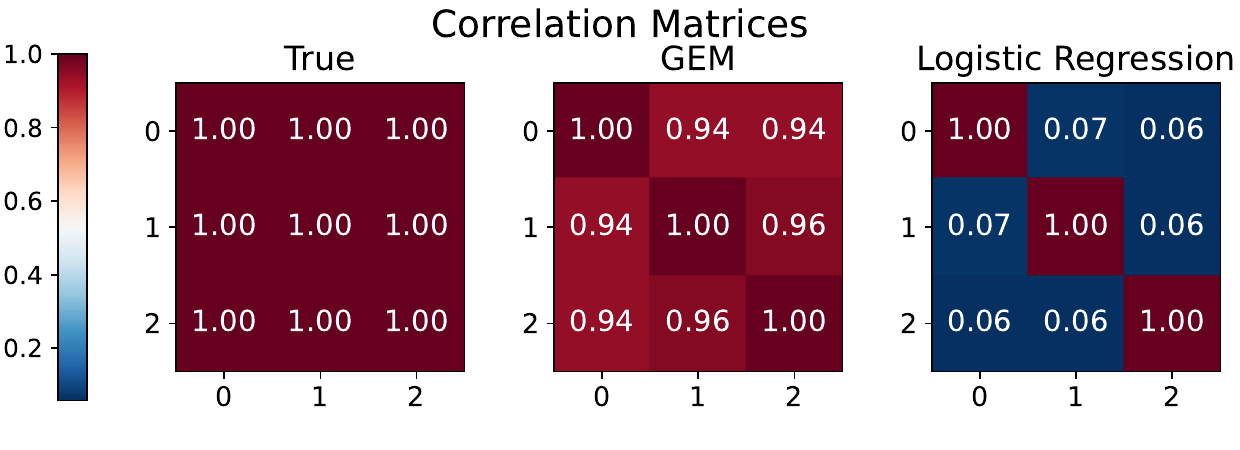}\\[2pt]
	(b) Mixed-type outcome: one continuous and one categorical variable with four levels\\
	\includegraphics[width=.96\textwidth]{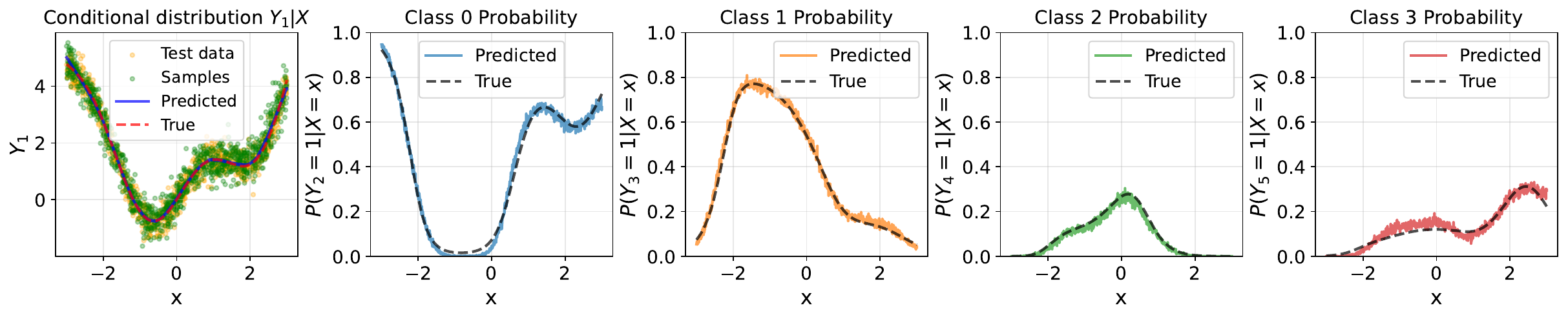}\\[2pt]
	(c) Ranking outcome: an ordering of three items\\
	\includegraphics[width=.9\textwidth]{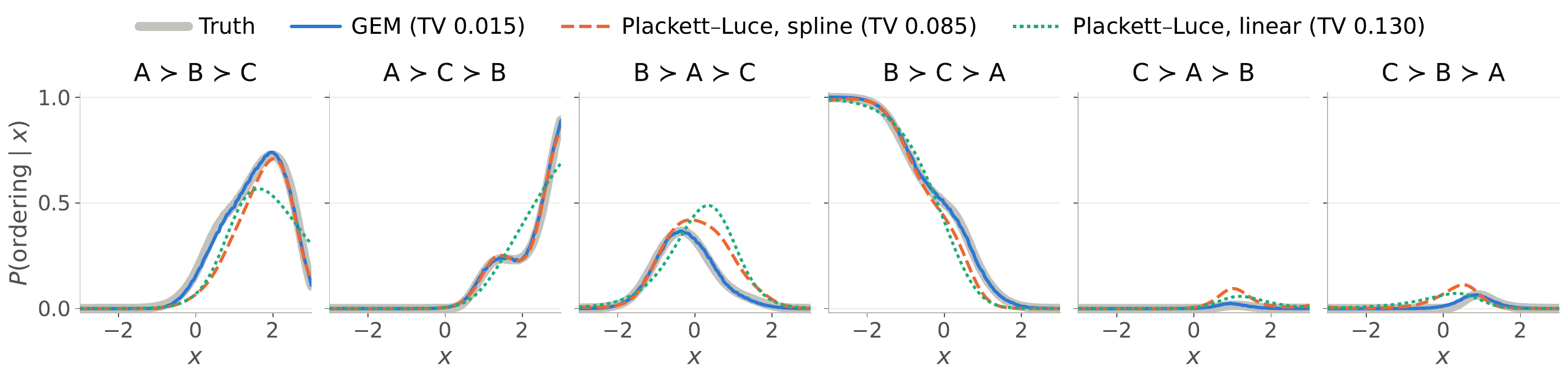}
\end{tabular}
\caption{GEMs on various outcome types, standard type-specific models. (a) Distribution of the eight label patterns and correlation between the labels: logistic regressions have the correct marginals but, drawing the labels independently, spread mass over all patterns and miss the correlation; GEM recovers both. (b) Conditional distribution and mean of the continuous coordinate and the four class probabilities as functions of $x$, recovered from one model of the joint outcome. (c) Probability of each of the six orderings as a function of $x$: the Plackett--Luce model cannot match them with linear or with spline scores, whereas GEM does, with a total-variation (TV) error of $0.015$. GEMs are fitted with the same algorithm in all three panels; only the link differs.}
\label{fig:illustration}
\end{figure}

Figure~\ref{fig:illustration} gives a first impression of what GEM does on various outcome types, against standard type-specific models. For three co-occurring binary labels, GEM recovers the distribution of label patterns and the correlation between the labels, which independent logistic regressions, correct on the marginals, cannot. For an outcome with a continuous and a categorical coordinate it recovers both conditional laws from a single model of the pair. For a ranking of three items it matches the probability of every ordering as a function of $x$, where the Plackett--Luce model cannot, with linear or with spline scores, because the data violate the independence of irrelevant alternatives. The simulations of Section~\ref{sec:simulation} conduct the comparison at scale, where GEM remains competitive.

The range of outcomes the model accepts opens up many applications; Section~\ref{sec:applications} develops two. The first is from community ecology, where the outcome is the vector of presences and absences of many species at a site. Joint species distribution modeling is a large field \citep{pollock2014understanding,warton2015many,ovaskainen2017hmsc,tikhonov2020joint,pichler2021new}, motivated by questions only a joint model can answer: the distribution of the number of species a site will hold, with a prediction interval that can be trusted, and which species co-occur beyond what their shared environmental factors explain \citep{galiana2024power}. The leading implementations model the joint presence--absence vector through a multivariate probit or logistic regression with a low-rank Gaussian latent factor structure to induce species correlations. While they have advanced the field considerably, they share important limitations: the linear--Gaussian latent factor structure captures only pairwise, Gaussian correlations between species, while higher-order or nonlinear co-occurrence patterns---such as multi-species guilds or mutually exclusive competitors---cannot be represented. Moreover, inference by Markov chain Monte Carlo (MCMC), as in \citet{ovaskainen2017hmsc}, is computationally expensive for large species richness or many sites, which has motivated faster gradient-based implementations \citep{pichler2021new}. GEM with the sign link offers a nonparametric generative model that can represent higher-order co-occurrence dependencies and is trained efficiently via gradient descent without MCMC. On a published benchmark of arctic vegetation \citep{niittynen2018importance,norberg2019comprehensive}, $242$ species on $1{,}200$ plots with half held out, GEM matches or exceeds the purpose-built joint models on the joint predictive distribution and is far better calibrated for species richness than models fitting the species one at a time. It also has a clear computational advantage: one fit takes seconds on one GPU, against minutes to hours for the other joint models.

The second application is a health survey, in which each of $27{,}983$ participants contributes seven continuous measurements, five diagnosis and medication binary indicators and an ordinal self-rated health score. The quantities of clinical interest are functionals of their joint law, such as the probability of undiagnosed disease, a laboratory value above a diagnostic threshold together with the absence of a diagnosis; a block-by-block model reaches them only by assuming the blocks conditionally independent or through a copula, and a point prediction not at all. GEM fits all three blocks jointly in one model, improving on stacked-marginal, Gaussian-copula and tree-based alternatives on the joint distribution on every split and on the continuous and binary marginals on almost every split, and coming close on the ordinal one.

The rest of the paper is organized as follows. Section~\ref{sec:background} reviews the type-specific models and the generative models and recalls engression. Section~\ref{sec:method} sets out the model, the links, the closed-form gradient and the fitting algorithm. Section~\ref{sec:universality} establishes that the model can represent any conditional distribution of a continuous, discrete or mixed outcome. Sections~\ref{sec:simulation} and~\ref{sec:applications} report the simulations and the two applications, and Section~\ref{sec:discussion} concludes. Proofs and additional results and experimental details are in the appendix.

\section{Background}\label{sec:background}

This section reviews the two bodies of work, the type-specific regression models of statistics and the generative models of machine learning, and recalls engression, on which the proposal builds.

\subsection{Regression models for outcomes of a given type}\label{sec:classical}

Most classical regression models can be written in a latent-variable form,
\begin{equation*}
	Y=h(BX+\xi),
\end{equation*}
where $BX\in\rb^k$ is a vector of linear predictors, $\xi\in\rb^k$ is noise independent of $X$, and $h$ reads the outcome off the latent vector by a rule fixed by the type of $Y$, as the links of Table~\ref{tab:links} do, block by block for a mixed outcome \citep{sammel1997latent}, except that the cut points of an ordinal outcome are estimated along with $B$. The identity link with Gaussian noise gives the linear model. For discrete outcomes \citep[see, e.g.,][]{agresti2013categorical}, independent standard logistic or Gumbel noise gives the logit models, among them the proportional-odds model \citep{mccullagh1980regression}, the multinomial logit \citep{mcfadden1974conditional} and, for rankings, the rank-ordered logit \citep{beggs1981assessing,hausman1987specifying}; Gaussian noise gives the probit models and the model of \citet{thurstone1927law}. For a single binary or categorical response, averaging the rule over the noise gives the smooth inverse link of the GLM~\eqref{eq:glm}. For several coordinates, the dependence between them comes from the covariance of a Gaussian $\xi$: the multivariate probit \citep{ashford1970multi,chib1998analysis} leaves it unrestricted, joint species distribution models \citep{warton2015many,ovaskainen2017hmsc,niku2019gllvm} build it from a few latent factors, $\xi=\Lambda\varepsilon+\eta$, and the Gaussian copula reads mixed continuous and ordered coordinates off $\xi$ through monotone maps that set their marginal laws, left free by \citet{hoff2007extending} and given by GLMs in $X$ by \citet{song2009joint}.

These models are fitted through their likelihood; for a discrete outcome it is the probability that the latent vector falls in the set $h^{-1}(y)$, a closed form in the logit models and otherwise an integral over the noise, which calls for simulation, Markov chain Monte Carlo or variational approximations. GEM has the same form, with a neural network $g(X,\varepsilon)$ in place of $BX+\Lambda\varepsilon$ and a scale $\sigma(X)$ on $\eta$, and needs no likelihood, being fitted from its draws. Multivariate distributional regression makes the marginal laws flexible but keeps the dependence in a parametric copula, whose parameters may depend on $X$ \citep[see][]{klein2022multivariate,kneib2023rage}; distributional random forests \citep{cevid2022drf} are fully nonparametric, but their draws are confined to the observed outcome vectors, which they reweight.

\subsection{Generative models}
Generative models differ in the neural network that serves as the generator $g$ and in the criterion by which it is fitted. 
The multilayer perceptron, which alternates affine maps and coordinate-wise nonlinearities, is the basic neural network and the one we use throughout; generative models of images and text use convolutional networks and transformers \citep{vaswani2017attention} instead, and either could serve as $g$ here. As to the criterion, generative adversarial networks \citep{goodfellow2014generative} train $g$ through a minimax game against a discriminator, variational autoencoders \citep{kingma2014auto} by a variational lower bound on the likelihood, normalizing flows \citep{papamakarios2021normalizing} by the explicit likelihood of an invertible $g$, diffusion models \citep{ho2020denoising} by a denoising loss, and engression, recalled below, by a proper scoring rule.

For a discrete outcome, the common route is to model probabilities rather than draws, so that a likelihood is available: autoregressive models, large language models among them \citep{brown2020language}, take the coordinates one at a time, each given the previous ones, as classifier chains \citep{read2011classifier} do for binary labels, and discrete diffusion \citep{austin2021structured} reverses, step by step, a Markov chain that corrupts the categories. Generative models of mixed-type tables handle each discrete column by a continuous relaxation of its categorical draw \citep{xu2019modeling} or by discrete diffusion \citep{kotelnikov2023tabddpm}.

\subsection{Engression and the energy score}

Engression \citep{shen2024engression} is a generative model fitted by a proper scoring rule. Let $X\in\X$ be the covariates and $Y\in\Y\subseteq\rb^k$ the outcome. The target is the conditional distribution of $Y$ given $X=x$, denoted $P^*_{Y|X=x}$; $X$ may be absent, in which case the target is the marginal distribution of $Y$. The model is a map $g:\X\times\rb^m\to\rb^k$ applied to the covariates and a noise vector $\varepsilon\sim\mathcal N(0,I_m)$ independent of $X$, and for each $x$ the law of $g(x,\varepsilon)$ is the model's conditional distribution of $Y$ given $X=x$. Engression fits $g$ by minimizing the expected energy score \citep{gneiting2007strictly}: for a predictive law $P$ and an observation $y$ the energy score is $\E_P\|\hat Y-y\|^\beta-\frac12\E_P\|\hat Y-\hat Y'\|^\beta$, with $\hat Y,\hat Y'$ independent draws from $P$ and an exponent $\beta\in(0,2)$, and it is strictly proper on the distributions with a finite moment of order $\beta$. Throughout the paper we use the default $\beta=1$, for which the energy score of a scalar outcome is the continuous ranked probability score of the predictive distribution function $F$, $\mathrm{CRPS}(F,y)=\E_F|\hat Y-y|-\frac12\E_F|\hat Y-\hat Y'|$ \citep{gneiting2007strictly}.\footnote{\citet{shen2024engression} define the energy score and the CRPS with the opposite sign, as in \citet{gneiting2007strictly}, so that higher values are better, and minimize the negative energy score, which they call the energy loss; in this paper we use negatively oriented scores so that lower values are better.} The population objective is therefore
\begin{equation*}
	\E\bigg[\|Y-g(X,\varepsilon)\| - \frac12\|g(X,\varepsilon) - g(X,\varepsilon')\|\bigg],
\end{equation*}
where $(X,Y)$ follows the data distribution, $\varepsilon,\varepsilon'\sim\mathcal N(0,I_m)$ are independent, and the norm is the Euclidean norm. Since the score is strictly proper, if a class $\mathcal G$ contains a map $g$ with $g(x,\varepsilon)\sim P^*_{Y|X=x}$ for almost every $x$, the minimizers of the objective over $\mathcal G$ are exactly such maps \citep[Proposition~1]{shen2024engression}. In practice the population objective is replaced by its empirical counterpart over a sample, $\mathcal G$ is a class of neural networks, and the minimization is by stochastic gradient descent with fresh draws of $\varepsilon,\varepsilon'$ at each step.

With a neural network as its generator, engression is a model for a continuous outcome: its draws $g(x,\varepsilon)$ lie in $\rb^k$, so for a binary, categorical or ranked outcome they are not valid values of $Y$, and the fitted law is a continuous distribution spread around the discrete values rather than a distribution on them. What it does offer the discrete case is the energy score itself, which remains strictly proper for discrete and mixed outcomes. To extend it to outcomes of any type, we constrain the output of the generator to the support of the outcome, keep the energy score as the loss, and work out how to compute the gradient of the resulting objective with discontinuous links. The next section sets out the procedure.

\section{Generalized engression}\label{sec:method}
We define the model (Section~\ref{sec:model}), derive the gradient of its energy-score objective, which exists even when the link is discontinuous (Section~\ref{sec:gradient}), and give the fitting algorithm (Section~\ref{sec:algorithm}); Section~\ref{sec:kernel} replaces the energy score by other kernel scores.

\subsection{Model and links}\label{sec:model}

With $X\in\X$ and $Y\in\Y\subseteq\rb^k$ as in Section~\ref{sec:background}, the outcome may now be of any of the types of Table~\ref{tab:links}, or a vector of several of them. GEM represents the conditional distribution of $Y$ given $X$ by
\begin{equation}\label{eq:gem_model}
	Y=h\big(g(X,\varepsilon)+\sigma(X)\odot\eta\big),
\end{equation}
where $g:\X\times\rb^m\to\rb^k$, with $\varepsilon\sim\mathcal N(0,I_m)$, is a neural network generator as in engression, whose output is a continuous vector; $h:\rb^k\to\Y$ is the link, a fixed map from this vector to the outcome space; $\eta\sim\mathcal N(0,I_k)$ is a perturbation of the generator's output, introduced for the computational reason detailed in Section~\ref{sec:gradient}, whose scale $\sigma:\X\to(0,\infty)^k$ is a second neural network with a positive output layer; and $\odot$ and $\oslash$ denote the coordinate-wise product and quotient. We recommend a scale common to all coordinates, $\sigma(x)=\sigma_0(x)\mathbf 1_k$, as the default, and coordinate-wise scales for harder cases, such as a mixed outcome whose blocks are on different scales, which a single scale cannot suit at once. With the identity link the model recovers engression.

The two noises play different parts. The noise $\varepsilon$ is what makes the model distributional: passed through the network $g$ together with $X$, it yields a rich class of conditional distributions of $Y$ given $X$, both in the marginal laws of the coordinates and in the dependence between them. The perturbation $\eta$, with its learned scale, is what makes a discontinuous link tractable.
Both noises are independent of $X$ and standard Gaussian. Neither choice is essential: any atomless real-valued noise can play the part of $\varepsilon$ without changing the conditional laws the model can represent (Lemma~\ref{lem:noise_outsourcing}), and any continuously differentiable density whose value and gradient decay fast enough can replace the Gaussian $\eta$ (Appendix~\ref{app:proof_grad}). We sample from the conditional distribution of $Y|X=x$ by $h(g(x,\varepsilon)+\sigma(x)\odot\eta)$ with draws of $(\varepsilon,\eta)$, and any functional of the conditional law, such as the mean or quantiles, is estimated from them by Monte Carlo. 

With a linear generator $g(x,\varepsilon)=Bx+\Lambda\varepsilon$ and a constant scale $\sigma(x)=\sigma_0\mathbf 1_k$, GEM takes the latent-variable form of Section~\ref{sec:classical} with Gaussian noise $\xi=\Lambda\varepsilon+\sigma_0\eta$, whose covariance is unrestricted when $m\ge k$ and has the factor structure of joint species distribution models when $m<k$. The sign link then gives the multivariate probit, the argmax link the multinomial probit and the ranking link Thurstone's model; the ordinal link, whose cut points are fixed, gives the ordered probit with equally spaced cut points. With $\Lambda=0$ and a logistic or Gumbel $\eta$ in place of the Gaussian one, the same links give the logit models, again with equally spaced cut points for an ordinal outcome, and a nonzero $\Lambda$ adds Gaussian random effects to them. With a neural network $g$, GEM given $\varepsilon$ is still one of these models, with index $g(x,\varepsilon)$ and scale $\sigma(x)$, and mixing over $\varepsilon$ gives a much larger class (Section~\ref{sec:universality}).

The generator and the scale are learned, the link is fixed by the type of $Y$ (Table~\ref{tab:links}), and the noise dimension $m$ is a modeling choice, like the width of the networks. The next two sections show how the model is fitted, and Section~\ref{sec:universality} which conditional distributions it can represent.

\begin{table}[htp!]
	\centering
	\setlength{\belowcaptionskip}{4pt}
	\caption{Links by data type. $e_i$ is the $i$th unit vector, an ordinal level is coded by its value, and a ranking by its rank vector, whose $i$th entry is the position of item $i$. The ordinal link rounds $z$ to the nearest level. Ties, which have probability zero, are broken by a fixed rule.}
	\label{tab:links}
	\small
	\begin{tabular}{@{}llll@{}}
		\toprule
		data type & $\Y$ & $h(z)$& representation \\
		\midrule
		continuous, $k$ variables & $\rb^k$ & $z$& Theorem~\ref{thm:continuous}\\
		multiple binary, $k$ labels & $\{-1,+1\}^k$ & $\sign(z)$, coordinate-wise& Theorem~\ref{thm:discrete_continuous}\\
		categorical, $k$ classes & $\{e_1,\dots,e_k\}$ & $e_{i^*}$ with $i^*=\argmax_i z_i$& Theorem~\ref{thm:discrete_continuous}\\
		ordinal, $L$ levels & $\{1,\dots,L\}$ & $1+\sum_{l=1}^{L-1}\mathbf 1\{z>l+\frac12\}$& Theorem~\ref{thm:discrete_continuous}\\
		ranking of $k$ items & permutations of $(1,\dots,k)$ & $h(z)_i=1+\#\{l:z_l>z_i\}$& Theorem~\ref{thm:discrete_continuous}\\
		mixed & product of the above & concatenation of the above & Theorem~\ref{thm:mixed}\\
		\bottomrule
	\end{tabular}
\end{table}

\subsection{Objective and gradient}\label{sec:gradient}

The model is fitted by minimizing the expected energy score of the model's conditional distribution at the observed outcome. Writing $\hat Y=h(g(X,\varepsilon)+\sigma(X)\odot\eta)$ and $\hat Y'=h(g(X,\varepsilon')+\sigma(X)\odot\eta')$ for two draws given the same $X$, the objective, at the population level, is
\begin{equation}\label{eq:gem_objective}
	\mathcal{L}(g,\sigma)=\E\Big[\|Y-\hat Y\|-\frac12\|\hat Y-\hat Y'\|\Big],
\end{equation}
where $(X,Y)$ follows the data distribution, $\varepsilon,\varepsilon'\sim\mathcal N(0,I_m)$ and $\eta,\eta'\sim\mathcal N(0,I_k)$ are independent, and the expectation is over all of these. Given $X=x$, the conditional expectation of the integrand is the expected energy score of the model's conditional law $P^{g,\sigma}_{Y|X=x}$ at an observation drawn from $P^*_{Y|X=x}$. Since the energy score is strictly proper, minimizing $\mathcal L$ recovers the conditional distribution whenever the model class can express it, as the following lemma states.
\begin{Lemma}\label{lem:minimizers}
	Let the model class be a set of pairs $(g,\sigma)$, and assume that $\E\|Y\|<\infty$ and that $\E\|\hat Y\|<\infty$ for every pair in the class, as is the case for any discrete outcome. If the class contains a pair with $P^{g,\sigma}_{Y|X=x}=P^*_{Y|X=x}$ for almost every $x$ under the law of $X$, then the minimizers of $\mathcal L$ over the class are exactly the pairs with this property.
\end{Lemma}
The lemma is a direct implication of Proposition~1 of \citet{shen2024engression}. Whether such a model exists is the subject of Section~\ref{sec:universality}.

Engression minimizes its objective by stochastic gradient descent, differentiating the integrand through the generator directly. With a discrete link this fails: for fixed $(x,y,\varepsilon,\varepsilon',\eta,\eta')$ the integrand of \eqref{eq:gem_objective} is piecewise constant in the outputs of $g$ and $\sigma$, so that its derivative is zero almost everywhere. Averaged over the perturbations $\eta$ and $\eta'$, however, the loss is continuously differentiable in the outputs of $g$ and $\sigma$. To state its gradient, fix $(X,Y,\varepsilon,\varepsilon')$ and write $\mu=g(X,\varepsilon)$, $\mu'=g(X,\varepsilon')$ and $s=\sigma(X)$ for the outputs, so that $\hat Y=h(\mu+s\odot\eta)$ and $\hat Y'=h(\mu'+s\odot\eta')$, and let
\begin{equation*}
	c=\|Y-\hat Y\|-\|\hat Y-\hat Y'\|,\qquad c'=\|Y-\hat Y'\|-\|\hat Y-\hat Y'\|
\end{equation*}
be the costs of the two draws.

\begin{Proposition}\label{prop:grad}
	Let $\bar\ell(\mu,\mu',s)$ be the expectation over $(\eta,\eta')$ of $\frac12\|Y-\hat Y\|+\frac12\|Y-\hat Y'\|-\frac12\|\hat Y-\hat Y'\|$, so that $\mathcal L=\E\,\bar\ell(\mu,\mu',s)$ since the two draws are exchangeable. Then $\bar\ell$ is continuously differentiable on $\rb^k\times\rb^k\times(0,\infty)^k$, with
	\begin{align*}
		G&=\nabla_\mu\bar\ell=\tfrac12\E[c\,\eta]\oslash s,\qquad G'=\nabla_{\mu'}\bar\ell=\tfrac12\E[c'\,\eta']\oslash s,\\
		S&=\nabla_s\bar\ell=\tfrac12\E\big[c\,(\eta\odot\eta-\mathbf 1_k)+c'\,(\eta'\odot\eta'-\mathbf 1_k)\big]\oslash s,
	\end{align*}
	where the expectations are over $(\eta,\eta')$ at fixed $(X,Y,\varepsilon,\varepsilon')$.
\end{Proposition}
The proof is in Appendix~\ref{app:proof_grad} and follows the same principle as Stein's lemma. 
The price of never differentiating the link is that Monte Carlo estimates of $G$, $G'$ and $S$ can have a large variance: it grows like $s_j^{-2}$ as $s_j$ shrinks, and it is inflated by the variation of the one scalar cost $c$ that multiplies the weights $\eta_j$ of all coordinates. Two refinements in Section~\ref{sec:algorithm} reduce it for continuous and for binary coordinates. 

Let $\theta_g$ and $\theta_\sigma$ be the parameters of the networks $g$ and $\sigma$, so that $\mu=g(X,\varepsilon;\theta_g)$, $\mu'=g(X,\varepsilon';\theta_g)$ and $s=\sigma(X;\theta_\sigma)$. Whenever differentiation and expectation may be exchanged, which holds under the usual conditions on the networks provided $\sigma$ stays bounded away from zero, the chain rule gives
\begin{equation*}
	\nabla_{\theta_g}\mathcal L=\E\big[J_{\theta_g}(\mu)^\top G+J_{\theta_g}(\mu')^\top G'\big],\qquad\nabla_{\theta_\sigma}\mathcal L=\E\big[J_{\theta_\sigma}(s)^\top S\big],
\end{equation*}
where $J_\theta$ is the Jacobian of a network output with respect to $\theta$; a product $J^\top G$ is computed without forming the Jacobian, by backpropagating $G$ through the network.

\subsection{Algorithm}\label{sec:algorithm}

With a sample $\{(X_i,Y_i):i=1,\dots,n\}$, the expectations in $\nabla_{\theta_g}\mathcal L$ and $\nabla_{\theta_\sigma}\mathcal L$ are replaced by averages over a mini-batch, with fresh draws of $\varepsilon,\varepsilon',\eta,\eta'$ at every step and, as in engression, two draws per observation, the fewest that the energy score allows; more draws per observation can be used, but we use two by default. Algorithm~\ref{alg:gem} gives the procedure. The quantities $\widehat G_i$, $\widehat G'_i$ and $\widehat S_i$ are unbiased for $G$, $G'$ and $S$ at $(X_i,Y_i,\varepsilon_i,\varepsilon'_i)$, so that, when the mini-batch is drawn uniformly, $\Delta_g$ and $\Delta_\sigma$ are unbiased for the gradients of the empirical objective. Any first-order method can use these estimates in place of the plain update, and the number of iterations, which acts as a regularization parameter, is chosen by the energy score on held-out data. Under the usual conditions for stochastic gradient methods, the algorithm converges to stationary points of the empirical objective \citep{bottou2018optimization}. 

Two refinements of Algorithm~\ref{alg:gem} reduce the variance of its gradient estimator. The first concerns continuous coordinates, on which the link is the identity and can be differentiated directly. On the coordinates $j$ of a continuous block, where $\hat Y_{ij}=\mu_{ij}+s_{ij}\eta_{ij}$, $\widehat G_{ij}$ and $\widehat G'_{ij}$ are replaced by pathwise derivatives, as engression computes them: $\widehat G_{ij}=\frac12\partial c_i/\partial\mu_{ij}$, obtained by differentiating the joint norms through $\mu_i+s_i\odot\eta_i$ with the coordinates of the discrete blocks of both draws held fixed, that is,
\begin{equation*}
	\widehat G_{ij}=\frac12\bigg(\frac{\hat Y_{ij}-Y_{ij}}{\|\hat Y_i-Y_i\|}-\frac{\hat Y_{ij}-\hat Y'_{ij}}{\|\hat Y_i-\hat Y'_i\|}\bigg),
\end{equation*}
where a fraction with a zero denominator, whose numerator is then zero as well, is set to $0$, and likewise $\widehat G'_{ij}=\frac12\partial c'_i/\partial\mu'_{ij}$. These estimates are also unbiased, because the norms are Lipschitz in these coordinates, and their variance is typically much smaller: each is at most $1$ in absolute value, whereas $\frac12c_i\eta_{ij}/s_{ij}$ grows as $s_{ij}$ shrinks. The scale can be differentiated pathwise as well, with $\widehat S_{ij}=\eta_{ij}\widehat G_{ij}+\eta'_{ij}\widehat G'_{ij}$. For mixed-type outcomes, the norms must be the joint ones.

\begin{algorithm}[!t]
\KwInput{Sample $\{(X_i,Y_i):i=1,\dots,n\}$, link $h$, learning rate $\alpha$, iterations $T$, initial $\theta_g,\theta_\sigma$}
\For{$t=1,\dots,T$}{
    Sample a mini-batch $\mathcal{B} \subseteq \{1,\dots,n\}$\;
    Sample $\varepsilon_i, \varepsilon'_i \sim \mathcal{N}(0, I_m)$ and $\eta_i, \eta'_i\sim\mathcal N(0,I_k)$, independently, for each $i\in \mathcal{B}$\;
    \For{each $i \in \mathcal{B}$}{
        Compute the network outputs $\mu_i\gets g(X_i,\varepsilon_i;\theta_g)$,\quad $\mu'_i\gets g(X_i,\varepsilon'_i;\theta_g)$,\quad $s_i\gets\sigma(X_i;\theta_\sigma)$\;
        Generate two draws $\hat Y_i\gets h(\mu_i+s_i\odot\eta_i)$,\quad $\hat Y'_i\gets h(\mu'_i+s_i\odot\eta'_i)$\;
        Compute the costs $c_i\gets\|Y_i-\hat Y_i\|-\|\hat Y_i-\hat Y'_i\|$,\quad $c'_i\gets\|Y_i-\hat Y'_i\|-\|\hat Y_i-\hat Y'_i\|$\;
        Estimate the gradients in $\mu_i$ and $\mu'_i$: $\widehat G_i\gets\frac12c_i\,\eta_i\oslash s_i$,\quad $\widehat G'_i\gets\frac12c'_i\,\eta'_i\oslash s_i$\;
        Estimate the gradient in $s_i$: $\widehat S_i\gets\frac12\big(c_i(\eta_i\odot\eta_i-\mathbf 1_k)+c'_i(\eta'_i\odot\eta'_i-\mathbf 1_k)\big)\oslash s_i$\;
    }
    Estimate the gradient in $\theta_g$ by backpropagation: $\Delta_g\gets|\mathcal B|^{-1}\sum_{i\in\mathcal B}\big(J_{\theta_g}(\mu_i)^\top\widehat G_i+J_{\theta_g}(\mu'_i)^\top\widehat G'_i\big)$\;
    Estimate the gradient in $\theta_\sigma$ by backpropagation: $\Delta_\sigma\gets|\mathcal B|^{-1}\sum_{i\in\mathcal B}J_{\theta_\sigma}(s_i)^\top\widehat S_i$\;
    Take a gradient step $\theta_g\gets\theta_g-\alpha\Delta_g$,\quad $\theta_\sigma\gets\theta_\sigma-\alpha\Delta_\sigma$\;
}
\KwReturn{Fitted parameters $\theta_g,\theta_\sigma$}
\caption{Generalized engression 
}
\label{alg:gem}
\end{algorithm}

The second refinement concerns binary coordinates. In $\widehat G_{ij}=\frac12c_i\eta_{ij}/s_{ij}$, only the part of the cost $c_i$ that changes with $\hat Y_{ij}$ contributes to the mean, since $\eta_{ij}$ has mean zero and is independent of everything else in $c_i$; but the whole of $c_i$ contributes to the variance, which is of the order of $\E[c_i^2]/s_{ij}^2$. Under the sign link with many labels, $c_i$ is a difference of norms of vectors in $\{-1,+1\}^k$, set by all the labels together, whereas the part of it that depends on label $j$, the change of the norms when $\hat Y_{ij}$ flips, is small when many other labels differ. Most of the variance then comes from the other labels, and a control variate removes most of it.

For a coordinate $j$ under the sign link, write $\hat y$ for a draw, $w$ for the point it is compared with, the observation or the other draw, and $A_j=\sum_{l\ne j}(\hat y_l-w_l)^2$, which does not involve coordinate $j$. Since $\hat y_j$ and $w_j$ are $\pm1$, $(\hat y_j-w_j)^2$ is $0$ or $4$, so the norm takes one of two values,
\begin{equation*}
	\|\hat y-w\|=\sqrt{A_j}\quad\text{if }\hat y_j=w_j,\qquad\|\hat y-w\|=\sqrt{A_j+4}\quad\text{if }\hat y_j\ne w_j.
\end{equation*}
Only the difference between the two values reflects coordinate $j$; their common level, of the order of $\sqrt{A_j}$, is set by the other coordinates. The control variate subtracts from the norm the midpoint $b_j=\frac12(\sqrt{A_j+4}+\sqrt{A_j})$ of the two values, which leaves
\begin{equation*}
	\|\hat y-w\|-b_j=\pm\tfrac12\big(\sqrt{A_j+4}-\sqrt{A_j}\big)=\pm\frac{2}{\sqrt{A_j+4}+\sqrt{A_j}},
\end{equation*}
with the plus sign exactly when $\hat y_j\ne w_j$. The remainder is half the change of the norm when coordinate $j$ flips; it is at most $1$ in absolute value and of the order of $1/\sqrt{A_j}$ when many other coordinates differ, whereas the norm itself is of the order of $\sqrt{A_j}$. This subtraction is the classical control variate of Monte Carlo, a quantity of known mean subtracted from an estimator with which it is correlated, so it does not bias the estimator: $b_j$ is a function of the other coordinates of $\hat y$ and of $w$, which are independent of $\eta_j$, so $\E[b_j\eta_j]=0$ and $\E[b_j(\eta_j^2-1)]=0$, and the means of $\widehat G_{ij}$ and $\widehat S_{ij}$ are unchanged. In the terminology of \citet{williams1992simple}, $b_j$ is a baseline, chosen separately for each coordinate. Both norms in each cost are treated in this way: with $b_j(\hat y,w)$ the midpoint above, coordinate $j$ of the first draw gets the centered cost
\begin{equation*}
	c_{ij}=\big(\|Y_i-\hat Y_i\|-b_j(\hat Y_i,Y_i)\big)-\big(\|\hat Y_i-\hat Y'_i\|-b_j(\hat Y_i,\hat Y'_i)\big),
\end{equation*}
and $c'_{ij}$ is defined in the same way with $\hat Y'_i$ as the draw. The costs $c_i$ and $c'_i$ thus become vectors that multiply the weights coordinate by coordinate, $\widehat G_{ij}=\frac12c_{ij}\eta_{ij}/s_{ij}$, $\widehat G'_{ij}=\frac12c'_{ij}\eta'_{ij}/s_{ij}$ and $\widehat S_{ij}=\frac12\big(c_{ij}(\eta_{ij}^2-1)+c'_{ij}((\eta'_{ij})^2-1)\big)/s_{ij}$, with the plain cost kept on the coordinates that are not binary. Each centered cost is at most $2$ in absolute value, for any number of labels. Computing the midpoints takes one pass over the coordinates per norm. On the species data of Section~\ref{sec:jsdm}, with $242$ labels, the control variate reduces the variance of the gradient estimator by two orders of magnitude (Appendix~\ref{app:jsdm}); we recommend it whenever the outcome has many binary coordinates.

\subsection{Other kernel scores}\label{sec:kernel}

GEM can be fitted with any kernel score \citep{gneiting2007strictly} in place of the energy score. For a symmetric function $\kappa$ on $\Y\times\Y$ that vanishes on the diagonal, the kernel score is
\begin{equation*}
	S_\kappa(P,y)=\E_P\big[\kappa(\hat Y,y)\big]-\frac12\E_P\big[\kappa(\hat Y,\hat Y')\big],\qquad \hat Y,\hat Y'\sim P\text{ independent},
\end{equation*}
of which the energy score of Section~\ref{sec:background} is the case $\kappa(v,w)=\|v-w\|^\beta$. Other examples are $\kappa(v,w)=1-\exp\{-\|v-w\|^2/(2\gamma^2)\}$ with a bandwidth $\gamma>0$, which corresponds to the Gaussian kernel; and, on a finite outcome space, $\kappa(v,w)=\mathbf 1\{v\ne w\}$, for which $S_\kappa$ is half the Brier score of $Y$ taken as one categorical variable. If the observation is drawn from $Q$, the expected score of $P$ exceeds that of $Q$ by $-\frac12\iint\kappa\,d(P-Q)\,d(P-Q)$, so $S_\kappa$ is proper when $\kappa$ is conditionally negative definite, $\sum_{v,w}a_va_w\kappa(v,w)\le0$ whenever $\sum_va_v=0$, and, on a finite outcome space, strictly proper when the inequality is strict for every such $a\ne0$, as it is for all the examples above. The excess score is then half the squared maximum mean discrepancy between $P$ and $Q$ for the positive definite kernel $\kappa(v,v_0)+\kappa(w,v_0)-\kappa(v,w)$, with $v_0$ fixed \citep{gretton2012kernel,sejdinovic2013equivalence}. Replacing the norm by $\kappa$ in \eqref{eq:gem_objective} and in the costs of Algorithm~\ref{alg:gem} therefore fits the model by minimizing the average over $X$ of the squared discrepancy between the conditional laws of the model and of the data. Proposition~\ref{prop:grad} and Algorithm~\ref{alg:gem} carry over unchanged when $\kappa$ grows at most linearly, as it does on a discrete outcome, because they use the norm only through its values at the draws. The control variate carries over, with the two values of $\kappa$ in place of those of the norm.

\section{Universality theory}\label{sec:universality}

In this section, we establish universality results, in the sense of exact distributional representation, for continuous, discrete, and mixed-type outcomes. The results below concern the special case $\sigma(x)\equiv0$, in which the model reduces to $Y=h(g(X,\varepsilon))$.
This case is not itself part of model~\eqref{eq:gem_model}, which requires $\sigma>0$, but is its limit. For each continuous generator $g$ constructed in Sections~\ref{sec:continuous}--\ref{sec:mixed_type}, the law of $g(x,\varepsilon)$ has a Lebesgue density for every $x$, and every link of Table~\ref{tab:links} is continuous outside a Lebesgue-null set, so $h(g(X,\varepsilon)+\sigma_0\eta)\to h(g(X,\varepsilon))$ almost surely as the constant scale $\sigma_0\downarrow0$. Hence $(X,h(g(X,\varepsilon)+\sigma_0\eta))$ converges to $(X,Y)$ in distribution, and in total variation when $Y$ is discrete. The perturbation $\sigma(X)\odot\eta$ is needed for fitting (Section~\ref{sec:gradient}), not for representation.

The link $h$ may be discontinuous, to accommodate different types of outcomes, whereas we seek a continuous generator $g$ to facilitate neural network approximation. We work with the simplest qualitative regularity conditions and leave quantitative guarantees based on norms of partial derivatives~\citep{gyorfi2002distribution} to future work.

We first give the results for a measurable map $g$ (Section~\ref{sec:measurable}) and then, in the rest of this section, for a continuous one. The proofs are collected in Appendix~\ref{app:proofs}.

\subsection{Measurable map \texorpdfstring{$g$}{g}}
\label{sec:measurable}

We start with the following lemma \citep[Proposition~8.20; see also Theorem~8.17]{Kallenberg1997}, which gives a measurable map $g$ in great generality. It applies to random elements in arbitrary standard Borel spaces, beyond the Euclidean outcomes considered in this paper, and in particular covers continuous, discrete, and mixed-type outcomes without requiring a conditional density. We state it with standard Gaussian noise for consistency with our model. For measurable representations, any atomless real-valued noise can be used instead, since its distribution function transforms it into a uniform random variable.

\begin{Lemma}[Noise outsourcing, Gaussian form]
\label{lem:noise_outsourcing}
	Let $\X$ and $\Y$ be standard Borel spaces, and let $X$ and $Y$ be random elements in $\X$ and $\Y$, respectively. There exist a measurable function $g:\X\times\rb\to\Y$ and a random variable $\varepsilon\sim\mathcal N(0,1)$ such that $\varepsilon$ is independent of $X$ and
	\[
		(X,Y)\overset{d}{=}(X,g(X,\varepsilon)).
	\]
\end{Lemma}

If $h$ is measurable and has a measurable right inverse $\bar h:\Y\to\rb^k$, that is, $h\circ\bar h=\mathrm{Id}_{\Y}$, we can apply the lemma above to $(X,\bar h(Y))$. Then the distribution of $h(g(x,\varepsilon))$ is exactly the conditional distribution of $Y$ given $X=x$ for $P_X$-almost every $x$, where $P_X$ denotes the distribution of $X$.

\begin{Proposition}[Noise outsourcing through the link $h$, Gaussian form]
\label{prop:noise_outsourcing_h}
	Let $\X$ and $\Y$ be standard Borel spaces, and let $X$ and $Y$ be random elements in $\X$ and $\Y$, respectively. Let $h:\rb^k\to\Y$ be measurable.

	\smallskip\noindent (i) If there exists a measurable map $\bar h:\Y\to\rb^k$ such that $h(\bar h(Y))=Y$ almost surely, then there exist a random variable $\varepsilon\sim\mathcal N(0,1)$, independent of $X$, and a measurable map $g:\X\times\rb\to\rb^k$ such that
	\[
		(X,Y)\overset{d}{=}(X,h(g(X,\varepsilon))).
	\]
	\smallskip\noindent (ii) In particular, the conclusion of (i) holds whenever $h$ is a surjective Borel map.
\end{Proposition}

Even a univariate $\varepsilon$ is thus enough to generate a high-dimensional random vector $Y$, but such a $g$ can be highly discontinuous. In practice, for tractable optimization and computation, we may use $\varepsilon\sim\mathcal N(0,I_m)$, where $m\ge1$ is a modeling choice, for example $m=k$. Continuity of $g$, which the measurable representation does not guarantee, facilitates uniform neural network approximation on compact subsets of Euclidean input spaces. The rest of this section gives conditions under which $g$ can be chosen continuous. Table~\ref{tab:links} lists, for each data type, the link and the theorem that gives its representation: continuous outcomes are treated in Section~\ref{sec:continuous}, and discrete and mixed-type outcomes in Sections~\ref{sec:discrete} and~\ref{sec:mixed_type}, respectively.

\subsection{Continuous map \texorpdfstring{$g$}{g} for continuous outcomes}
\label{sec:continuous}

Let $(\X,d_{\X})$ be a metric space equipped with a $\sigma$-finite Borel measure $\nu$, and let $\Y=\mathbb{R}^k$. Suppose that the distribution of $(X,Y)$ is absolutely continuous with respect to the product measure $\nu\otimes\lambda_k$, where $\lambda_k$ denotes Lebesgue measure on $\mathbb{R}^k$. Let $p(x,y)$ denote the corresponding joint density, so that, for all Borel sets $A\subseteq\X$ and $B\subseteq\mathbb{R}^k$,
\[
	\mathbb{P}(X\in A,Y\in B)=\int_A\int_B p(x,y)\,dy\,\nu(dx),
\]
where $dy=dy_1\cdots dy_k$ denotes integration with respect to Lebesgue measure on $\mathbb{R}^k$.

All continuity statements refer to the product topology induced by $d_{\X}$ on $\X$ and the Euclidean metric on $\mathbb{R}^k$; in particular, continuity of $p$ means that
\[
	d_{\X}(x_n,x)\to0\quad\text{and}\quad\|y_n-y\|_2\to0\quad\Longrightarrow\quad p(x_n,y_n)\to p(x,y).
\]

For $0\leq j\leq k$, we define the marginal density
\[
	p_{\leq j}(x,y_{\leq j})=\int_{\mathbb{R}^{k-j}}p(x,y_1,\ldots,y_k)\,dy_{j+1}\cdots dy_k,
\]
so that $p_{\leq0}=p_X$ is the density of $X$ and $p_{\leq k}=p$; we also write $p_{<j}=p_{\leq j-1}$. The following condition on $p$ ensures that $g$ can be chosen continuous.
\begin{Condition}[Positive and continuous density]
\label{cond:regular}
	The density $p:\X\times\mathbb{R}^k\to(0,\infty)$ satisfies:
	\begin{enumerate}
		\item The joint density $p$ is strictly positive and continuous.
		\item For every $0\leq j\leq k$, the marginal density $p_{\leq j}(x,y_{\leq j})$ is finite and continuous.
	\end{enumerate}
\end{Condition}

Define the sequential conditional density and distribution function by
\[
	f_j(s\mid x,y_{<j})=\frac{p_{\leq j}(x,y_{<j},s)}{p_{<j}(x,y_{<j})},\qquad F_j(t\mid x,y_{<j})=\int_{-\infty}^{t}f_j(s\mid x,y_{<j})\,ds,
\]
and the corresponding conditional quantile function by
\[
	Q_j(\tau\mid x,y_{<j})=\inf\left\{t\in\mathbb{R}:F_j(t\mid x,y_{<j})\geq\tau\right\},\qquad 0<\tau<1.
\]

The following lemma shows that the conditional quantile function $Q_j(\tau\mid x,y_{<j})$ is continuous in all of its arguments.
\begin{Lemma}
	Suppose that Condition~\ref{cond:regular} holds. Then, for $1\leq j\leq k$, the distribution function $F_j(t\mid x,y_{<j})$ is jointly continuous in $(t,x,y_{<j})$ and strictly increasing in $t$. Moreover, $Q_j(\tau\mid x,y_{<j})$ is jointly continuous in $(\tau,x,y_{<j})$ on $(0,1)\times\X\times\mathbb{R}^{j-1}$.
	\label{lem:continuity}
\end{Lemma}
The above lemma and the sequential quantile construction underlie the following representation of Euclidean outcomes with a continuous map $g$. For a one-dimensional outcome, the conditional quantile function $Q_1(\cdot\mid x)$ is itself a continuous generator driven by uniform noise on $(0,1)$. For a multivariate outcome, we use the inverse Rosenblatt transformation \citep{Rosenblatt1952}, which generates each coordinate in turn from its conditional quantile function given $X$ and the coordinates generated before it.

\begin{Theorem}[Euclidean representation]\label{thm:continuous}
	Suppose that Condition~\ref{cond:regular} holds. Then, for $\varepsilon\sim\mathcal N(0,I_k)$ independent of $X$, there exists a continuous map $g:\X\times\mathbb{R}^k\to\mathbb{R}^k$ such that
	\[
		\left(X,g(X,\varepsilon)\right)\overset{d}{=}(X,Y).
	\]
\end{Theorem}

Equivalently, one may use a vector $U$ of independent $\operatorname{Unif}(0,1)$ variables and define the generator on $\X\times(0,1)^k$.

\subsection{Continuous map \texorpdfstring{$g$}{g} for discrete outcomes}
\label{sec:discrete}

We now turn to discrete outcomes, starting with a latent-lift criterion for a measurable link $h$ and a continuous latent random vector $Z\in\mathbb{R}^k$.
\begin{Definition}[Regular positive latent lift]\label{def:latent_lift}
	Let $h:\mathbb{R}^k\to\Y$ be measurable. A \emph{regular positive latent lift} of $(X,Y)$ through $h$ is a density
	\[
		a:\X\times\mathbb{R}^k\to(0,\infty)
	\]
	with respect to $\nu(dx)\,dz$ that satisfies Condition~\ref{cond:regular}, with $z$ in place of $y$, and such that, if $(X,Z)$ has joint density $a$, then
	\begin{equation}
		(X,h(Z))\overset{d}{=}(X,Y).
		\label{eq:equal_in_dist}
	\end{equation}
\end{Definition}
The main requirement of the above definition is that there exists a continuous latent variable $Z$ such that $(X,Z)$ has a regular positive density.

The following theorem shows that a regular positive latent lift yields a continuous generator $g$, whatever the type of $Y$.
\begin{Theorem}[Latent-lift criterion]\label{thm:lift_criterion}
	If $(X,Y)$ admits a regular positive latent lift through $h$, then, for $\varepsilon\sim\mathcal N(0,I_k)$ independent of $X$, there exists a continuous map $g:\X\times\mathbb{R}^k\to\mathbb{R}^k$ such that
	\[
		\left(X,h(g(X,\varepsilon))\right)\overset{d}{=}(X,Y).
	\]
\end{Theorem}

We next obtain a continuous $g$ for discrete outcomes by constructing a latent density and applying Theorem~\ref{thm:lift_criterion}. Suppose that $\Y$ is finite, and let $p_y(x)$ denote the joint subdensity of $X$ and the event $\{Y=y\}$ with respect to $\nu$:
\[
	\mathbb{P}(X\in B,Y=y)=\int_Bp_y(x)\,d\nu(x),\qquad\sum_{y\in\Y}\int_{\X}p_y(x)\,d\nu(x)=1.
\]
With this notation, the pushforward requirement \eqref{eq:equal_in_dist} is equivalent to
\begin{equation*}
	\int_{A_y}a(x,z)\,dz=p_y(x),\qquad A_y=h^{-1}(\{y\}),\quad\text{for $\nu$-almost every $x$ and every $y\in\Y$}.
\end{equation*}

\begin{figure}[tp!]
	\centering
	\includegraphics[width=\textwidth]{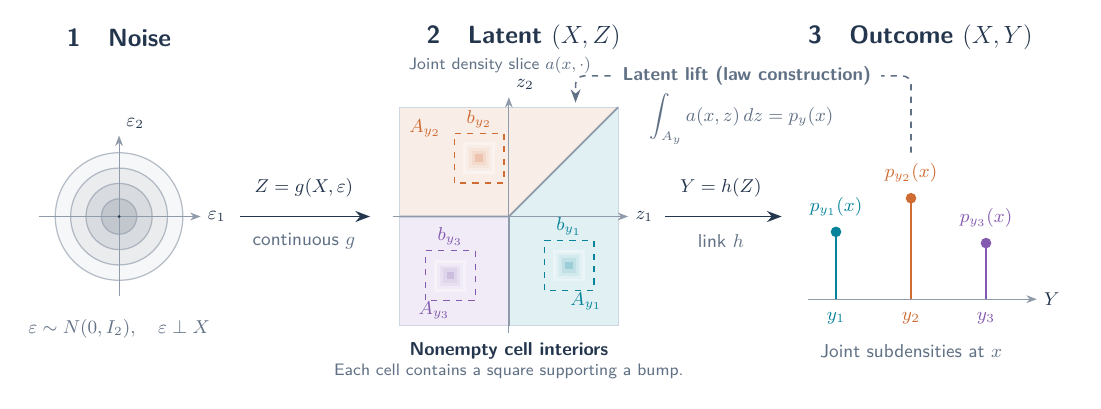}

    \vspace*{-.35cm}
	\caption{Latent-lift construction for discrete outcomes. Independent Gaussian noise $\varepsilon$ is transformed by the continuous generator into $Z=g(X,\varepsilon)$, which the link $h$ maps to $Y=h(Z)$. For a fixed $x$, the middle and right panels show the slice $a(x,\cdot)$ of the latent density and the joint class subdensities $p_y(x)$. Each cell $A_y=h^{-1}(\{y\})$ has nonempty interior and supports a bump density $b_y$. Together with a positive Gaussian background (not shown), these bumps form a latent density satisfying $\int_{A_y}a(x,z)\,dz=p_y(x)$, as indicated by the dashed arrow. The conditional density of $Z$ given $X=x$ and the conditional probabilities of $Y$ given $X=x$ are obtained by dividing $a(x,z)$ and $p_y(x)$ by $p_X(x)=\sum_y p_y(x)>0$.}
	\label{fig:latent_lift}
\end{figure}

The representation for discrete outcomes rests on the following condition.
\begin{Condition}[Interior cells and positive class subdensities]
\label{cond:interior}
	The link $h:\mathbb{R}^k\to\Y$ is measurable, and for every $y\in\Y$:
	\begin{enumerate}
		\item the cell $A_y=h^{-1}(\{y\})$ has nonempty interior;
		\item the subdensity $p_y$ is finite, continuous, and strictly positive on $\X$.
	\end{enumerate}
\end{Condition}

Under this condition, we construct a regular positive latent lift $a$ through the given link $h$, and Theorem~\ref{thm:lift_criterion} then yields the following result. Figure~\ref{fig:latent_lift} illustrates the construction.

\begin{Theorem}[Discrete representation]\label{thm:discrete_continuous}
	Suppose that Condition~\ref{cond:interior} holds. Then, for $\varepsilon\sim\mathcal N(0,I_k)$ independent of $X$, there exists a continuous map $g:\X\times\mathbb{R}^k\to\mathbb{R}^k$ such that
	\[
		\left(X,h\left(g\left(X,\varepsilon\right)\right)\right)\overset{d}{=}(X,Y).
	\]
\end{Theorem}

\subsection{Continuous map \texorpdfstring{$g$}{g} for mixed-type outcomes}
\label{sec:mixed_type}

We next consider mixed-type outcomes $Y=(W,D)\in\mathbb{R}^{k_1}\times\mathcal D$, where $W$ is a continuous random vector and $D$ takes values in a finite set $\mathcal D\subseteq\mathbb{R}^{k_2}$, with $k=k_1+k_2$. For each $d\in\mathcal D$, let $p_d(x,w)$ denote the joint subdensity of $(X,W)$ and the event $\{D=d\}$ with respect to $\nu(dx)\,dw$, so that
\[
	\mathbb P(X\in B,W\in C,D=d)=\int_B\int_C p_d(x,w)\,dw\,\nu(dx).
\]
Write $f(x,w)=\sum_{d\in\mathcal D}p_d(x,w)$ for the joint density of $(X,W)$. Let $h_D:\mathbb{R}^{k_2}\to\mathcal D$ be a measurable link for the discrete block, and define
\[
	h:\mathbb{R}^{k}\to\mathbb{R}^{k_1}\times\mathcal D,\qquad h(w,z)=(w,h_D(z)).
\]
Thus, $h$ retains the continuous block $w$ and maps the latent block $z$ into $\mathcal D$.

The idea is to first generate the continuous block $W$ given $X$, and then $D$ given $X$ and the generated $W$; this preserves the dependence between the two blocks.

\begin{Condition}[Mixed-type outcomes]
\label{cond:mixed}
	The joint density $f$ of $(X,W)$, the subdensities $p_d$, and the link $h_D$ satisfy:
	\begin{enumerate}
		\item The joint density $f(x,w)$ satisfies Condition~\ref{cond:regular}, with $w$ and $k_1$ in place of $y$ and $k$.
		\item For every $d\in\mathcal D$, the subdensity $p_d$ is finite, continuous, and strictly positive on $\X\times\mathbb{R}^{k_1}$.
		\item For every $d\in\mathcal D$, the cell $h_D^{-1}(\{d\})$ has nonempty interior.
	\end{enumerate}
\end{Condition}

\begin{Theorem}[Mixed-type representation]
\label{thm:mixed}
	Suppose that Condition~\ref{cond:mixed} holds. Then, for $\varepsilon\sim\mathcal N(0,I_{k})$ independent of $X$, there exists a continuous map $g:\X\times\mathbb{R}^{k}\to\mathbb{R}^{k}$ such that
	\[
		\left(X,h(g(X,\varepsilon))\right)\overset d=(X,Y).
	\]
\end{Theorem}

\section{Simulation study}\label{sec:simulation}

We study GEM on simulated data for three outcome types: multi-label, mixed-type and ranking. In each setting we compare it with type-specific baselines and with an oracle that samples from the true conditional distribution, and we assess the fit of the joint distribution as well as of its marginals. The metric curves are means over $50$ replications, with bands of one standard error; Appendix~\ref{app:simulation} gives the settings of all methods and the computation of the metrics.

\subsection{Multi-label outcome}

We first consider predicting a binary label vector $Y\in\{-1,+1\}^{10}$ from covariates $X\sim\mathcal N(0,I_5)$. The data follow a probit model $Y_j=\sign(Z_j)$ whose latent scores couple a per-label signal with a component shared across all labels, in both mean and noise,
\begin{equation*}
Z_j=\sqrt{1-\rho}\,f_j(X)+\sqrt{\rho}\,f_0(X)+\sqrt{1-\rho}\,\zeta_j+\sqrt{\rho}\,u,\qquad \zeta_1,\dots,\zeta_{10},u\overset{\text{iid}}{\sim}\mathcal N(0,1).
\end{equation*}
The parameter $\rho\in[0,1]$ sets the dependence between the labels and leaves the noise level of each label unchanged: for every $\rho$ one has $\Var(Z_j\mid X)=1$ and the closed-form marginal $\P(Y_j=1\mid X=x)=\Phi(f_j^\rho(x))$, where $\Phi$ is the standard normal distribution function and $f_j^\rho=\sqrt{1-\rho}\,f_j+\sqrt{\rho}\,f_0$, while $\corr(Z_j,Z_l\mid X)=\rho$. The labels are thus conditionally independent at $\rho=0$, and at $\rho=1$ the latent scores coincide, $Z_1=\dots=Z_{10}$, so that $Y$ is either $\mathbf 1$ or $-\mathbf 1$. The per-label component is the sum of a linear and a sine term, $f_j(x)=\beta_j^\top x+\sin(2\gamma_j^\top x)$, and in the shared component the sine term $\sin(2\gamma_0^\top x)$ is replaced by a centered quadratic ridge term of the same variance, $f_0(x)=\beta_0^\top x+c\big((\gamma_0^\top x)^2-\|\gamma_0\|^2\big)$, so that only its shape changes, not its strength. The $\beta_j$'s and $\gamma_j$'s are randomly generated coefficients. The shared component, which carries the dependence, is thus nonlinear in a way that no linear model of the marginals can represent. We sweep $\rho\in\{0.1,0.2,\dots,1.0\}$ at sample sizes $n\in\{100,1000,10000\}$. We show the results at $n=10{,}000$; Appendix~\ref{app:multilabel} gives those at $n=100$ and $n=1{,}000$, where GEM leads on the joint metrics once the dependence is moderate.

We compare GEM with the sign link against two marginal-only baselines. Both baselines treat the labels as conditionally independent given $x$ but differ in their marginal model: {linear logistic} regression fits one generalized linear model per label, whereas {NN logistic} fits the conditional probabilities using a neural network model and the cross-entropy loss. Because every $f_j^\rho$ is nonlinear, the linear model is misspecified even on the marginals while the neural network one is not. We report the joint energy score on the full $\pm1$ vector to assess the overall distributional fit, mean per-label CRPS and the per-label cross-entropy as marginal calibration checks, the co-occurrence error (the mean of $|\hat{\P}(Y_j=Y_l\mid X=x)-\P(Y_j=Y_l\mid X=x)|$ over label pairs $j<l$ and test points $x$) as a joint-dependence check, and the conditional distribution of the number of positive labels at a given $X$ value.

\begin{figure}
	\centering
	\begin{tabular}{@{}c@{}c@{}c@{}c@{}}
		\includegraphics[width=0.24\textwidth]{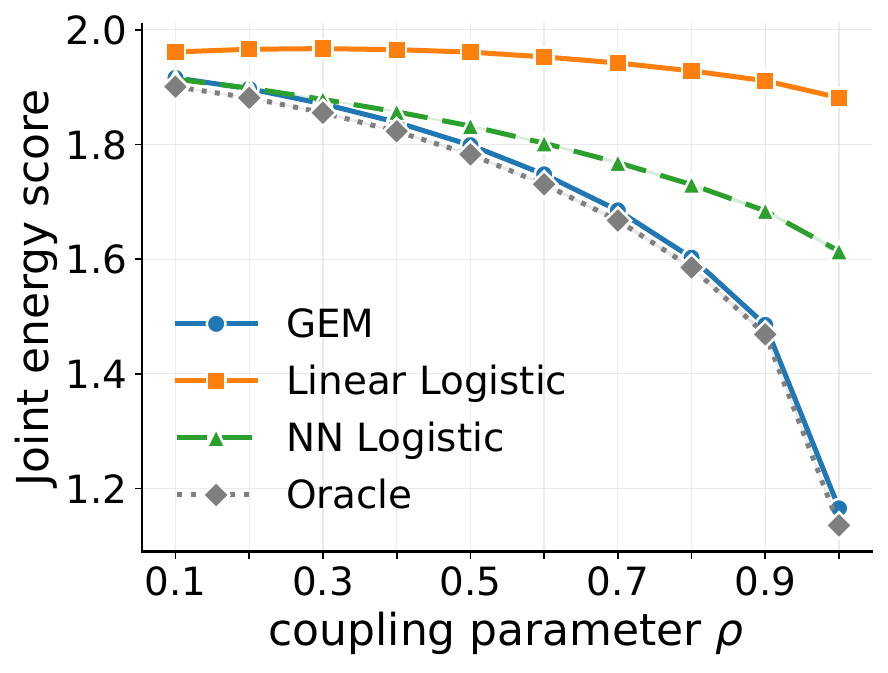} &
		\includegraphics[width=0.24\textwidth]{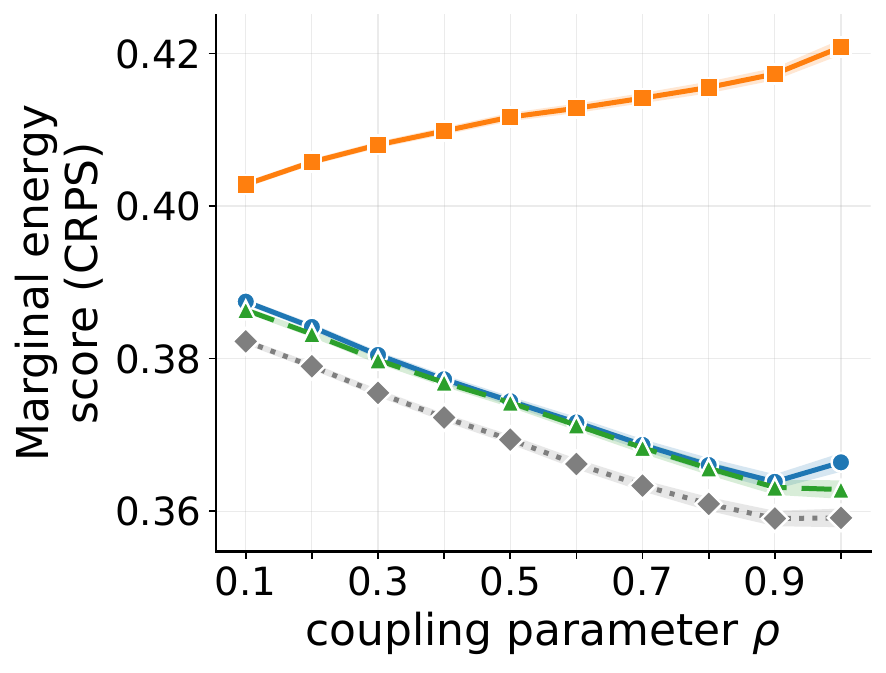} &
		\includegraphics[width=0.24\textwidth]{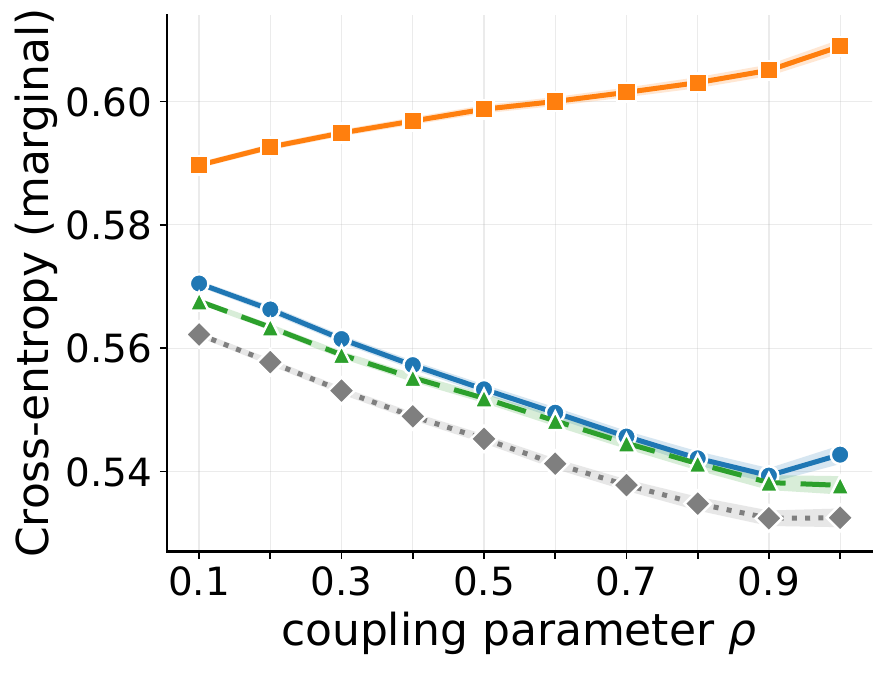} &
		\includegraphics[width=0.24\textwidth]{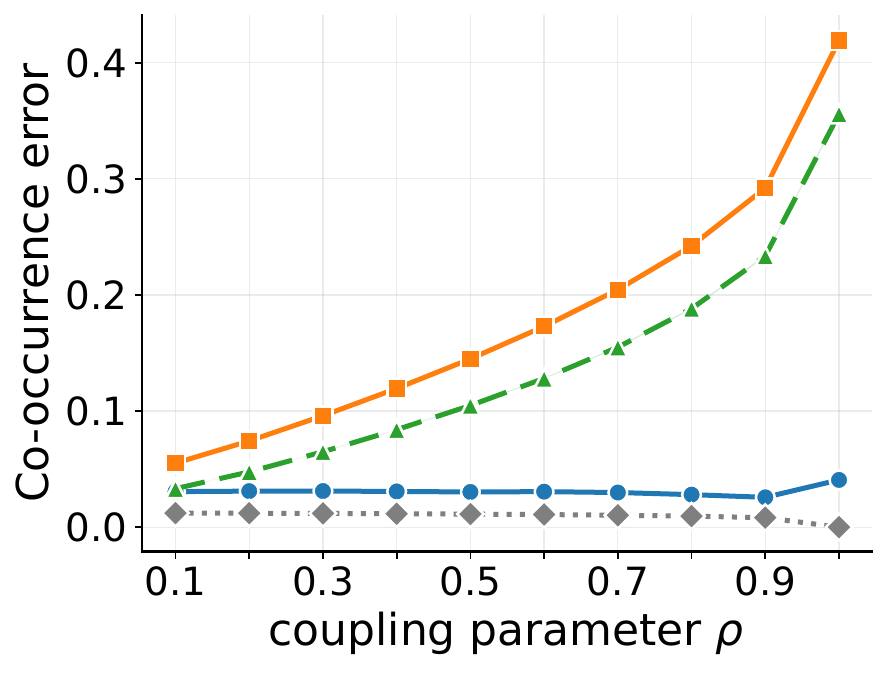} 
	\end{tabular}
	\caption{Multi-label prediction: joint energy score, mean per-label CRPS, per-label cross-entropy and co-occurrence error (left to right) against $\rho$. Lower is better.}
	\label{fig:sim62_metrics}
\end{figure}

\begin{figure}
	\centering
	\includegraphics[width=0.95\textwidth]{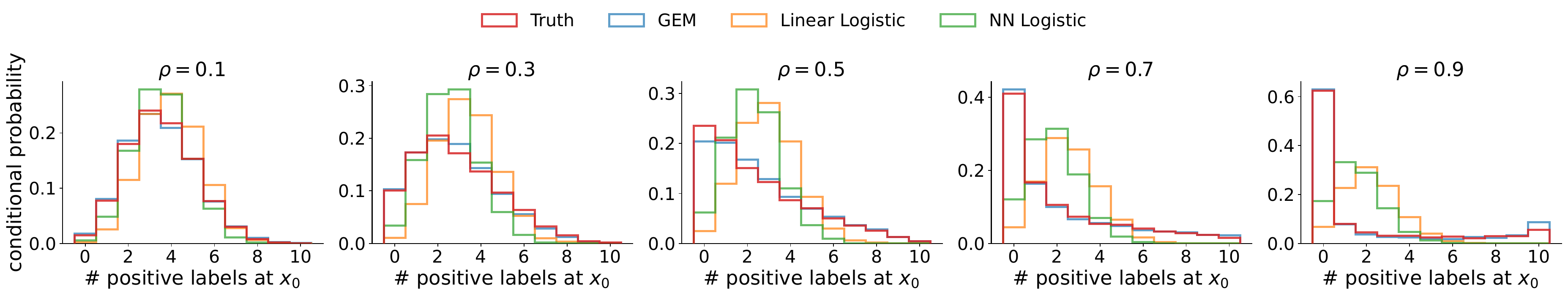}
	\caption{Conditional distribution of the number of positive labels at a fixed value $x_0$ of $X$, for five values of $\rho$, from $4{,}000$ draws of each method and of the true mechanism.}
	\label{fig:sim62_counts}
\end{figure}

Figure~\ref{fig:sim62_metrics} shows these metrics. On the marginal metrics (cross-entropy and marginal CRPS) GEM lies slightly above the flexible NN-logistic baseline, both just above the oracle, whereas linear logistic is uniformly worse and deteriorates as $\rho$ grows, since a linear model captures none of the quadratic shared mean, which dominates the marginals as $\rho\to1$. The ordering changes on the joint metrics, where GEM stays close to the oracle for every $\rho$ and both baselines fall increasingly behind as $\rho$ grows, linear logistic more than NN logistic. Their joint energy score separates from the oracle's, and their co-occurrence error climbs steeply with $\rho$, while GEM's stays small and nearly flat, slightly above the oracle's. The neural network baseline gets the marginals right but, having no joint model, moves away from the oracle when the labels become dependent. Figure~\ref{fig:sim62_counts} shows the effect on the conditional law of the number of positive labels, which spreads out as $\rho$ increases, with a growing peak at zero. GEM follows it closely throughout, while both independence-based baselines remain under-dispersed and miss the peak.

\subsection{Mixed-type outcome}\label{sec:mixed}

We next examine how well GEM captures a mixed-type outcome with dependence across types. We draw $Y=(Y^{\mathrm{cont}},Y^{\mathrm{bin}},Y^{\mathrm{cat}})$ from a shared latent Gaussian vector,
\begin{equation*}
Z=m(BX)+\Sigma^{1/2}\xi,\quad \xi\sim\mathcal N(0,I_p),\qquad Y=\Big(Z_{1:k_c},\ \sign(Z_{k_c+1:k_c+k_b}),\ \argmax_{i\in\mathrm{block}} Z_i\Big),
\end{equation*}
so that a $k_c$-dimensional continuous block, $k_b$ binary labels and a $K$-way categorical label are all read off the same latent draw $Z\in\rb^p$, with $p=k_c+k_b+K$. The covariance $\Sigma$ is a fixed positive definite matrix with substantial correlations within and across the three types. We report the discrete-heavy configuration $k_c=1$, $k_b=20$, $K=10$ under a nonlinear mean $m(u)=u+\sin(8u)$, with $n\in\{100,1000,10000\}$.

GEM uses the identity, sign and argmax links on the three blocks. The natural competitor is a {stacked-marginals} pipeline that models each block with its own correctly specified estimator: engression for the continuous block, an elementwise binary network per label, and a softmax network for the categorical label trained with cross-entropy loss. This approach targets each per-type marginal directly but captures no dependence across types, nor among the binary labels. We evaluate the joint energy score on the full mixed vector, the per-type marginal metrics (CRPS on the continuous block and cross-entropy on the binary and categorical blocks), and the cross-type dependence directly, through the errors of the correlation ratios that link the continuous residual to the binary labels and to the class (Appendix~\ref{app:mixed}).

\begin{figure}
\centering
\begin{tabular}{ccc}
	\includegraphics[width=0.3\textwidth]{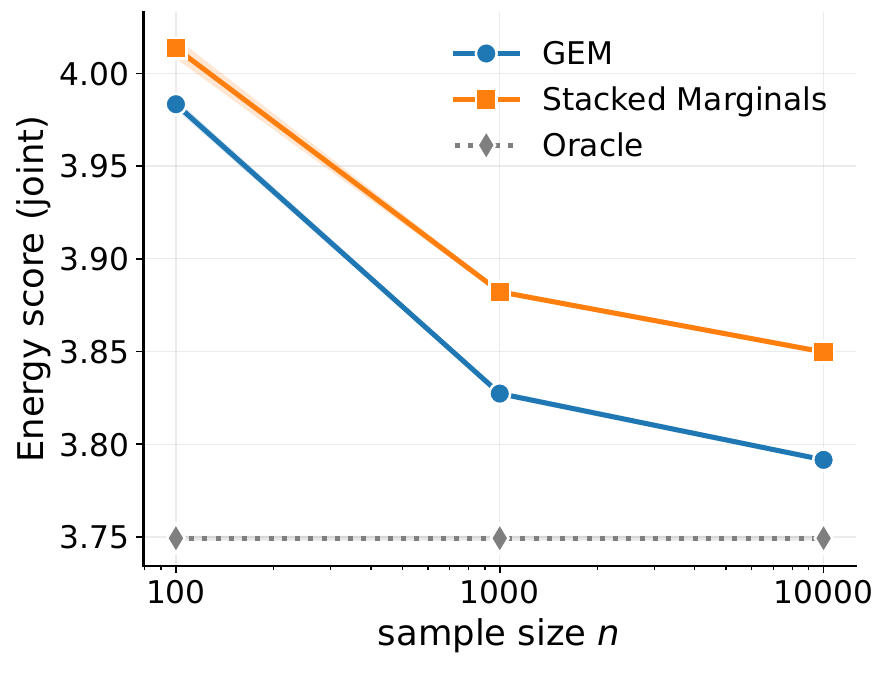} &
	\includegraphics[width=0.3\textwidth]{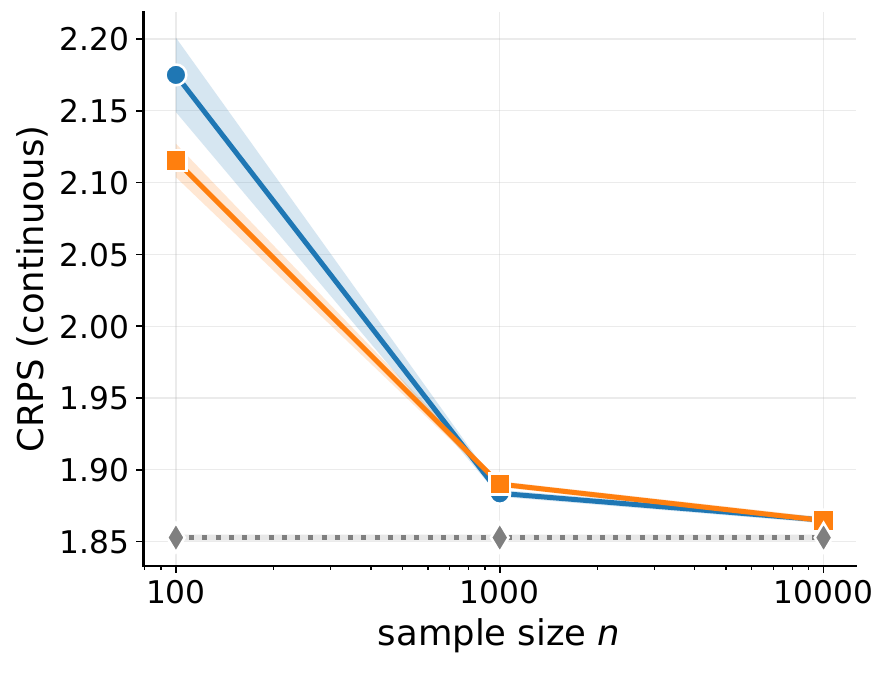} &
	\includegraphics[width=0.3\textwidth]{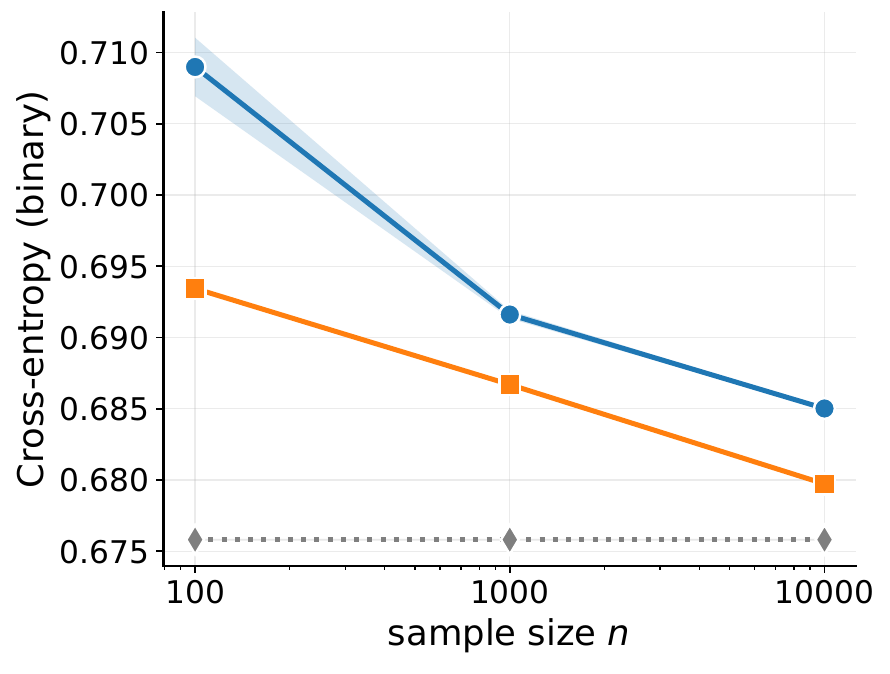}\\
	\includegraphics[width=0.3\textwidth]{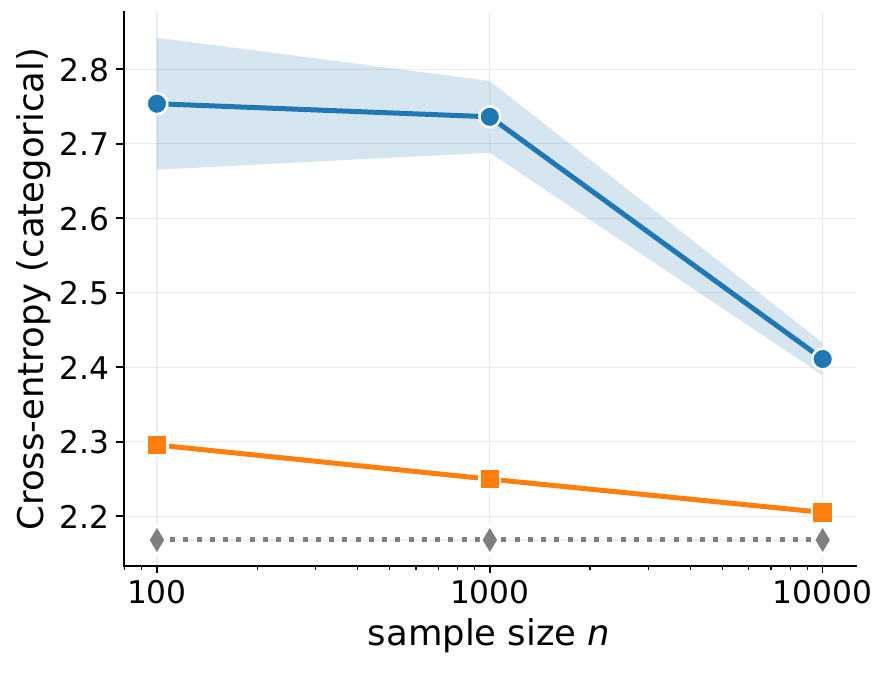} &
	\includegraphics[width=0.3\textwidth]{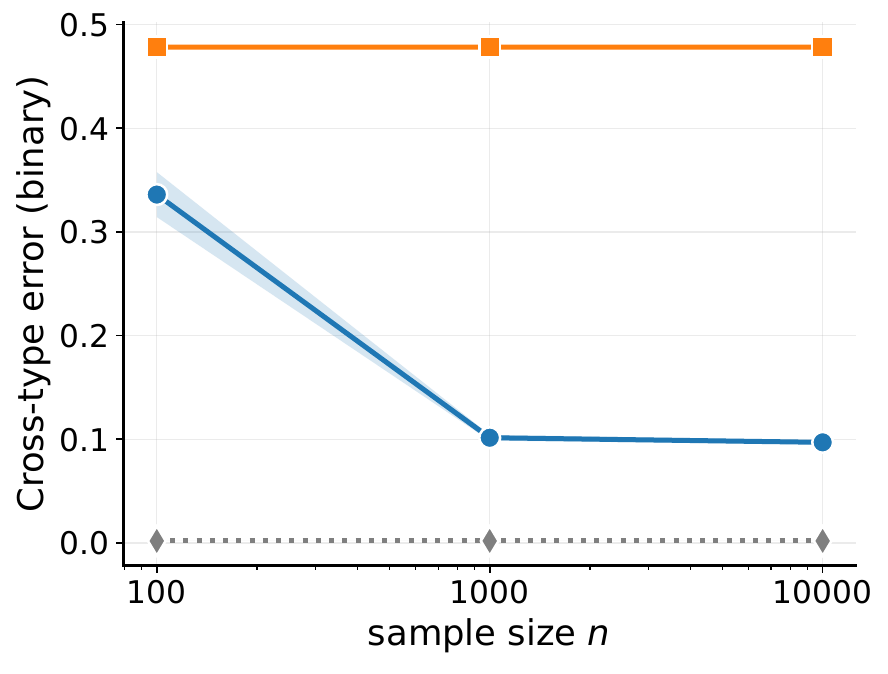} &
	\includegraphics[width=0.3\textwidth]{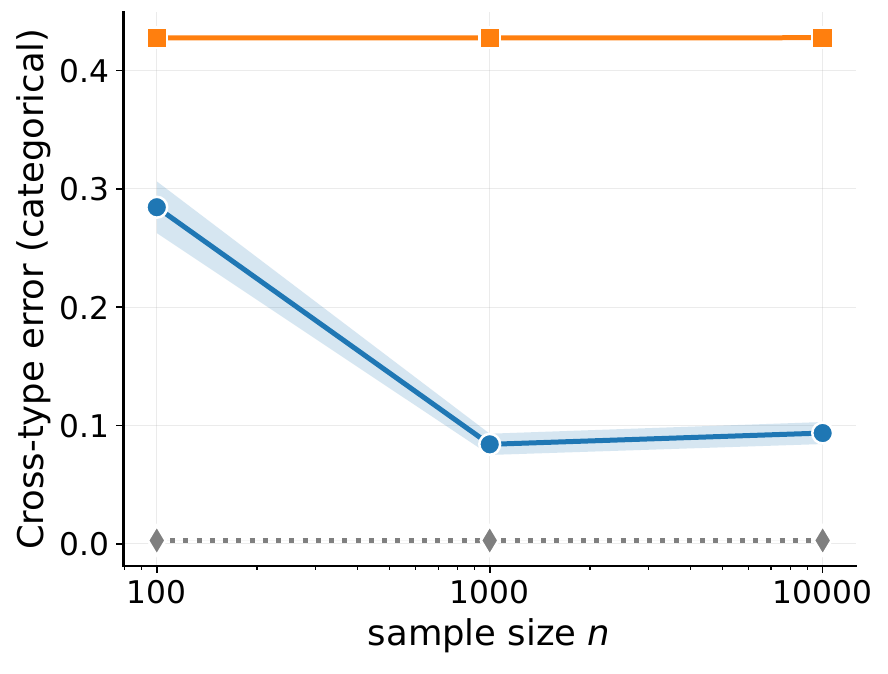}
\end{tabular}
\caption{Mixed-type simulation: joint energy score, continuous CRPS, binary and categorical cross-entropy, and the two cross-type correlation errors, against the sample size $n$. Lower is better.}\label{fig:sim63_mixed}
\end{figure}

Figure~\ref{fig:sim63_mixed} reports the metrics against sample size. On the per-type marginals GEM is close to the stacked pipeline from $n=1{,}000$ on: the continuous CRPS and binary cross-entropy of the two methods approach the oracle together, GEM's binary cross-entropy being slightly higher at every $n$. The exception is the categorical block, where GEM's cross-entropy is higher than that of the stacked pipeline, which is trained on this loss, though the gap is smaller at $n=10{,}000$ than at the two smaller sample sizes. The stacked pipeline draws the three blocks independently, so it cannot represent dependence across types: its cross-type errors equal the size of the true correlation ratios at every $n$. GEM's errors are lower already at $n=100$, fall to a fraction of the pipeline's by $n=1{,}000$ and change little after that, though they stay well above the oracle's. GEM has the lower joint energy score at every sample size and is closer to the oracle, even though its marginals are comparable or marginally less sharp. The joint score rewards the dependence across types and among the labels, which the per-type metrics do not reflect.

\subsection{Ranking outcome}

The last setting is a ranking outcome, an ordering of $K$ items, which is a categorical variable with $K!$ levels and a great deal of structure. The standard regression model for ranked responses is the Plackett--Luce model \citep{luce1959individual,plackett1975analysis} with item scores linear in the covariates, also known as the rank-ordered logit \citep{beggs1981assessing,hausman1987specifying}. It is the law of the ordering induced by random utilities with independent Gumbel noise, and it has the property of independence of irrelevant alternatives: the relative order of two items does not depend on which other items are present. More flexible versions replace the linear scores by, for example, splines but keep the Plackett--Luce form. GEM with the ranking link assumes neither linear scores nor the Plackett--Luce form.

The data come from a random-utility model with $K=4$ items and $X\sim\mathcal N(0,I_5)$,
\begin{equation*}
	U_j=f_j(X)+\xi_j,\quad j=0,\dots,3,\qquad Y=\text{the ordering of the four items by }U,
\end{equation*}
so that $Y$ takes one of $4!=24$ values; we write $1\succ0\succ2\succ3$ for the ordering that places item $1$ first, then items $0$, $2$ and $3$. The item means are nonlinear, $f_j(x)=\beta_j^\top x+\sin(2\gamma_j^\top x)$, with randomly generated coefficients, held fixed across replications, under which the standard deviation of $f_j(X)$ lies between $0.77$ and $0.96$. Items $0$ and $1$ are twins, which differ only by a constant offset in the mean that is small compared with the noise, $f_1=f_0+\delta$ with $\delta=0.1$. The noise is $\xi_j=c\,G_j$ with $G_j$ standard Gumbel and $c=\sqrt6/\pi$ so that they all have unit variance.

We consider two settings that differ only in how the twins' noise is drawn. In the \emph{independent} setting, $\xi_0,\dots,\xi_3$ are drawn independently; the conditional law of $Y$ is then exactly Plackett--Luce with scores $f_j(x)/c$, and the Plackett--Luce model is misspecified only through its linear scores. In the \emph{coupled} setting, the twins share their draw, $\xi_1=\xi_0$, with probability $0.8$, and are drawn independently otherwise. Each $\xi_j$ is still marginally Gumbel and every pairwise comparison other than the twins' own is unchanged, but item $1$ now ranks above item $0$ with probability $0.91$ rather than $0.53$, at every $x$. No Plackett--Luce law can reproduce this.

We compare GEM with the ranking link to the Plackett--Luce model with linear scores, fitted by maximum likelihood with early stopping on a validation sample, and an oracle that samples from the truth; we have also tried the Plackett--Luce model with spline scores, which is worse at small $n$ in the independent setting and at most marginally better elsewhere. We report four metrics. The total-variation distance compares the true conditional law over the $24$ rankings with its estimate from the cell frequencies of predictive draws. The held-out log-loss of the observed test rankings needs no ground truth and, up to its estimation from draws, is also the criterion on which the baseline is trained. The pairwise-preference error is the mean absolute error of the $12$ probabilities $\P(j\succ l\mid x)$, the lowest-order functionals of the law. Kendall's $\tau$ is computed between the modal ranking of the fitted law and the observed ranking. Sample sizes are $n\in\{500,2000,10000\}$ with $50$ replications each. Appendix~\ref{app:ranking} gives the remaining details and the results with spline scores.

Figure~\ref{fig:sim64_metrics} shows the four metrics against $n$; they broadly agree. In the independent setting the Plackett--Luce model performs better at $n=500$ due to its efficiency and the right form for the noise despite the misspecified mean; at $n=2{,}000$ the two are close, the Plackett--Luce model still ahead in total variation, and GEM is ahead on all four metrics at $n=10{,}000$. In the coupled setting the Plackett--Luce model is far from the truth at every $n$, whereas GEM improves with $n$ and is ahead on all four metrics from $n=2{,}000$. 

\begin{figure}
	\centering
	\begin{tabular}{@{}cc@{}}
		\includegraphics[width=0.485\textwidth]{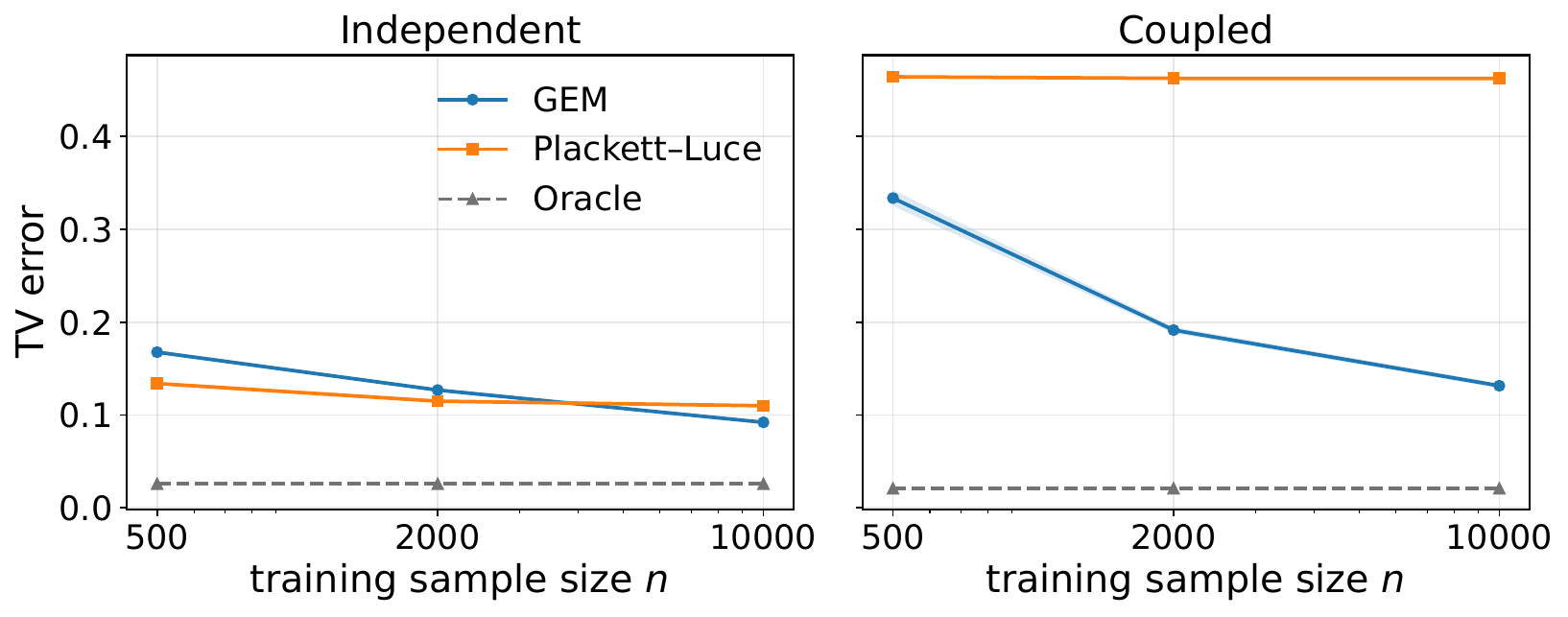} &
		\includegraphics[width=0.485\textwidth]{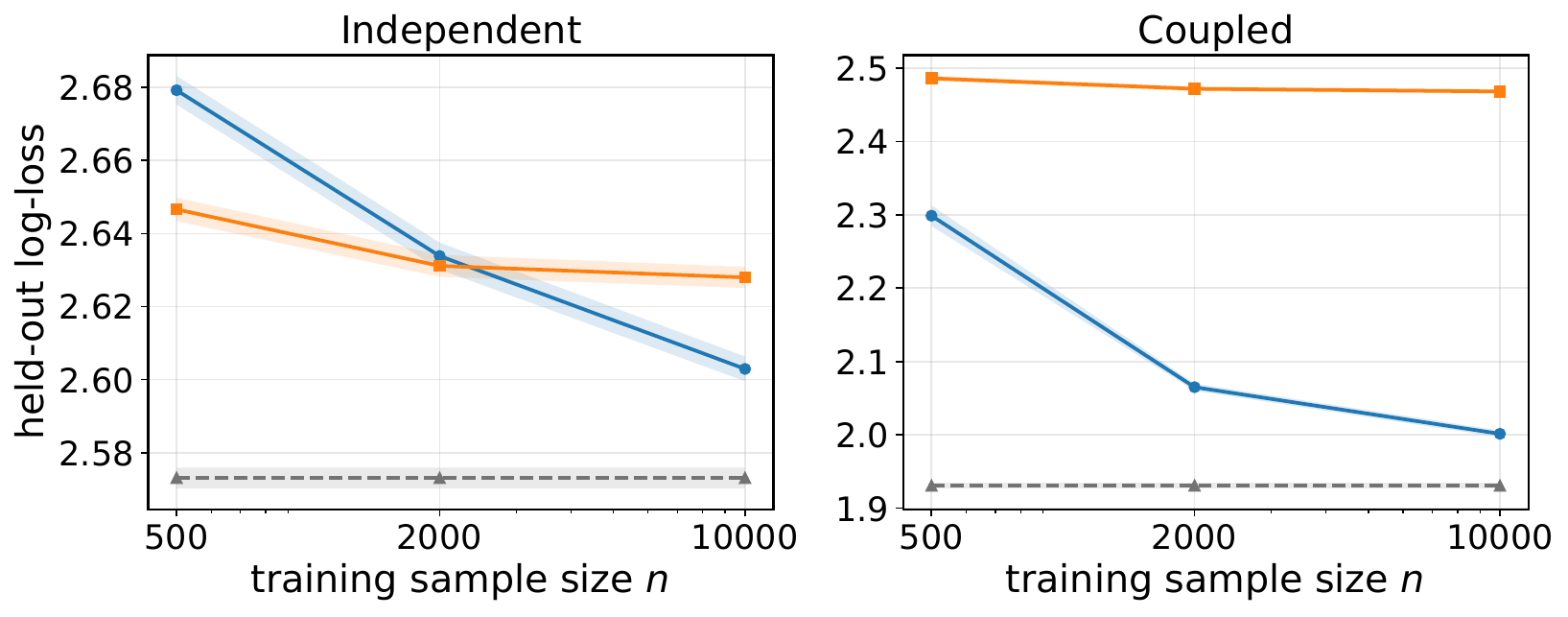}\\
		\includegraphics[width=0.485\textwidth]{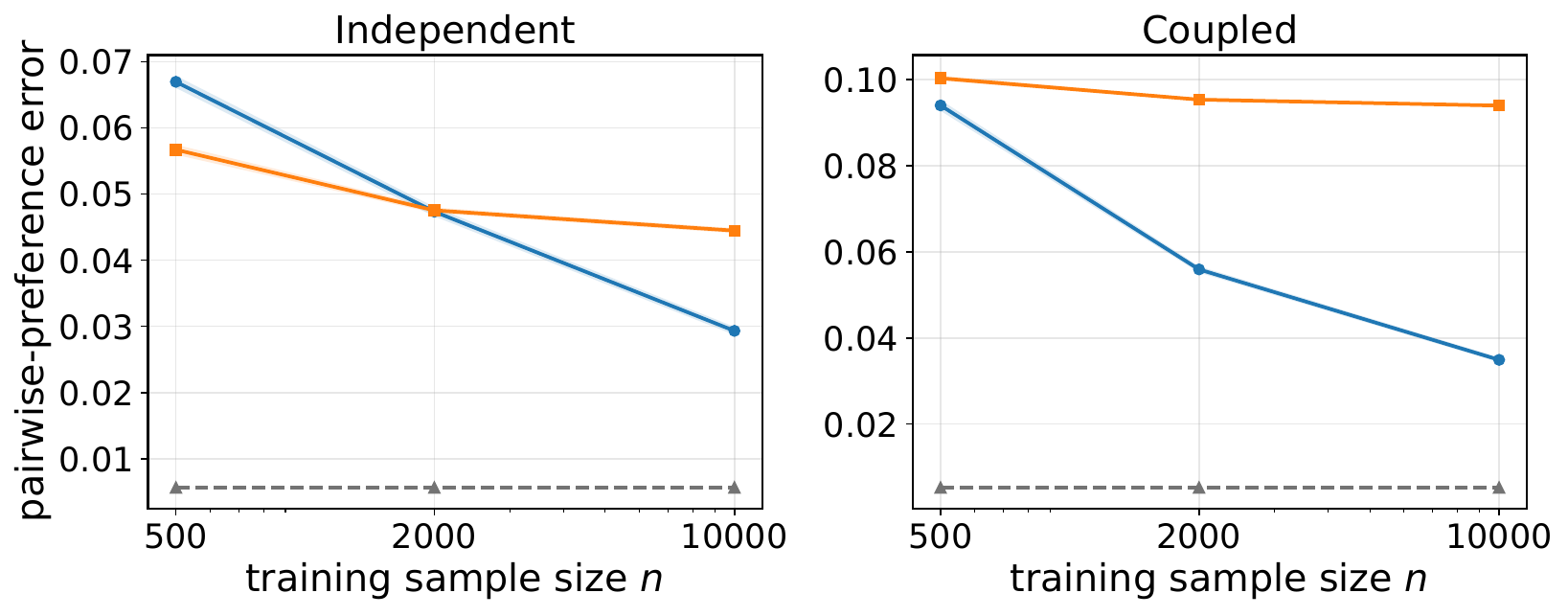} &
		\includegraphics[width=0.485\textwidth]{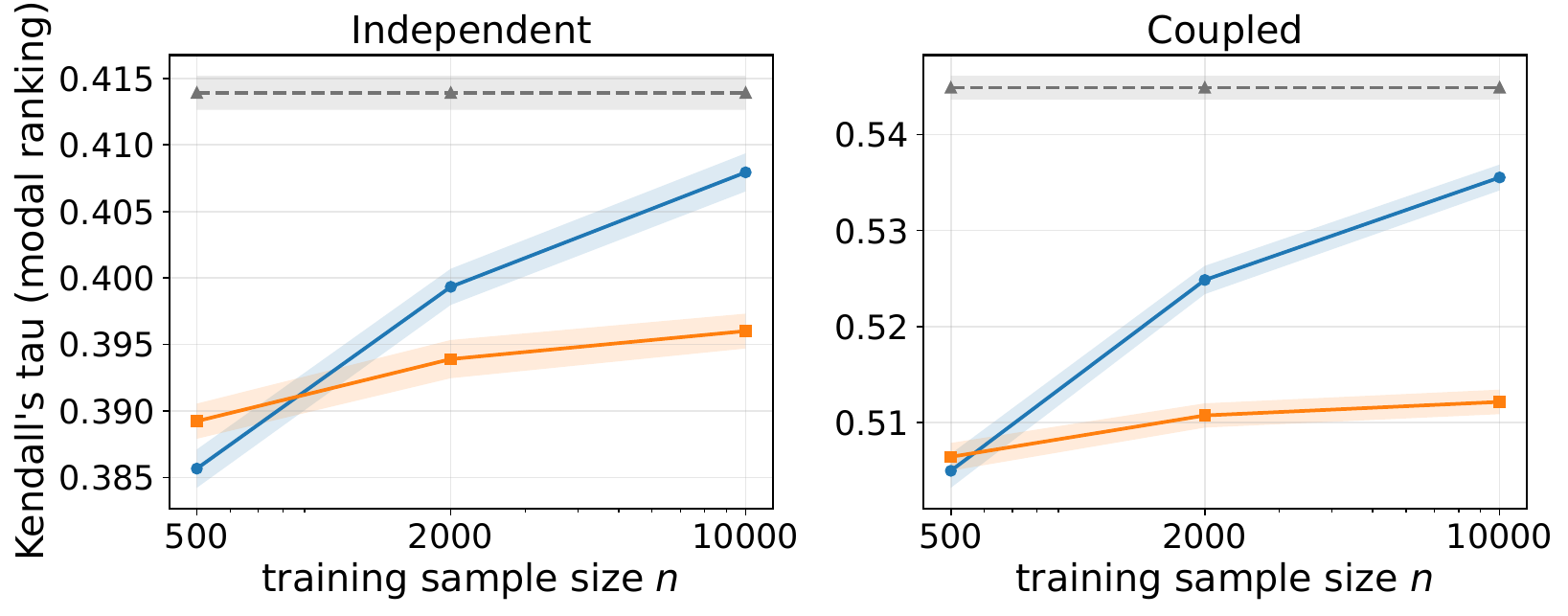}
	\end{tabular}
	\caption{Metrics for the ranking outcome regression against sample size, in the independent setting (left panel of each metric) and the coupled setting (right panel). Lower is better except for Kendall's $\tau$.}
	\label{fig:sim64_metrics}
\end{figure}

\section{Real-world applications}\label{sec:applications}
\subsection{Joint species distribution modeling}\label{sec:jsdm}

In community ecology the outcome at a site is a binary \emph{community vector}, the presence or absence of each of many species. Predictions are judged at two levels: the species level, the occurrence of each species taken alone, and the community level, functionals of the joint law such as the number of species a site holds (its richness), how its species differ from those of another site (its composition), and the co-occurrence left once shared environmental preferences are accounted for \citep{pollock2014understanding,galiana2024power}. The standard joint species distribution models \citep{warton2015many,ovaskainen2017hmsc,niku2019gllvm} place a multivariate probit or logistic regression with a low-rank Gaussian latent factor on the community vector, so the associations between species they can express are pairwise and Gaussian. GEM with the sign link imposes neither restriction.

We take the vegetation data of \citet{niittynen2018importance} from the benchmark of \citet{norberg2019comprehensive}: $242$ species of plants and lichens on $1{,}200$ plots in northern Norway, with four principal components of six environmental variables as covariates. We consider the benchmark's interpolation setting, a random division of the plots into equal training and test halves, and four further random divisions of the same plots. The species are far from independent given the covariates: the variance of richness left after the covariates is about nine times what independence would give, and little of this dependence is spatial (see Appendix~\ref{app:jsdm} for details).

\begin{figure}[t]
\centering
\includegraphics[width=\textwidth]{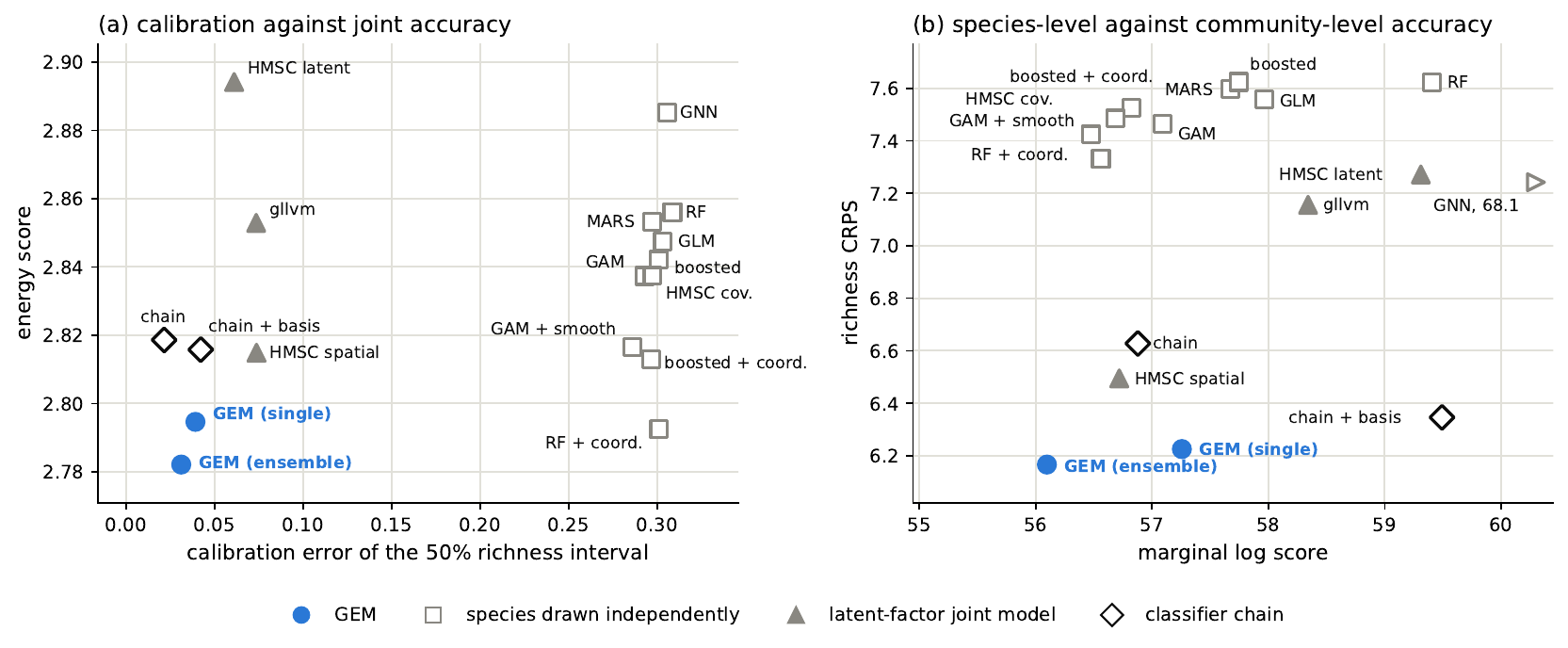}
\caption{Joint species distribution modeling on the vegetation data: the $17$ models, each averaged over the five splits; the labels abbreviate the model names of Table~\ref{tab:veg_interp_full} in Appendix~\ref{app:jsdm}. For all metrics, lower is better. (a) The richness interval error, the calibration error of the central $50\%$ interval, separates the ten models that draw the species independently (squares) from the seven with a dependence structure. (b) The richness CRPS against the marginal log score: the models that draw species independently are good on the marginals and poor on richness, the joint baselines better on richness and worse on the marginals, and GEM best on both. Figure~\ref{fig:veg_paired} in Appendix~\ref{app:jsdm} gives the split-by-split differences.}\label{fig:veg_scatter}
\end{figure}

\begin{table}
	\centering
	\caption{Joint species distribution modeling on the vegetation data: mean over the five splits; Table~\ref{tab:veg_interp_full} in Appendix~\ref{app:jsdm} gives the standard deviations. Columns: energy score; CRPS and interval error of richness and of composition (the $\beta$-S{\o}rensen dissimilarity between two plots); marginal log score (NLL, the negative log-likelihood of the observed occurrences) and AUC of species occurrence; the benchmark's overall score, an average of nine of its measures, each standardized across the models, corrected as in Appendix~\ref{app:jsdm}; time of one fit (GEM on one GPU, baselines on CPU).}
	\label{tab:veg_main}
	\small\setlength{\tabcolsep}{3pt}
	\begin{tabular*}{\textwidth}{@{\extracolsep{\fill}}lccccccccc@{}}
		 & & \multicolumn{2}{c}{richness} & \multicolumn{2}{c}{composition} & \multicolumn{2}{c}{occurrence} & & \\ \cline{3-8}
		model & energy$\scriptstyle\downarrow$ & CRPS$\scriptstyle\downarrow$ & int.$\scriptstyle\downarrow$ & CRPS$\scriptstyle\downarrow$ & int.$\scriptstyle\downarrow$ & NLL$\scriptstyle\downarrow$ & AUC$\scriptstyle\uparrow$ & overall$\scriptstyle\uparrow$ & fit time$\scriptstyle\downarrow$\\
		\hline
		GEM, 5-seed ensemble & \textbf{2.782} & \textbf{6.17} & \textbf{0.031} & \textbf{0.091} & \textbf{0.036} & \textbf{56.1} & \textbf{0.738} & \textbf{1.09} & 5 $\times$ 19 s\\
		GEM, single fit & 2.795 & \textbf{6.23} & 0.039 & \textbf{0.091} & \textbf{0.039} & 57.3 & 0.727 & 0.91 & \textbf{19 s}\\
		\hline
		Random forest + coord. & \textbf{2.792} & 7.33 & 0.301 & 0.105 & 0.290 & 56.6 & 0.734 & 0.24 & 18 min\\
		GAM + spatial smooth & 2.817 & 7.42 & 0.286 & 0.112 & 0.309 & \textbf{56.5} & \textbf{0.738} & 0.05 & 2 min\\
		GAM & 2.837 & 7.46 & 0.293 & 0.113 & 0.314 & 57.1 & 0.726 & $-$0.32 & \textbf{15 s}\\
		HMSC, covariates only & 2.837 & 7.48 & 0.297 & 0.113 & 0.305 & 56.7 & \textbf{0.745} & $-$0.20 & 10 min\\
		HMSC, latent factors & 2.894 & 7.27 & 0.061 & 0.106 & 0.117 & 59.3 & 0.694 & $-$0.78 & 14 min\\
		HMSC, spatial latent f. & 2.815 & 6.49 & 0.074 & 0.095 & 0.093 & 56.7 & 0.731 & 0.66 & 3.9 h\\
		gllvm & 2.853 & 7.16 & 0.073 & 0.100 & 0.127 & 58.3 & 0.730 & $-$0.16 & 82 min\\
		Classifier chain & 2.819 & 6.63 & \textbf{0.021} & 0.097 & 0.079 & 56.9 & 0.735 & 0.44 & 2 min\\
		Classifier chain + basis & 2.816 & 6.35 & 0.042 & 0.093 & 0.073 & 59.5 & 0.735 & \textbf{1.01} & 3 min\\
		\hline
	\end{tabular*}
\end{table}

GEM uses the sign link and, because the number of species is large, the control variate of Section~\ref{sec:algorithm}, with the four environmental components as its only inputs. We report a single fit and, inspired by deep ensembles \citep{lakshminarayanan2017simple}, an ensemble of five fits with different seeds, whose predictive distribution is the equal mixture of the five. With so few training plots a single fit depends on its random seed, and the ensemble averages over it. Of the $15$ baselines, ten draw the species independently, among them a generalized additive model (GAM), a random forest and HMSC (Hierarchical Modelling of Species Communities; \citealp{ovaskainen2017hmsc}) with covariates only, and the GAM and the random forest also with the plot coordinates. Five model the joint law: HMSC with latent factors and with spatial latent factors \citep{ovaskainen2017hmsc}, gllvm \citep{niku2019gllvm}, and a classifier chain \citep{read2011classifier} with and without a spatial basis. Unlike GEM, the spatial variants of the baselines also use the plot coordinates. Appendix~\ref{app:jsdm} gives the settings and the tuning of every model and how each produces its draws.

The energy score assesses the whole community vector. We also score the three targets on which the benchmark evaluates predictions, the column groups of Table~\ref{tab:veg_main}. At the community level, the targets are the richness of a plot and its composition, measured for a pair of plots by the $\beta$-S{\o}rensen dissimilarity, which is zero when the two plots hold the same species and one when they share none. Each gets the CRPS of Section~\ref{sec:background} and an interval error, the distance between $50\%$ and the coverage of the central $50\%$ interval, computed with the randomized probability integral transform \citep{czado2009predictive}. At the species level, the occurrence of each species gets the marginal log score, which scores the species one at a time, and the AUC, averaged over species. We also report the benchmark's overall score, which averages its measures of accuracy, discrimination and calibration for the three targets, each standardized across the models; Appendix~\ref{app:jsdm} defines it and every other metric, and explains a correction we make to it.

Figure~\ref{fig:veg_scatter} and Table~\ref{tab:veg_main} give the results. The GEM ensemble is first on six of the eight metrics in Table~\ref{tab:veg_main} and second on the other two, the richness interval error and the AUC, where the classifier chain and HMSC with covariates only, respectively, are ahead. It is therefore the best or nearly the best model for the whole community vector (energy score), for the community-level targets (richness and composition) and for the occurrence of each species (log score and AUC), and the only one among the best at all three. Appendix~\ref{app:jsdm} shows that GEM captures pairwise associations less well than other joint baselines, as measured by the log score of pairs of common species (Table~\ref{tab:veg_interp_full}). Its calibrated and accurate predictions of richness and composition show, however, that it captures the dependence among many species at once, such as the sum of the covariances, which sets the variance of richness, and possibly higher-order dependence. Together with its accurate marginals, this gives GEM the best joint performance.

In panel (a), all ten models that draw the species independently are far from calibrated for richness, their central $50\%$ interval covering only about a fifth of the plots, while all seven with a dependence structure are close to calibrated; no model lies in between. Independent draws cannot reproduce the variance of richness left after the covariates, so their predictions of richness are overconfident, whatever their marginals. A model to be trusted at the community level needs a dependence structure, and GEM has one without assuming a latent-factor form.

The same holds for composition. In the composition columns of Table~\ref{tab:veg_main}, GEM has the best CRPS and by far the smallest interval error, below that of every baseline on four of the five splits. The classifier chain, the best calibrated for richness, is much less well calibrated for composition.

Panel (b) shows why the two have to come together. The models that draw species independently include the baselines with the best marginals, the GAM with a spatial smooth, the random forest with coordinates and HMSC with covariates only, yet all ten are worse on richness than the best joint models, so good marginals alone do not give good community predictions. The gradient nearest-neighbor model, whose marginal log score is far worse than that of any other model, is drawn at the right edge of the panel. Four of the five joint baselines do better on richness than every model that draws species independently (HMSC with latent factors is level with the gradient nearest-neighbor model), but each has worse marginals than the GEM ensemble, and HMSC with latent factors, which recovers the most pairwise association, has one of the lowest AUCs of all. GEM alone is in the lower left, best on both. On these data a flexible joint model is therefore preferable for predicting community-level quantities.

Table~\ref{tab:veg_interp_full} in Appendix~\ref{app:jsdm} gives all $17$ models, and Figure~\ref{fig:veg_paired} in Appendix~\ref{app:jsdm} compares the ensemble with nine of the baselines split by split. On every split, the ensemble has a lower energy score than every baseline except the random forest with coordinates, and a lower richness CRPS than every baseline except the classifier chain with a spatial basis; it beats these two on most splits. Its AUC is level with that of the GAM with a spatial smooth and behind only that of HMSC with covariates only, and it is first on the benchmark's corrected overall score, just ahead of the chain with a spatial basis. The single fit keeps the community-level results but not the lead on the marginals, which the ensemble supplies.

GEM also has a clear computational advantage, since it is fast both to fit and to draw from. Trained by stochastic gradient descent, it takes seconds per fit on one GPU, against minutes to hours on CPU for the joint baselines, up to almost four hours for HMSC with spatial latent factors (Table~\ref{tab:veg_main}); even the five fits of the ensemble take under two minutes together. Each further draw from the fitted model is one pass through the network, whereas each draw from HMSC is a retained sample of its Markov chain, so more draws require a longer chain. These data are small, with $600$ training plots. Each gradient step of GEM costs time linear in the number of plots, and mini-batches can cap it, whereas established joint species distribution models such as HMSC become slow on large community data \citep{pichler2021new}, so the computational advantage of GEM should grow with the size of the data.

\subsection{Health examination surveys}\label{sec:nhanes}

Health surveys record measurements of several types at once for each participant: continuous laboratory and examination values, binary indicators of diagnoses and medication, and ordinal self-reports. The clinically interesting quantities are functionals of their joint distribution given covariates rather than of any single coordinate.
A model fitted separately to each block cannot answer such questions without an added model of the dependence between the blocks, such as a copula.

We use the National Health and Nutrition Examination Survey \citep[NHANES;][]{nhanes}, cycles 2007--2008 through 2017--2018, restricted to examined non-pregnant adults aged $20$ and over with complete outcomes, giving $n=27{,}983$ participants. The outcome has $13$ variables: seven continuous measurements (body mass index, waist circumference, systolic and diastolic blood pressure, HbA1c, total and HDL cholesterol), five binary indicators (diagnosed diabetes, diagnosed hypertension, blood-pressure medication, diagnosed high cholesterol, cholesterol medication), and self-rated health on a five-point ordinal scale from excellent to poor. Four of the continuous variables are log-transformed and all seven are standardized; the binary block is coded in $\{-1,+1\}$ and self-rated health as an indicator vector of length five, except that GEM models it as one ordinal coordinate (Appendix~\ref{app:nhanes}). The covariates, coded as $29$ columns, are demographic, socioeconomic and behavioral: age, sex, race and ethnicity, education, income-to-poverty ratio, marital status, smoking, recreational activity, alcohol use and survey cycle. Survey weights are ignored, the target being conditional prediction rather than population estimation. We report means and standard errors of metrics over $20$ random $60/20/20$ splits into training, validation and test sets, shared across methods.

GEM is applied with the identity, sign and ordinal links. Three baselines model the blocks separately or resample them jointly. The {stacked marginals} pipeline is built as in Section~\ref{sec:mixed}, with engression on the continuous block, an independent network per binary label and a softmax network for self-rated health; it targets each block's marginals directly and has no dependence between the blocks or among the binary labels. The {Gaussian copula} couples those same marginals through a Gaussian dependence structure fitted on their latent scores. {Distributional random forests} (DRF) \citep{cevid2022drf} is a nonparametric alternative whose predictive draws resample entire observed outcome vectors with covariate-dependent weights, so that any dependence present in the sample is preserved, but only outcome vectors that were observed can be drawn. Based on $500$ draws of the outcome vector per test point, we report the joint energy score on the $17$-dimensional outcome vector together with the per-block marginal metrics: CRPS averaged over the seven continuous coordinates, Bernoulli negative log-likelihood averaged over the five binary labels, and, for self-rated health, the ranked probability score $\sum_{k}(\hat F_k-\mathbf 1\{y\le k\})^2$ for an estimated distribution $\hat F$.

\begin{figure}
\centering
\includegraphics[width=\textwidth]{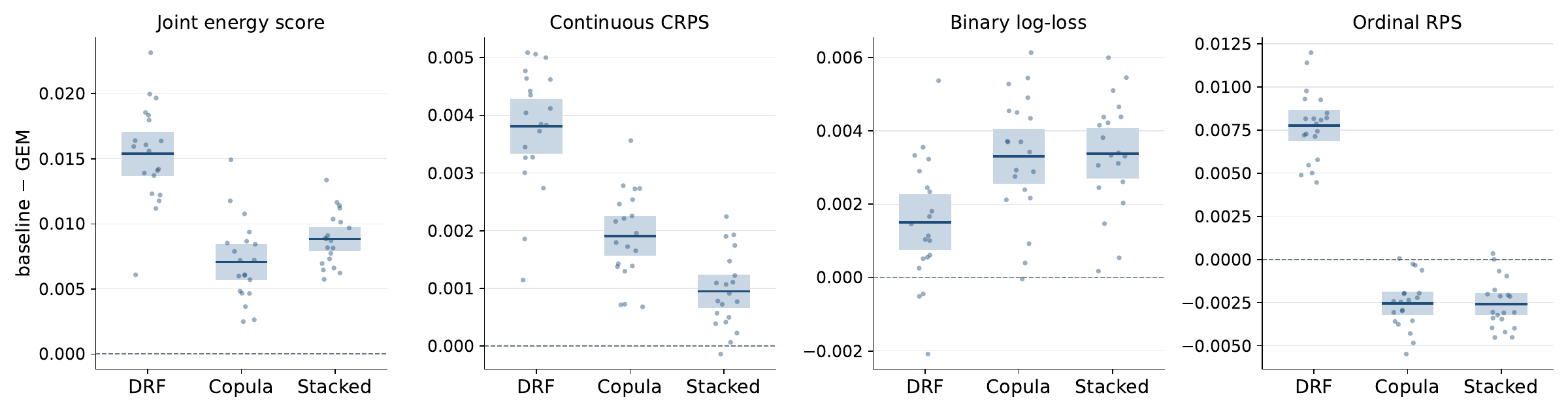}
\caption{Health examination survey data from NHANES: paired per-split differences, baseline minus GEM, for the joint energy score and the per-block marginal metrics over $20$ random splits. A positive value favors GEM. Each dot is one split; the line is the mean across splits and the band extends two standard errors on either side.}\label{fig:nhanes}
\end{figure}

Figure~\ref{fig:nhanes} shows, for each metric and each of the $20$ splits, the score of a baseline minus that of GEM on the same split, so that a positive difference favors GEM; Table~\ref{tab:nhanes_absolute} in Appendix~\ref{app:nhanes} gives the absolute values. On the joint energy score the difference is positive on every split against every baseline, and largest against the distributional random forest, so that on this criterion GEM fits the joint distribution best. This is not at the expense of the marginals. The continuous CRPS is lower on every split against DRF and the copula and on all but one against the stacked marginals, and the binary log-loss is lower on every split against the stacked marginals, on $19$ of the $20$ against the copula and on $17$ of the $20$ against DRF. Appendix~\ref{app:nhanes} breaks both comparisons down by outcome, with the same paired differences per variable. The ordering holds for every continuous variable and every binary label against the copula and the stacked marginals; against DRF, the only exception larger than two standard errors is HbA1c, where resampling observed outcome vectors reproduces a long right tail among diagnosed diabetics that GEM does not. Self-rated health is the one block where GEM does not lead but remains competitive: the ranked probability score averages $0.4980$ against $0.4954$ for the softmax-based baselines, a shortfall of half a percent that recurs on nearly every split. Computationally, GEM takes on average $62$ seconds on one GPU, against $55$ seconds for the stacked baseline and $244$ for DRF on CPU (Appendix~\ref{app:nhanes}).

\section{Discussion}\label{sec:discussion}

We have proposed generalized engression models. Their main features are as follows.
\begin{enumerate}[label=(\roman*)]%
\setlength{\itemsep}{2pt}
\item GEMs are distributional: by generating draws from it, they estimate the joint conditional distribution of a multivariate $Y$ given $X$, which carries the uncertainty of each coordinate as well as the dependence between them.
\item GEMs are nonparametric in both $X$ and $Y$: the covariates enter the model as in nonparametric regression, and the response distribution is not restricted to any family. 
\item GEMs apply to outcomes whose coordinates may be continuous, binary, categorical, ordinal or rankings, or a mixture of these.
\item GEMs are fitted in a unified way: the same loss function and the same algorithm serve every outcome type, whereas a GLM requires a likelihood specified for each family.
\end{enumerate}

GEMs share the practical features of engression \citep{shen2024engression}, both being lightweight generative models. They are fitted by minimizing a single loss evaluated on samples from the model, with no adversary and no explicit likelihood, which makes training stable and fast. The fitted model is a one-step sampler, from which means, quantiles, event probabilities and other functionals follow by Monte Carlo. For continuous outcomes the two methods coincide, but GEM has a much broader scope. Outcomes that are partly or entirely discrete, such as presence--absence records, diagnoses and preference rankings, are common in applications, and they can now be modeled in the same way as continuous ones.

We hope that GEMs will serve as a practical statistical tool for the multivariate discrete and mixed-type outcomes that arise in scientific applications. The two applications illustrate their potential. On the arctic vegetation survey, with $242$ binary presence--absence indicators, GEM matches or improves on the specialized joint species distribution models in predicting the joint distribution. On the health survey, whose outcome mixes continuous, binary and ordinal coordinates, a single GEM gives a better joint predictive distribution than stacked marginal models, a Gaussian copula and distributional random forests, with little or no loss of accuracy on the individual outcomes. 

Several aspects of the framework call for further work. Modeling the joint distribution comes at a small price on some individual margins: in our experiments, type-specific methods were occasionally more accurate on a single coordinate, as with a softmax classifier on a categorical outcome, and the Plackett--Luce model when it holds and the sample is small. This is to be expected, since the energy score is designed for the joint distribution and, for a categorical coordinate, reduces to a Brier score rather than the log-likelihood on which such classifiers are trained. One way to recover the lost efficiency is to add classical type-specific losses, such as marginal log-likelihoods, to the energy score. Because a sum of proper scoring rules is again proper, the combined objective still targets the true joint distribution; and because, given the latent noise, the class probabilities of the smoothed model are available in closed form for binary and ordinal coordinates, and for categorical ones under a Gumbel perturbation, the added terms can be estimated by Monte Carlo. How to weight the two parts, and how much is gained in practice, remain to be explored.

Engression was applied by \citet{shen2024engression} to extrapolation, that is, prediction outside the training support of the covariates. It remains open whether GEM can also extrapolate when the outcome is discrete. The question matters in practice, for instance when species communities are to be predicted under environmental conditions that do not occur in the training data.

On the theoretical side, the results of Section~\ref{sec:universality} concern representation rather than estimation: they show that the model class can express the target distribution, but not how quickly the fitted distribution approaches it. Fitting a GEM remains a score minimization problem of the standard kind, an instance of optimum score estimation \citep{gneiting2007strictly}, so that guarantees for the minimizer, such as consistency, can be obtained by classical techniques. Error bounds have recently been established for engression with continuous outcomes \citep{huang2026theoretical,chen2026error}. Nevertheless, meaningful finite-sample results of GEMs are highly involved, and we leave them for future work.

\section*{Acknowledgement}
The authors thank Nicolai Meinshausen for helpful feedback on the manuscript.
Xinwei Shen was supported by Royalty Research Fund grant GR067344. Francis Bach has received support from the French government, managed by the National Research Agency,
under the France 2030 program with the reference “PR[AI]RIE-PSAI” (ANR-23-IACL-0008).
Large language models (Claude, Anthropic; ChatGPT, OpenAI) were used to assist with code implementation, improve the language of the manuscript, search for references, and validate proofs; the authors reviewed and edited all output and take full responsibility for the content.

\bibliography{ref.bib}
\bibliographystyle{apalike2}

\clearpage

\appendix

\section{Proofs}
\label{app:proofs}

\subsection{Proof of Proposition~\ref{prop:grad}}\label{app:proof_grad}
Proposition~\ref{prop:grad} follows from the following lemma, applied to the Gaussian density.

\begin{Lemma}\label{lem:grad}
	Let $\eta$ have a continuously differentiable density $p$ on $\rb^k$ with
	\[
		p(\xi)+(1+\|\xi\|)\|\nabla p(\xi)\|\le C(1+\|\xi\|)^{-(k+1+\delta)}
	\]
	for some constants $C,\delta>0$, and let $f:\rb^k\to\rb$ be measurable with $|f(u)|\le C'(1+\|u\|)$ for a constant $C'$. Then $F(\mu,s)=\E f(\mu+s\odot\eta)$, the expectation over $\eta$ at fixed $(\mu,s)$, is continuously differentiable on $\rb^k\times(0,\infty)^k$, and for $j=1,\dots,k$
	\begin{equation}\label{eq:grad}
		\begin{aligned}
			\frac{\partial F}{\partial\mu_j}(\mu,s)&=-\frac{1}{s_j}\E\big[f(\mu+s\odot\eta)\,\partial_j\log p(\eta)\big],\\
			\frac{\partial F}{\partial s_j}(\mu,s)&=-\frac{1}{s_j}\E\big[f(\mu+s\odot\eta)\,(1+\eta_j\,\partial_j\log p(\eta))\big],
		\end{aligned}
	\end{equation}
	where the expectations are again over $\eta$, and $\partial_j\log p=\partial_jp/p$, set to zero where $p$ vanishes.
\end{Lemma}

\begin{proof}[Proof of Lemma~\ref{lem:grad}]
The function $f$ need not be differentiable, so we cannot differentiate $f(\mu+s\odot\eta)$ in $(\mu,s)$. Instead, we move $(\mu,s)$ from the argument of $f$ into the density of $\mu+s\odot\eta$, which is differentiable.

\textit{Step 1: change of variables.} The substitution $u=\mu+s\odot\xi$, with $\xi=(u-\mu)\oslash s$ and $du=\prod_{l=1}^ks_l\,d\xi$, gives
\begin{equation*}
	F(\mu,s)=\int f(\mu+s\odot\xi)\,p(\xi)\,d\xi=\int f(u)\,p_{\mu,s}(u)\,du,\qquad p_{\mu,s}(u)=p\big((u-\mu)\oslash s\big)\prod_{l=1}^ks_l^{-1},
\end{equation*}
where $p_{\mu,s}$ is the density of $\mu+s\odot\eta$. In the second integral, $(\mu,s)$ enters only through $p_{\mu,s}$; for each fixed $u$, $f(u)$ is a constant, and no derivative of $f$ is needed.

\textit{Step 2: derivatives of the density.} Since $\xi_j=(u_j-\mu_j)/s_j$, we have $\partial\xi_j/\partial\mu_j=-s_j^{-1}$ and $\partial\xi_j/\partial s_j=-s_j^{-1}\xi_j$, while the other coordinates of $\xi$ do not depend on $\mu_j$ or $s_j$, and $\partial\big(\prod_ls_l^{-1}\big)/\partial s_j=-s_j^{-1}\prod_ls_l^{-1}$. The chain rule gives
\begin{equation*}
	\frac{\partial p_{\mu,s}(u)}{\partial\mu_j}=-\frac{1}{s_j}\,\partial_jp(\xi)\prod_{l=1}^ks_l^{-1},\qquad
	\frac{\partial p_{\mu,s}(u)}{\partial s_j}=-\frac{1}{s_j}\big(p(\xi)+\xi_j\,\partial_jp(\xi)\big)\prod_{l=1}^ks_l^{-1},
\end{equation*}
both continuous in $(\mu,s)$ because $p$ is continuously differentiable.

\textit{Step 3: differentiation under the integral sign.} Let $\mathcal K$ be a compact subset of $\rb^k\times(0,\infty)^k$, so that $\|\mu\|\le M$ and $a\le s_l\le b$ for all $l$ and all $(\mu,s)\in\mathcal K$, for some constants $M\ge0$ and $0<a\le b$. Since $\|u\|\le M+b\|\xi\|$, we have $1+\|u\|\le(1+M+b)(1+\|\xi\|)$. The condition on $p$ bounds both $\|\nabla p(\xi)\|$ and $p(\xi)+|\xi_j\,\partial_jp(\xi)|$ by $C(1+\|\xi\|)^{-(k+1+\delta)}$, so for all $(\mu,s)\in\mathcal K$ both derivatives in Step~2 are at most
\begin{equation*}
	\frac{C}{a^{k+1}}(1+\|\xi\|)^{-(k+1+\delta)}\le\frac{C(1+M+b)^{k+1+\delta}}{a^{k+1}}(1+\|u\|)^{-(k+1+\delta)}
\end{equation*}
in absolute value. Multiplied by $|f(u)|\le C'(1+\|u\|)$, they are bounded by a constant times $(1+\|u\|)^{-(k+\delta)}$, which does not depend on $(\mu,s)$ and is integrable over $\rb^k$ because $k+\delta>k$; a bound of the same form holds for $|f(u)|\,p_{\mu,s}(u)$, so $F$ is finite. By the mean value theorem and dominated convergence, $F$ can therefore be differentiated under the integral sign at every interior point of $\mathcal K$, and, as $\mathcal K$ is arbitrary, on all of $\rb^k\times(0,\infty)^k$:
\begin{equation*}
	\frac{\partial F}{\partial\mu_j}(\mu,s)=\int f(u)\,\frac{\partial p_{\mu,s}(u)}{\partial\mu_j}\,du,\qquad \frac{\partial F}{\partial s_j}(\mu,s)=\int f(u)\,\frac{\partial p_{\mu,s}(u)}{\partial s_j}\,du.
\end{equation*}
These derivatives are continuous in $(\mu,s)$, again by dominated convergence, because the integrands are continuous in $(\mu,s)$ for each $u$ and share the same integrable bound. Hence $F$ is continuously differentiable.

\textit{Step 4: back to expectations over $\eta$.} Substituting $u=\mu+s\odot\xi$ back, the factor $du=\prod_ls_l\,d\xi$ cancels $\prod_ls_l^{-1}$, and
\begin{equation*}
	\frac{\partial F}{\partial\mu_j}(\mu,s)=-\frac{1}{s_j}\int f(\mu+s\odot\xi)\,\partial_jp(\xi)\,d\xi,\qquad
	\frac{\partial F}{\partial s_j}(\mu,s)=-\frac{1}{s_j}\int f(\mu+s\odot\xi)\big(p(\xi)+\xi_j\,\partial_jp(\xi)\big)\,d\xi.
\end{equation*}
Where $p(\xi)=0$, the point $\xi$ minimizes the nonnegative, continuously differentiable function $p$, so $\nabla p(\xi)=0$. Hence $\partial_jp=p\,\partial_j\log p$ everywhere, with $\partial_j\log p$ set to zero where $p$ vanishes, and $p+\xi_j\partial_jp=p\,(1+\xi_j\partial_j\log p)$. The two integrals are therefore the expectations over $\eta\sim p$ in \eqref{eq:grad}.
\end{proof}

\begin{proof}[Proof of Proposition~\ref{prop:grad}]
Fix $(X,Y,\varepsilon,\varepsilon')$, so that $Y$, $\mu$, $\mu'$ and $s$ are fixed. Given $(X,Y)$, the pairs $(\varepsilon,\eta)$ and $(\varepsilon',\eta')$ are independent with the same distribution, so $\E\|Y-\hat Y\|=\E\|Y-\hat Y'\|$: the integrand of \eqref{eq:gem_objective} and $\frac12\|Y-\hat Y\|+\frac12\|Y-\hat Y'\|-\frac12\|\hat Y-\hat Y'\|$ have the same expectation, and $\mathcal L=\E\,\bar\ell(\mu,\mu',s)$.

We apply Lemma~\ref{lem:grad} in dimension $2k$ to the two perturbations together. Let $\tilde\eta=(\eta,\eta')\sim\mathcal N(0,I_{2k})$, $\tilde\mu=(\mu,\mu')$ and $\tilde s=(s,s)$, and for $u,u'\in\rb^k$ let
\begin{equation*}
	\tilde f(u,u')=\tfrac12\|Y-h(u)\|+\tfrac12\|Y-h(u')\|-\tfrac12\|h(u)-h(u')\|,
\end{equation*}
so that $\bar\ell(\mu,\mu',s)=\E\,\tilde f(\tilde\mu+\tilde s\odot\tilde\eta)$, which is the function $F$ of the lemma with $\tilde f$ and $\tilde\eta$ in place of $f$ and $\eta$. Every link of Table~\ref{tab:links} satisfies $\|h(u)\|\le C_h(1+\|u\|)$ for a constant $C_h$, the identity trivially and the other links because their values lie in a bounded set, so $\tilde f$ is measurable with $|\tilde f(u,u')|\le\|Y\|+\|h(u)\|+\|h(u')\|\le C'(1+\|(u,u')\|)$ for a constant $C'$ depending on $Y$. The Gaussian density satisfies the condition of the lemma for every $\delta>0$, and $\partial_j\log p(\xi)=-\xi_j$. Hence $F(\tilde\mu,\tilde s)$ is continuously differentiable on $\rb^{2k}\times(0,\infty)^{2k}$, and so is $\bar\ell(\mu,\mu',s)=F((\mu,\mu'),(s,s))$, with, by \eqref{eq:grad} and the chain rule,
\begin{align*}
	\frac{\partial\bar\ell}{\partial\mu_j}&=\frac{1}{s_j}\E\big[\tilde f(\tilde\mu+\tilde s\odot\tilde\eta)\,\eta_j\big],\qquad
	\frac{\partial\bar\ell}{\partial\mu'_j}=\frac{1}{s_j}\E\big[\tilde f(\tilde\mu+\tilde s\odot\tilde\eta)\,\eta'_j\big],\\
	\frac{\partial\bar\ell}{\partial s_j}&=\frac{1}{s_j}\E\big[\tilde f(\tilde\mu+\tilde s\odot\tilde\eta)\,(\eta_j^2-1)\big]+\frac{1}{s_j}\E\big[\tilde f(\tilde\mu+\tilde s\odot\tilde\eta)\,((\eta'_j)^2-1)\big],
\end{align*}
the derivative in $s_j$ collecting the two coordinates of $\tilde s$ that equal $s_j$. Finally, $\tilde f(\tilde\mu+\tilde s\odot\tilde\eta)=\frac12c+\frac12\|Y-\hat Y'\|$, where $\|Y-\hat Y'\|$ does not depend on $\eta$, and $\eta$ is independent of $\eta'$ with $\E[\eta_j]=\E[\eta_j^2-1]=0$, so that $\E[\|Y-\hat Y'\|\,\eta_j]=\E[\|Y-\hat Y'\|\,(\eta_j^2-1)]=0$; in the same way $\tilde f(\tilde\mu+\tilde s\odot\tilde\eta)=\frac12c'+\frac12\|Y-\hat Y\|$ with $\|Y-\hat Y\|$ independent of $\eta'$. Substituting into the three derivatives gives $G$, $G'$ and $S$ as stated.
\end{proof}

\subsection{Proof of Proposition~\ref{prop:noise_outsourcing_h}}
For part (i), define $Z=\bar h(Y)$, so that $Y=h(Z)$ a.s. By Lemma~\ref{lem:noise_outsourcing} applied to $(X,Z)$, there exist $\varepsilon\sim\mathcal N(0,1)$, independent of $X$, and a measurable map $g:\X\times\rb\to\rb^k$ such that $(X,Z)\overset{d}{=}(X,g(X,\varepsilon))$. Applying the measurable map $(x,z)\mapsto(x,h(z))$ to both sides gives $(X,Y)\overset{d}{=}(X,h(g(X,\varepsilon)))$.

For general standard Borel $\Y$, let $P_Y$ denote the law of $Y$. By \citet[Theorem~3.1]{arveson2008noncommutative}, since $h$ is surjective, there exist a Borel set $N\subseteq\Y$ with $P_Y(N)=0$ and a Borel map $\tilde h:\Y\setminus N\to\rb^k$ such that $h(\tilde h(y))=y$ for all $y\in\Y\setminus N$. Fix $z_\ast\in\rb^k$ and define
\[
	\bar h(y)=
	\begin{cases}
		\tilde h(y), & y\notin N,\\
		z_\ast, & y\in N.
	\end{cases}
\]
Since $N$ is Borel, $\bar h$ is Borel measurable, and $h(\bar h(Y))=Y$ almost surely. The claim then follows from part (i).

\subsection{Proof of Lemma~\ref{lem:continuity}}
Condition~\ref{cond:regular} implies that the conditional density $f_j(s\mid x,y_{<j})$ is strictly positive and continuous in all its arguments. Tonelli's theorem and the definition of the sequential marginals give
\[
	\int_{\mathbb{R}}f_j(s\mid x,y_{<j})\,ds=1.
\]
Write $\theta=(x,y_{<j})$ and consider $(\theta_n,t_n)\to(\theta,t)$. Then
\begin{equation}
	\begin{aligned}
		|F_j(t_n\mid\theta_n)-F_j(t\mid\theta)|\leq\int_{-\infty}^{\infty}|f_j(s\mid\theta_n)-f_j(s\mid\theta)|\,ds+\int_{\min\{t,t_n\}}^{\max\{t,t_n\}}f_j(s\mid\theta)\,ds.
	\end{aligned}
	\label{eq:cdf_continuity_bound}
\end{equation}
For each fixed $s$, continuity gives $f_j(s\mid\theta_n)\to f_j(s\mid\theta)$. Since all these functions integrate to one, we have
\[
	\int_{-\infty}^{\infty}|f_j(s\mid\theta_n)-f_j(s\mid\theta)|\,ds\to0.
\]
The second term in \eqref{eq:cdf_continuity_bound} converges to zero by the absolute continuity of the Lebesgue integral, because $f_j(\cdot\mid\theta)\in L^1(\mathbb{R})$ and $t_n\to t$. This proves the joint continuity of $F_j$. Since $f_j(s\mid\theta)>0$ for every $s$, $F_j(\cdot\mid\theta)$ is strictly increasing.

We next establish the continuity of $Q_j$. For each $\theta$, $F_j(t\mid\theta)\to0$ as $t\to-\infty$ and $F_j(t\mid\theta)\to1$ as $t\to\infty$. Its continuity and strict monotonicity therefore imply that, for every $\tau\in(0,1)$, $q=Q_j(\tau\mid\theta)$ is finite, unique, and satisfies $F_j(q\mid\theta)=\tau$. Consider $(\theta_n,\tau_n)\to(\theta,\tau)$ with $\tau\in(0,1)$. For any $\delta>0$, strict monotonicity gives
\[
	F_j(q-\delta\mid\theta)<\tau<F_j(q+\delta\mid\theta).
\]
By the joint continuity of $F_j$ and $\tau_n\to\tau$, for all sufficiently large $n$,
\[
	F_j(q-\delta\mid\theta_n)<\tau_n<F_j(q+\delta\mid\theta_n).
\]
Strict monotonicity then implies
\[
	q-\delta<Q_j(\tau_n\mid\theta_n)<q+\delta.
\]
Hence $Q_j(\tau_n\mid\theta_n)\to Q_j(\tau\mid\theta)$, proving the claimed joint continuity.

\subsection{Proof of Theorem~\ref{thm:continuous}}
In the following, we conduct the coordinate-by-coordinate construction of $(y_1,\dots,y_k)$. Let $\Phi$ denote the standard normal CDF and define $U_j=\Phi(\varepsilon_j)$ for $1\leq j\leq k$. Then $U_1,\ldots,U_k$ are independent $\operatorname{Unif}(0,1)$ random variables and are independent of $X$.

For $j=1$, generate the first coordinate as
\[
	\widetilde Y_1=Q_1(U_1\mid X)\mathrel{=:}T_1(U_1,X).
\]
Lemma~\ref{lem:continuity} implies that $T_1$ is continuous in both $U_1$ and $X$. For $2\leq j\leq k$, conditional on $X$ and the previously generated coordinates $\widetilde Y_{<j}$, generate the $j$th coordinate through the conditional quantile function
\[
	\widetilde Y_j=Q_j(U_j\mid X,\widetilde Y_{<j})\mathrel{=:}T_j(U_1,\ldots,U_j,X).
\]
Because $U_j$ is independent of $(X,U_{<j})$, it is also independent of $(X,\widetilde Y_{<j})$. Consequently, conditional on $(X,\widetilde Y_{<j})$, the random variable $\widetilde Y_j$ has CDF $F_j(\cdot\mid X,\widetilde Y_{<j})$. Lemma~\ref{lem:continuity} and induction on $j$ show that $T_j$ is continuous in all its arguments.

Define
\[
	\begin{aligned}
		g(x,\varepsilon)=\big(&T_1\{\Phi(\varepsilon_1),x\},T_2\{\Phi(\varepsilon_1),\Phi(\varepsilon_2),x\},\ldots,T_k\{\Phi(\varepsilon_1),\ldots,\Phi(\varepsilon_k),x\}\big).
	\end{aligned}
\]
The continuity of $\Phi$ and of $T_1,\ldots,T_k$ implies that $g$ is continuous on $\X\times\mathbb{R}^k$. Moreover, the sequential conditional distributions of $\widetilde Y=g(X,\varepsilon)$ coincide with those of $Y$ given $X$. Induction over $j$, or equivalently the conditional-density factorization, therefore gives
\[
	\left(X,g(X,\varepsilon)\right)\overset{d}{=}(X,Y),
\]
which proves Theorem~\ref{thm:continuous}. In the Euclidean case, the link $h$ is simply the identity.

\subsection{Proof of Theorem~\ref{thm:lift_criterion}}
Let $(X,Z)$ have the lift density $a$. Apply Theorem~\ref{thm:continuous} with $Z$ as the continuous outcome. It produces a continuous $g$ such that
\[
	(X,g(X,\varepsilon))\overset{d}{=}(X,Z).
\]
Applying the measurable map $(x,z)\mapsto(x,h(z))$ to both sides gives
\[
	(X,h(g(X,\varepsilon)))\overset{d}{=}(X,h(Z))\overset{d}{=}(X,Y),
\]
where the last equality is the defining property of the lift. No continuity of $h$ is used.

\subsection{Proof of Theorem~\ref{thm:discrete_continuous}}
We construct a regular positive latent lift through $h$ and then apply Theorem~\ref{thm:lift_criterion}. Recall that $A_y=h^{-1}(\{y\})$. Let $\phi_k$ denote the density of $\mathcal N(0,I_k)$ and define
\[
	\alpha_y=\int_{A_y}\phi_k(z)\,dz.
\]
Because $A_y$ has nonempty interior, $\alpha_y>0$. Moreover, the measurable cells $\{A_y:y\in\Y\}$ partition $\mathbb{R}^k$, and hence
\begin{equation}
	\sum_{y\in\Y}\alpha_y=1.
	\label{eq:discrete_alpha_sum}
\end{equation}

Define
\begin{equation}
	\rho(x)=\frac12\min_{v\in\Y}\frac{p_v(x)}{\alpha_v},\qquad\omega_y(x)=p_y(x)-\rho(x)\alpha_y.
	\label{eq:discrete_rho_omega}
\end{equation}
Since $\Y$ is finite and every $p_y$ is continuous and strictly positive, $\rho$ is continuous and strictly positive. Furthermore,
\[
	\omega_y(x)\geq\frac12p_y(x)>0,
\]
so every $\omega_y$ is continuous and strictly positive.

For each $y\in\Y$, choose $c_y\in\mathbb{R}^k$ and $\delta_y>0$ such that
\[
	\prod_{j=1}^k[c_{y,j}-\delta_y,c_{y,j}+\delta_y]\subset A_y.
\]
Let $s(t)=(1-|t|)_+$ and define
\[
	b_y(z)=\frac{1}{\delta_y^k}\prod_{j=1}^k s\left(\frac{z_j-c_{y,j}}{\delta_y}\right).
\]
Since $\int_{\mathbb{R}}s(t)\,dt=1$, $b_y$ is a continuous probability density supported inside $A_y$. The cells $A_y$ are disjoint, and therefore
\begin{equation}
	\int_{A_y}b_v(z)\,dz=\mathbf{1}\{y=v\}.
	\label{eq:discrete_cell_indicator}
\end{equation}

Define
\[
	a(x,z)=\phi_k(z)\rho(x)+\sum_{v\in\Y}b_v(z)\omega_v(x).
\]
The function $a$ is continuous and is strictly positive because $\phi_k(z)\rho(x)>0$. By \eqref{eq:discrete_cell_indicator} and \eqref{eq:discrete_rho_omega}, for every $y\in\Y$,
\begin{equation}
	\begin{aligned}
		\int_{A_y}a(x,z)\,dz=\rho(x)\alpha_y+\omega_y(x)=p_y(x).
	\end{aligned}
	\label{eq:discrete_cell_calibration}
\end{equation}

We next verify the marginal regularity required by Condition~\ref{cond:regular}. For $0\leq j\leq k$, let $\phi_{\leq j}$ and $b_{v,\leq j}$ denote the marginals of $\phi_k$ and $b_v$ over their first $j$ coordinates, with both quantities defined to be one when $j=0$. Direct integration gives
\begin{equation}
	a_{\leq j}(x,z_{\leq j})=\rho(x)\phi_{\leq j}(z_{\leq j})+\sum_{v\in\Y}\omega_v(x)b_{v,\leq j}(z_{\leq j}).
	\label{eq:discrete_sequential_marginals}
\end{equation}
Every marginal in \eqref{eq:discrete_sequential_marginals} is finite and continuous. When $j=0$, \eqref{eq:discrete_alpha_sum} gives
\begin{equation}
	\int_{\mathbb{R}^k}a(x,z)\,dz=\rho(x)+\sum_{v\in\Y}\omega_v(x)=\sum_{v\in\Y}p_v(x).
	\label{eq:discrete_normalization}
\end{equation}
Integrating \eqref{eq:discrete_normalization} with respect to $\nu(dx)$ shows that $a$ has total mass one. Hence $a$ satisfies Condition~\ref{cond:regular}, with $z$ in place of $y$. Let $(X,Z)$ have joint density $a$.

Finally, for every measurable $B\subseteq\X$ and every $y\in\Y$, \eqref{eq:discrete_cell_calibration} yields
\[
	\begin{aligned}
		\mathbb{P}\{X\in B,h(Z)=y\}&=\int_B\int_{A_y}a(x,z)\,dz\,d\nu(x)\\
		&=\int_Bp_y(x)\,d\nu(x)=\mathbb{P}(X\in B,Y=y).
	\end{aligned}
\]
Thus $a$ is a regular positive latent lift of $(X,Y)$ through $h$. Applying Theorem~\ref{thm:lift_criterion} gives a continuous $g$ with $\left(X,h(g(X,\varepsilon))\right)\overset{d}{=}(X,Y)$, which proves Theorem~\ref{thm:discrete_continuous}.

\subsection{Proof of Theorem~\ref{thm:mixed}}
Take independent Gaussian vectors $\varepsilon_c\sim\mathcal N(0,I_{k_1})$ and $\varepsilon_d\sim\mathcal N(0,I_{k_2})$, independent of $(X,W,D)$.

\smallskip\noindent\textit{Step 1: generate the continuous outcome.} By the first part of Condition~\ref{cond:mixed}, we apply Theorem~\ref{thm:continuous} and establish that there exists a continuous map $g_c:\X\times\mathbb{R}^{k_1}\to\mathbb{R}^{k_1}$ such that $\varepsilon_c\sim\mathcal N(0,I_{k_1})$ is independent of $X$ and
\[
	\widetilde W=g_c(X,\varepsilon_c),\qquad(X,\widetilde W)\overset d=(X,W).
\]

\smallskip\noindent\textit{Step 2: construct a generator for the discrete outcome.} Apply Theorem~\ref{thm:discrete_continuous} with $(X,W)$ as the covariate and $D$ as the response. The enlarged covariate space is $\X\times\mathbb{R}^{k_1}$, equipped with the product topology and reference measure $\nu(dx)\,dw$. Its joint class subdensities are $p_d(x,w)$. Thus, the last two parts of Condition~\ref{cond:mixed} verify Condition~\ref{cond:interior}.

The discrete representation theorem therefore provides a continuous map $g_d:\X\times\mathbb{R}^{k_1}\times\mathbb{R}^{k_2}\to\mathbb{R}^{k_2}$ such that
\[
	\left(X,W,h_D(g_d(X,W,\varepsilon_d))\right)\overset d=(X,W,D).
\]

\smallskip\noindent\textit{Step 3: combine the two generators.} Use the generated continuous outcome as an input to $g_d$:
\[
	\widetilde Z=g_d(X,\widetilde W,\varepsilon_d),\qquad\widetilde D=h_D(\widetilde Z).
\]
Define the combined generator by
\[
	g(x,e_c,e_d)=\left(g_c(x,e_c),g_d(x,g_c(x,e_c),e_d)\right).
\]
This map is continuous because $g_c$ and $g_d$ are continuous. For $\varepsilon=(\varepsilon_c,\varepsilon_d)$, recalling $\widetilde W=g_c(X,\varepsilon_c)$, we have
\[
	h(g(X,\varepsilon))=\left(\widetilde W,h_D(g_d(X,\widetilde W,\varepsilon_d))\right).
\]
Thus, the combined generator first produces $\widetilde W$ and then uses it as an input to the discrete generator.

We now verify that this procedure reproduces the desired joint distribution. Step 1 gives $(X,\widetilde W)\overset d=(X,W)$. Since $\varepsilon_d$ is independent of both pairs,
\[
	(X,\widetilde W,\varepsilon_d)\overset d=(X,W,\varepsilon_d).
\]
Consequently, evaluating the same function $h_D\circ g_d$ at these inputs gives
\[
	\begin{aligned}
		\left(X,h(g(X,\varepsilon))\right)&=\left(X,\left(\widetilde W,h_D(g_d(X,\widetilde W,\varepsilon_d))\right)\right)\\
		&\overset d=\left(X,\left(W,h_D(g_d(X,W,\varepsilon_d))\right)\right)\overset d=(X,(W,D))=(X,Y).
	\end{aligned}
\]
The first equality in distribution follows from the equality of the augmented input laws above; the second is the discrete representation established in Step 2. Finally, $\varepsilon\sim\mathcal N(0,I_k)$ is independent of $X$, which proves the required representation.

\section{Settings of GEM and details for the simulation study}\label{app:simulation}

\subsection{Implementation of GEM}\label{app:gem_settings}

\begin{table}[!h]
	\centering
	\caption{Settings of GEM in the five experiments; the common settings are in the text. $d$ and $k$ are the dimensions of the covariates and of the coded outcome.}
	\label{tab:gem_settings}
	\footnotesize\setlength{\tabcolsep}{3pt}
	\begin{tabular}{@{}llllll@{}}
		\toprule
		 & multi-label & mixed-type & ranking & species & NHANES\\
		\midrule
		links & sign & identity, sign, argmax & ranking & sign & identity, sign, ordinal\\
		$d$, $k$ & $5$, $10$ & $5$, $31$ & $5$, $4$ & $4$, $242$ & $29$, $13$\\
		\#layers of $g$ & $4$ & $4$ & $4$ & $4$ & $4$\\
		width of $g$ & $512$ & $512$ & $512$ & $256$ & $512$\\
		scale $\sigma$ & common & common & common & common & coordinate-wise\\
		refinement (Section~\ref{sec:algorithm}) & none & pathwise, $g$ & none & control variate & pathwise, $g$ and $\sigma$\\
		learning rate & $10^{-4}$ & $10^{-4}$ & $10^{-4}$ & $10^{-3}$, cosine decay & $10^{-4}$\\
		iterations & $5{,}000$ & $5{,}000$ & $5{,}000$ & $4{,}000$ & $5{,}000$\\
		validation sample & $5{,}000$ & $10{,}000$ & $5{,}000$ & none, last iterate & $5{,}596$\\
		\bottomrule
	\end{tabular}
\end{table}

In every experiment of Sections~\ref{sec:simulation} and~\ref{sec:applications}, GEM is fitted by Algorithm~\ref{alg:gem} with the settings of Table~\ref{tab:gem_settings} and the following common choices. The generator $g$ is the network \texttt{StoNet} of the \texttt{engression} package \citep{shen2024engression}, with the number of layers and the width given in Table~\ref{tab:gem_settings}. As in the package, the number of layers counts the hidden layers and the linear output layer, so it is the number of hidden layers plus one. The $4$ layers of $g$ are three hidden layers of ReLU units followed by the output layer. A fresh draw of $\mathcal N(0,I_k)$ is concatenated to the input of each hidden layer. The scale $\sigma(x)$ is the softplus of a ReLU network of $2$ layers and width $100$. Neither network uses normalization, residual connections, dropout or weight decay. The outcome is coded as in Table~\ref{tab:links}, and the objective \eqref{eq:gem_objective}, with the Euclidean norm of the coded vector, is estimated with two draws per observation on the whole training sample at every iteration and minimized by Adam. Except on the species data, the energy score on the validation sample of Table~\ref{tab:gem_settings}, estimated with one pair of draws per point, is computed every $50$ iterations and after the last, and the parameters with the lowest value are kept. Apart from this choice of the iterate, GEM is not tuned in any simulation or on the NHANES data, and it uses these default hyperparameters throughout, at every sample size and in every setting. All methods, the oracles included, are scored from their predictive draws. Figures~\ref{fig:sim62_metrics}, \ref{fig:sim63_mixed}, \ref{fig:sim64_metrics} and~\ref{fig:sim62_small_n} show the mean and a band of one standard error over the $50$ replications. GEM was fitted with Python 3.13 and PyTorch 2.12 on a server with four NVIDIA RTX 4090 GPUs and two $16$-core AMD EPYC 7313 processors, with two to five fits sharing a GPU. All times are wall-clock times. In the multi-label and ranking simulations, a fit took about $20$ to $95$ seconds, depending on $n$.

\subsection{Multi-label outcome}\label{app:multilabel}

\paragraph{Design.} The coefficients $\beta_0,\dots,\beta_{10}$ have independent $\mathcal N(0,0.6^2/5)$ entries and $\gamma_0,\dots,\gamma_{10}$ independent $\mathcal N(0,0.5^2/5)$ entries; like all the coefficients of the simulations, they are drawn once and held fixed across settings, sample sizes and replications. The constant $c=\{(1-e^{-8\|\gamma_0\|^2})/(4\|\gamma_0\|^4)\}^{1/2}=1.28$ gives the quadratic term of $f_0$ the variance of $\sin(2\gamma_0^\top X)$. While $\Var(Z_j\mid X)=1$ at every $\rho$, the variance of $f_j^\rho(X)$, averaged over the labels, rises from $0.59$ at $\rho=0.1$ to $0.91$ at $\rho=1$. Each of the $50$ replications has a training sample of each size, a validation sample of $5{,}000$ and a test sample of $2{,}000$, with the same random seeds at every $\rho$ and $n$.

\paragraph{Baselines.} Linear logistic is one logistic regression per label on standardized covariates, with the default ridge penalty ($C=1$) of \texttt{scikit-learn} 1.9. NN logistic is one network $5\to256\to256\to256\to10$ with ReLU activations and sigmoid outputs, trained by the binary cross-entropy with Adam at learning rate $10^{-4}$ on the full batch for $5{,}000$ iterations, keeping the parameters with the lowest validation cross-entropy, computed every $50$ iterations and after the last. Both draw the labels independently from their fitted probabilities, and the oracle draws from the true mechanism.

\paragraph{Metrics.} The joint energy score, the marginal CRPS and the cross-entropy use $M=500$ draws at each of the $2{,}000$ test points, in the $\pm1$ coding. The energy score is $M^{-1}\sum_{m=1}^M\|\hat Y^{(m)}-y\|$ minus half the mean of $\|\hat Y^{(j)}-\hat Y^{(l)}\|$ over $250$ disjoint random pairs of draws. The CRPS is the same quantity for each label, averaged over the labels. The cross-entropy uses the smoothed frequency $(M_j+1)/(M+2)$ of the $M_j$ draws with $\hat Y_j=1$, and is averaged over the labels. The co-occurrence error is computed at the first $300$ of the $2{,}000$ test points from $2{,}000$ draws of each method, and $\P(Y_j=Y_l\mid X=x)$ is estimated from $2{,}000$ draws of the true mechanism. Since both are estimated from finitely many draws, the oracle's error is not zero but $0.008$ to $0.012$; it is exactly $0$ at $\rho=1$, where all labels are equal. Figure~\ref{fig:sim62_counts} uses the fits at $n=10{,}000$ of the first replication, at the first point $x_0$ of its test sample, with $4{,}000$ draws of each method and of the true mechanism.

\paragraph{Results at $n=100$ and $n=1{,}000$.} Figure~\ref{fig:sim62_small_n} shows the four metrics at the two smaller sample sizes. At $n=1{,}000$ the results are similar to those in Figure~\ref{fig:sim62_metrics}, with every fitted method further from the oracle. GEM is slightly above NN logistic on the marginal metrics and has the lowest joint energy score for $\rho\ge0.3$ and the lowest co-occurrence error for $\rho\ge0.2$. At $n=100$ all fitted methods are far from the oracle on the marginal metrics, with GEM slightly the worst of them at small $\rho$. GEM has the lowest joint energy score for $\rho\ge0.5$ and the lowest co-occurrence error for $\rho\ge0.3$. At $\rho=1$ both are much lower than for the baselines ($1.35$ against $1.84$ for NN logistic and $1.91$ for linear logistic, and $0.08$ against $0.39$ and $0.40$).

\begin{figure}
	\centering
	\begin{tabular}{@{}c@{}c@{}c@{}c@{}}
		\includegraphics[width=0.24\textwidth]{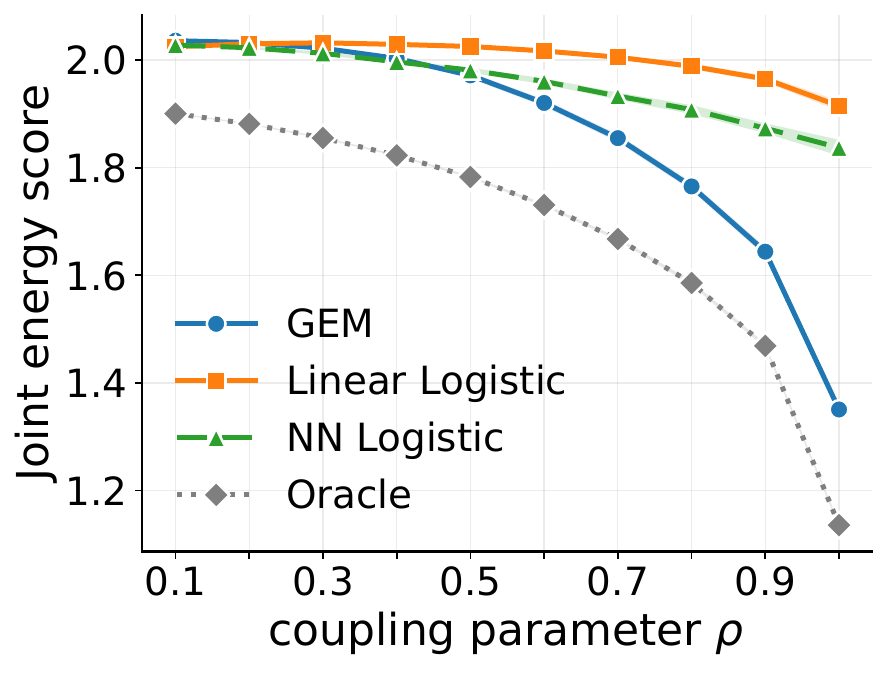} &
		\includegraphics[width=0.24\textwidth]{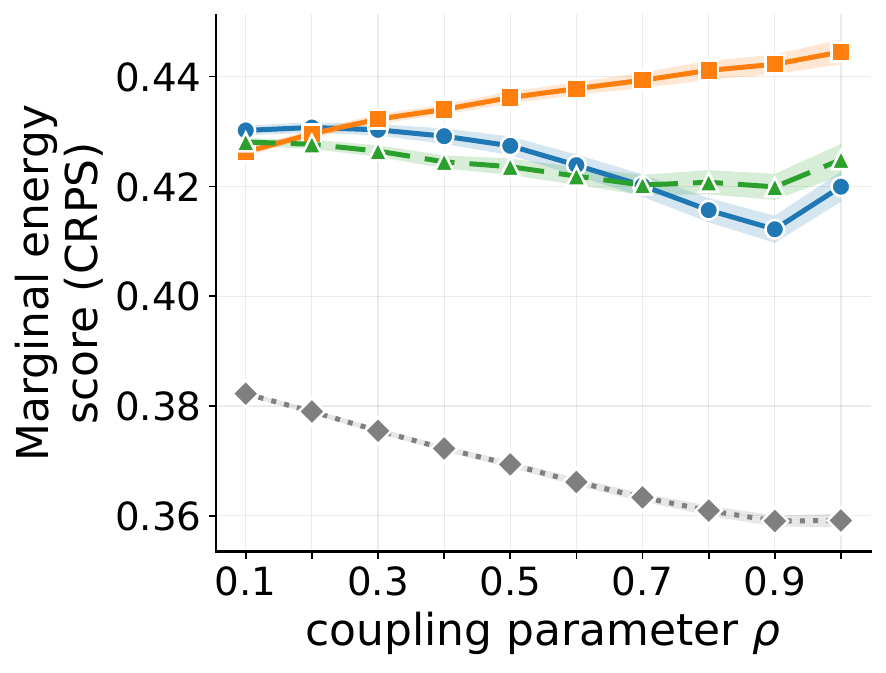} &
		\includegraphics[width=0.24\textwidth]{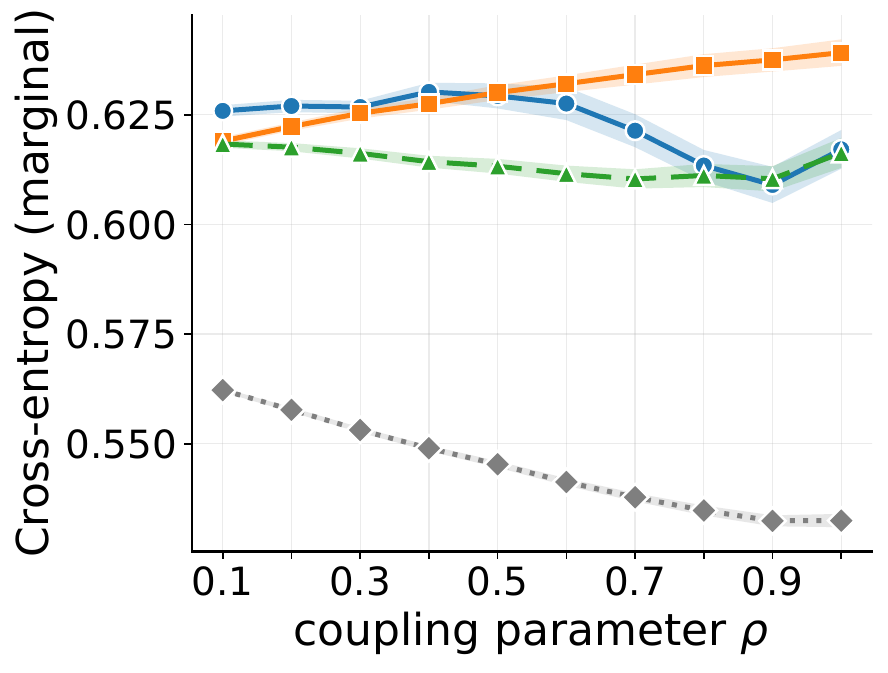} &
		\includegraphics[width=0.24\textwidth]{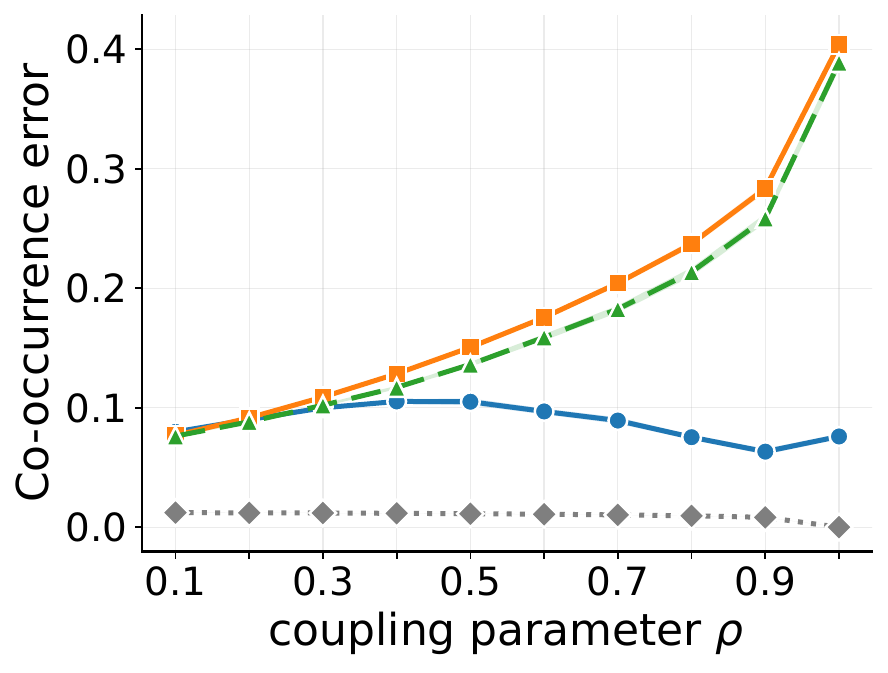}\\
		\includegraphics[width=0.24\textwidth]{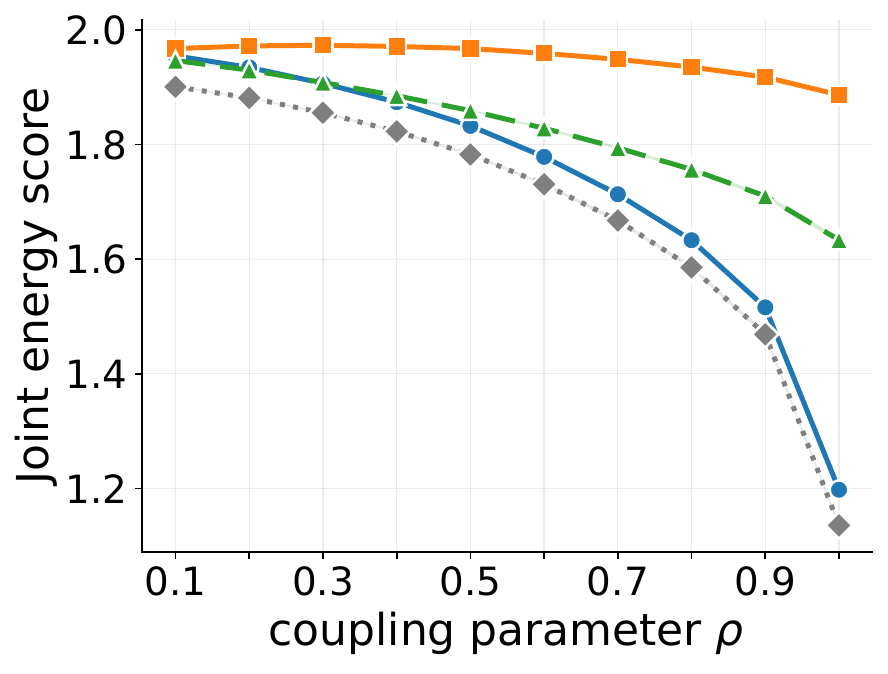} &
		\includegraphics[width=0.24\textwidth]{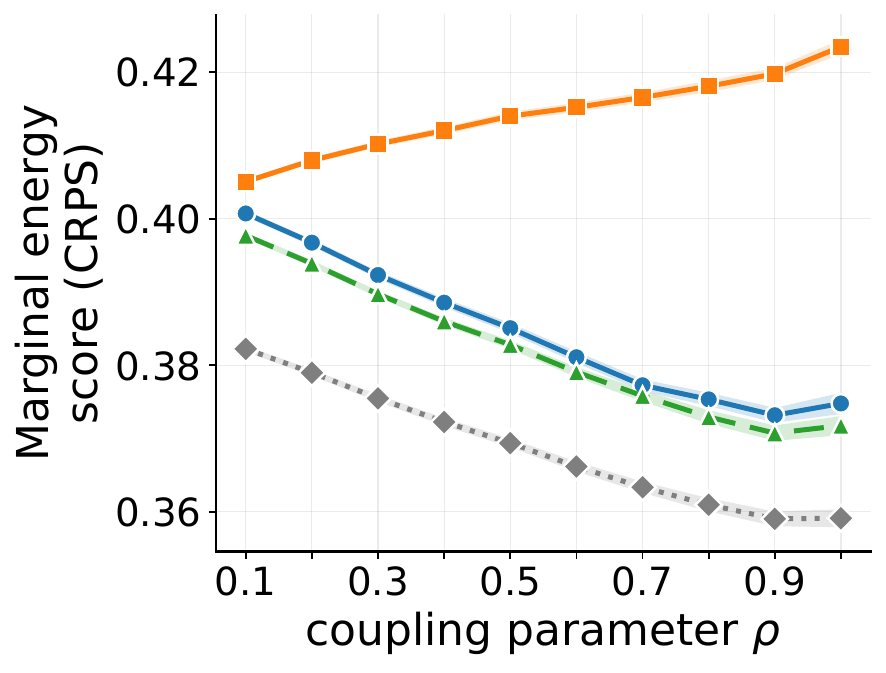} &
		\includegraphics[width=0.24\textwidth]{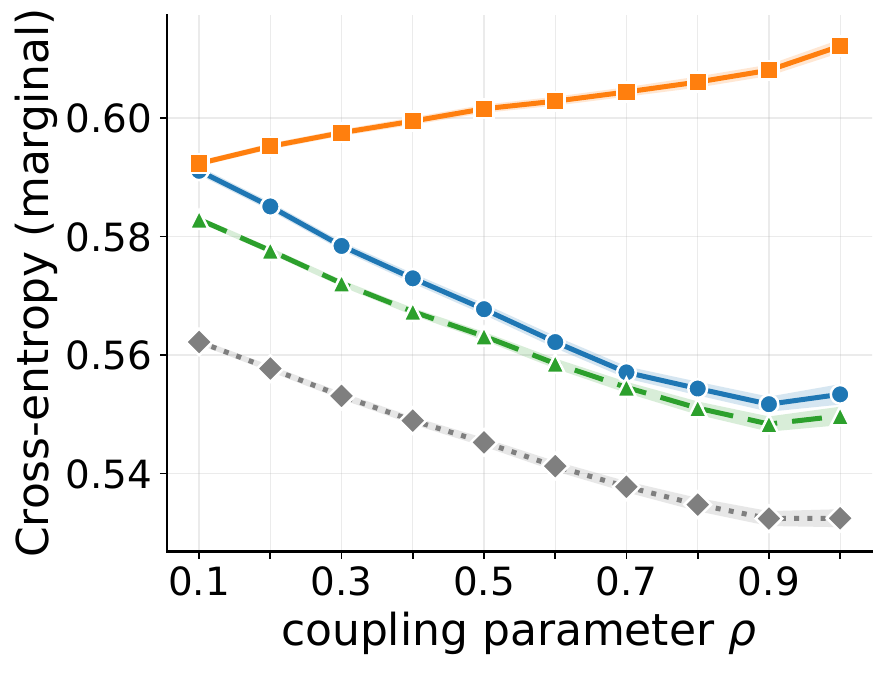} &
		\includegraphics[width=0.24\textwidth]{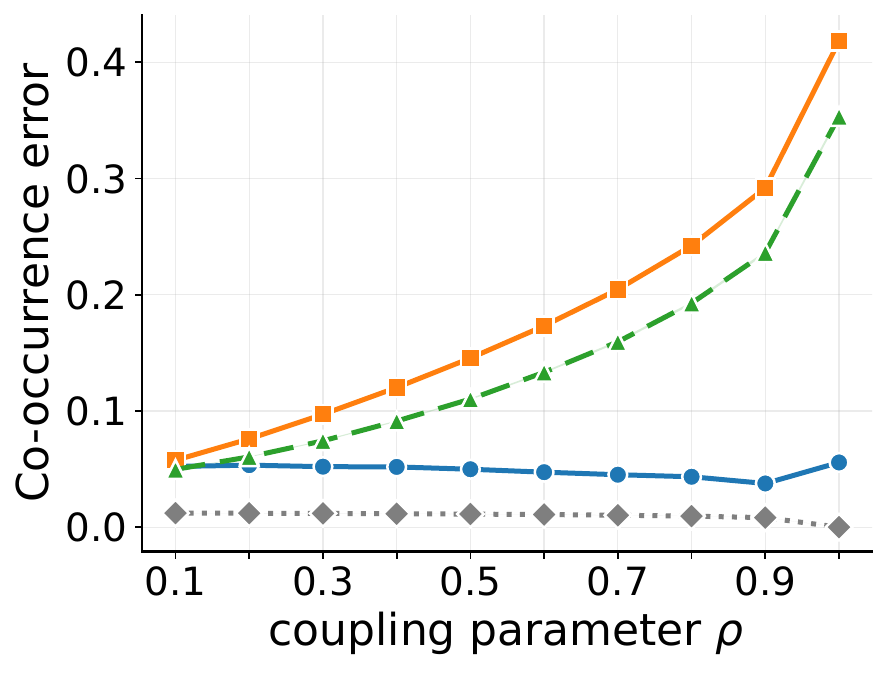}
	\end{tabular}
	\caption{Metrics for multiple binary labels at $n=100$ (top) and $n=1{,}000$ (bottom).}
	\label{fig:sim62_small_n}
\end{figure}

\subsection{Mixed-type outcome}\label{app:mixed}

\paragraph{Design.} The covariates are $X\sim\mathcal N(0,I_5)$, the latent vector $Z$ has $p=31$ coordinates, $B\in\rb^{31\times5}$ has independent $\mathcal N(0,0.1^2)$ entries, and $m$ acts on every coordinate of $BX$. The covariance is built from the loadings $A\in\rb^{31\times5}$ of five factors, $A_{if}=a_f+0.18\,z_{if}$ on the continuous and binary coordinates and $A_{if}=a_f+0.9\,z_{if}$ on the categorical ones, with independent $a_f\sim\mathcal N(0,0.9^2)$ and $z_{if}\sim\mathcal N(0,1)$. To $AA^\top$ are added a diagonal matrix with entries uniform on $[0.05,0.2]$ and, on each entry linking two types, a term of random sign and magnitude uniform on $[0.5,0.9]$ between the continuous and the binary coordinates and on $[0.4,0.8]$ otherwise. The sum is then shifted by $4.03\,I$ to make its smallest eigenvalue $0.01$. The resulting latent correlations are all positive: $0.51$ to $0.69$ between the continuous and the binary coordinates, $0.58$ to $0.67$ among the binary ones, and $0.10$ to $0.75$ in pairs involving a categorical coordinate. Each of the $50$ replications has a validation sample and a test sample of $10{,}000$, the same at every $n$.

\paragraph{Methods.} GEM applies the identity link to the first coordinate, the sign link to the next $20$ and the argmax link to the last $10$, with one scale for all $31$ coordinates. Each component of the stacked pipeline is trained by Adam at learning rate $10^{-4}$ on the full batch for $5{,}000$ iterations, keeping the parameters with the lowest validation loss, computed every $50$ iterations and after the last. The continuous coordinate is fitted by the model \eqref{eq:gem_model} with the identity link, that is engression with the added perturbation $\sigma(x)\eta$, with $g$ as in Appendix~\ref{app:gem_settings} at width $256$ except that a one-dimensional $\varepsilon$ enters at the input only. Each binary label is fitted by its own network $5\to256\to256\to256\to1$, trained by the binary cross-entropy on standardized covariates, and the categorical label by a network $5\to256\to256\to256\to10$ with a softmax output, trained by the cross-entropy. Its draws are independent across the three blocks and across the $20$ labels.

\paragraph{Metrics.} All metrics use $M=1{,}000$ draws at each of the $10{,}000$ test points. The energy score is computed on the vector of the continuous value, the $20$ labels in $\pm1$ and the indicator vector of the class, with its second term averaged over $25$ disjoint pairs from a random subset of $50$ draws. The CRPS is the same for the continuous coordinate. Each cross-entropy is the mean over the test points of $-\log\hat p$, where $\hat p$ estimates the probability of the observed value from the number $M_y$ of draws equal to it. For each binary label, $\hat p=(M_y+1)/(M+2)$, the smoothed frequency of Appendix~\ref{app:multilabel}, as if one draw of each of the two values were added, so that $\hat p$ is never $0$; the binary cross-entropy is averaged over the labels. For the categorical label, $\hat p=\max(M_y/M,10^{-8})$, the unsmoothed frequency of the observed class with a floor of $10^{-8}$, so that a class that none of the draws takes costs $-\log10^{-8}\approx18.4$.

The two cross-type errors measure whether a method reproduces the dependence across types: given $X$, the continuous coordinate depends on the binary labels and on the class through the latent correlations. We measure this dependence by correlation ratios. At each test point $x$, the residual of a draw is its continuous value minus the mean of the continuous values of the $M$ draws at $x$, so that only the variation given $x$ remains. The residuals of all draws at all test points are pooled and split into two groups by the value of label $l$ in the same draw. With $N_g$ residuals of mean $\bar r_g$ in group $g$ and overall mean $\bar r$, the correlation ratio $\hat\eta_l\in[0,1]$ is given by $\hat\eta_l^2=\sum_gN_g(\bar r_g-\bar r)^2/\sum(r-\bar r)^2$, the fraction of the variance of the residuals explained by the label. The ratio $\hat\eta_{\mathrm{cat}}$ is the same with the ten groups formed by the drawn class. The true ratios, $\eta_l$ from $0.39$ to $0.54$ and $\eta_{\mathrm{cat}}=0.43$, are computed in the same way from $2\times10^5$ draws of $(X,Y)$ from the mechanism, with the exact residual $Y^{\mathrm{cont}}-\E(Y^{\mathrm{cont}}\mid X)$. The two errors in Figure~\ref{fig:sim63_mixed} are the root mean square of $\hat\eta_l-\eta_l$ over the $20$ labels and $|\hat\eta_{\mathrm{cat}}-\eta_{\mathrm{cat}}|$. The stacked marginals draw the three blocks independently given $x$, so their ratios are near $0$ and their errors near $0.48$ and $0.43$.

\subsection{Ranking outcome}\label{app:ranking}

\paragraph{Design.} The coefficients $\beta_j$ and $\gamma_j$ have independent $\mathcal N(0,0.6^2/5)$ and $\mathcal N(0,0.5^2/5)$ entries; those of item $1$ are not used, since $f_1=f_0+\delta$. Each of the $50$ replications has a validation sample of $5{,}000$ and a test sample of $2{,}000$, the same at every $n$.

\paragraph{Baselines.} The Plackett--Luce model with linear scores $s_j(x)=w_j^\top x+b_j$ is fitted by maximum likelihood. For an observed ordering $i_1\succ\dots\succ i_4$ at covariate value $x$, the log-likelihood is $\sum_{t=1}^{4}\{s_{i_t}(x)-\log\sum_{u\ge t}\exp s_{i_u}(x)\}$. It is maximized by Adam at learning rate $10^{-3}$ on the full batch for $5{,}000$ iterations; the validation negative log-likelihood is computed every $50$ iterations and after the last, and the parameters with the lowest value are kept. The model with spline scores has $s_j(x)=b_j+\sum_{l=1}^5s_{jl}(x_l)$, where each $s_{jl}$ is a natural cubic spline with knots at quantiles of the training values of $x_l$. The number of knots, $3$, $5$ or $7$ and the same for all covariates, is chosen for each fit by the validation likelihood, and the model is fitted in the same way. Relative to linear scores in the same replication, spline scores change the total variation by $+0.014$, $+0.005$ and $-0.002$ at $n=500$, $2{,}000$ and $10{,}000$ in the independent setting, and by $-0.003$ to $-0.004$ in the coupled setting, with standard errors of at most $0.001$. Both models draw rankings exactly, by adding independent standard Gumbel variables to the fitted scores and ranking the items by the perturbed scores.

\paragraph{Metrics.} At each test point $x$, the conditional distribution of a method over the $24$ rankings is estimated from $M$ of its draws as $\hat P(\pi)=(M_\pi+1)/(M+24)$, where $M_\pi$ is the number of draws equal to the ranking $\pi$; adding one to each count keeps every estimate positive. The oracle's distribution is estimated in the same way. The total variation $\frac12\sum_\pi|\hat P(\pi)-P(\pi)|$ and the pairwise-preference error are computed at the first $200$ of the $2{,}000$ test points, with $M=4{,}000$. At each of these points, the true distribution $P$ is estimated by the frequencies of $4\times10^4$ draws from the true mechanism, without adding one. Since both $\hat P$ and $P$ are estimated from finitely many draws, the oracle's errors are not zero. Its total variation is $0.026$ in the independent setting and $0.021$ in the coupled setting, and its pairwise-preference error is $0.006$ and $0.005$. The log-loss $-\log\hat P(y)$ of the observed ranking $y$ and Kendall's $\tau$ between $y$ and the most probable ranking under $\hat P$ are computed at all $2{,}000$ test points, with $M=500$. For the Plackett--Luce models, the log-loss is also computed from $\hat P$, that is, from their draws rather than from their closed-form probabilities, so that all methods are scored by the same estimator.

\section{Details for the joint species distribution modeling application}\label{app:jsdm}

\subsection{Data and splits}

The plots are those distributed with the benchmark, $1{,}200$ in all, within an area about $15$ by $11$ kilometers. The benchmark reduces the six environmental variables to four principal components that carry $88\%$ of their variance. The five splits are the benchmark's own random half plus four further random halves of the same pool, so any two of the five share between $290$ and $316$ of their $600$ training plots and as many test plots. Their per-plot score differences correlate at about $0.7$ to $0.9$, depending on the score, so we report the number of splits won rather than paired $t$-statistics.

\paragraph{Dependence between species.} The richness variances of Section~\ref{sec:jsdm} are computed on the test plots of the benchmark's split from logistic regressions fitted species by species to the components and their squares. The observed richness varies around the sum of the fitted probabilities with a variance of $148$, whereas independent species with those probabilities would give $17$ on average. Little of the dependence between species is spatial. Across the training plots, the residual of a species present on at least $10\%$ of them correlates at $0.13$ on average with its mean residual at the eight nearest other plots.

\subsection{Models and settings}

\paragraph{Baselines.} Table~\ref{tab:veg_interp_full} gives all $17$ models of the comparison with every measure we computed. Each baseline is fitted once per split, with the benchmark's scripts or the authors' packages, and gives $1{,}000$ draws of the $242$-species vector per test plot. Nine of the $15$ are the benchmark's own models: the probit GLM, the GAM, MARS, the random forest, the boosted trees, the gradient nearest-neighbor model, and HMSC without associations, with plot-level latent factors and with spatially structured latent factors. The GAM, the random forest and the boosted trees with the coordinates, gllvm and the two classifier chains have no counterpart there and follow the same protocol. As in the benchmark, the models with a linear predictor (the GLM, HMSC, gllvm and the chains) receive the four principal components and their squares, and the flexible models the components alone, so GEM sees fewer columns than the models with a linear predictor. No model uses the plot coordinates except the five spatial variants. The spatial variants use the coordinates through a two-dimensional smooth (GAM), as covariates (trees) or as covariates together with a fixed basis of $64$ Gaussian bumps at $k$-means centers of the training plots (chain). The settings are the following: the probit GLM with linear and squared terms and an intercept-only fallback on separation, as in the benchmark; the GAM with \texttt{mgcv} shrinkage smooths selected by REML; MARS as a single multi-response \texttt{earth} model of degree one; a random forest with \texttt{mtry} and node size chosen by out-of-bag error; boosted trees with interaction depth and learning rate chosen by stratified five-fold cross-validation; and a gradient nearest-neighbor model. HMSC uses default priors and four chains of $250$ samples, thinned by $20$ after a transient of $5{,}000$ without associations and with plot-level latent factors and by $10$ after $2{,}500$ with spatially structured latent factors through a nearest-neighbor Gaussian process, with median potential scale reduction factors of $1.002$ to $1.016$. For gllvm, two latent variables are fitted by variational approximation, with no convergence warning on any reported fit. The classifier chain consists of ten logistic chains in random species orders, whose draws are generated sequentially. MARS and the gradient nearest-neighbor model share a basis and a neighbor set across the species, but their draws, like those of the other models without a dependence structure, are independent across species. At a test plot, HMSC draws the plot-level latent factors from their prior and the spatially structured ones from their conditional law given the training plots. HMSC with spatial latent factors, the benchmark's best model on these data, takes $3.9$ hours per fit and was run on the benchmark's own split only.

\begin{table}[t]
	\centering
	\caption{Vegetation data: every model of the comparison, mean (standard deviation) over the five splits (the benchmark's split only for HMSC with spatial latent factors), sorted by the energy score. Columns as in Table~\ref{tab:veg_main} except the fit time, with in addition whether the model uses the plot coordinates, the dependence gain and the correlation between predicted and observed residual co-occurrence.}
	\label{tab:veg_interp_full}
	\footnotesize\setlength{\tabcolsep}{2pt}
	{%
	\begin{tabular}{llccccccccccc}
		model & coord. & energy$\scriptstyle\downarrow$ & \begin{tabular}[b]{@{}c@{}}rich.\\CRPS$\scriptstyle\downarrow$\end{tabular} & \begin{tabular}[b]{@{}c@{}}rich.\\int.$\scriptstyle\downarrow$\end{tabular} & \begin{tabular}[b]{@{}c@{}}$\beta$\\CRPS$\scriptstyle\downarrow$\end{tabular} & \begin{tabular}[b]{@{}c@{}}$\beta$\\int.$\scriptstyle\downarrow$\end{tabular} & NLL$\scriptstyle\downarrow$ & AUC$\scriptstyle\uparrow$ & \begin{tabular}[b]{@{}c@{}}dep.\\gain$\scriptstyle\uparrow$\end{tabular} & \begin{tabular}[b]{@{}c@{}}co-\\occ.$\scriptstyle\uparrow$\end{tabular} & overall$\scriptstyle\uparrow$\\
		\hline
		GEM, 5-seed ensemble & no & 2.782 & 6.17 & 0.031 & 0.091 & 0.036 & 56.1 & 0.738 & 0.0078 & 0.81 & 1.09\\
		 &  & (0.038) & (0.16) & (0.022) & (0.003) & (0.028) & (1.2) & (0.004) &  &  & (0.06)\\[1pt]
		Random forest + coord. & yes & 2.792 & 7.33 & 0.301 & 0.105 & 0.290 & 56.6 & 0.734 & $-$0.0006 & $-$0.01 & 0.24\\
		 &  & (0.028) & (0.22) & (0.010) & (0.002) & (0.033) & (1.1) & (0.007) &  &  & (0.14)\\[1pt]
		GEM, single fit & no & 2.795 & 6.23 & 0.039 & 0.091 & 0.039 & 57.3 & 0.727 & 0.0056 & 0.77 & 0.91\\
		 &  & (0.039) & (0.17) & (0.019) & (0.003) & (0.020) & (1.4) & (0.009) &  &  & (0.13)\\[1pt]
		Boosted trees + coord. & yes & 2.813 & 7.53 & 0.297 & 0.110 & 0.310 & 56.8 & 0.697 & $-$0.0006 & $-$0.01 & $-$0.17\\
		 &  & (0.030) & (0.25) & (0.019) & (0.003) & (0.028) & (1.2) & (0.009) &  &  & (0.19)\\[1pt]
		HMSC, spatial latent f. & yes & 2.815 & 6.49 & 0.074 & 0.095 & 0.093 & 56.7 & 0.731 & 0.0106 & 0.90 & 0.66\\
		Classifier chain + basis & yes & 2.816 & 6.35 & 0.042 & 0.093 & 0.073 & 59.5 & 0.735 & 0.0096 & 0.90 & 1.01\\
		 &  & (0.028) & (0.14) & (0.023) & (0.002) & (0.033) & (1.3) & (0.007) &  &  & (0.13)\\[1pt]
		GAM + spatial smooth & yes & 2.817 & 7.42 & 0.286 & 0.112 & 0.309 & 56.5 & 0.738 & $-$0.0006 & $-$0.01 & 0.05\\
		 &  & (0.032) & (0.18) & (0.021) & (0.004) & (0.037) & (1.1) & (0.007) &  &  & (0.15)\\[1pt]
		Classifier chain & no & 2.819 & 6.63 & 0.021 & 0.097 & 0.079 & 56.9 & 0.735 & 0.0123 & 0.92 & 0.44\\
		 &  & (0.030) & (0.11) & (0.010) & (0.003) & (0.036) & (1.2) & (0.005) &  &  & (0.10)\\[1pt]
		GAM & no & 2.837 & 7.46 & 0.293 & 0.113 & 0.314 & 57.1 & 0.726 & $-$0.0006 & $-$0.01 & $-$0.32\\
		 &  & (0.030) & (0.19) & (0.020) & (0.004) & (0.041) & (1.1) & (0.006) &  &  & (0.10)\\[1pt]
		HMSC, covariates only & no & 2.837 & 7.48 & 0.297 & 0.113 & 0.305 & 56.7 & 0.745 & $-$0.0006 & $-$0.01 & $-$0.20\\
		 &  & (0.029) & (0.22) & (0.017) & (0.004) & (0.034) & (1.0) & (0.006) &  &  & (0.08)\\[1pt]
		Boosted trees & no & 2.842 & 7.62 & 0.301 & 0.112 & 0.311 & 57.7 & 0.687 & $-$0.0005 & $-$0.02 & $-$0.56\\
		 &  & (0.028) & (0.20) & (0.011) & (0.003) & (0.029) & (1.1) & (0.007) &  &  & (0.14)\\[1pt]
		Probit GLM & no & 2.847 & 7.56 & 0.303 & 0.113 & 0.312 & 58.0 & 0.733 & $-$0.0006 & $-$0.02 & $-$0.35\\
		 &  & (0.029) & (0.22) & (0.012) & (0.004) & (0.039) & (1.2) & (0.009) &  &  & (0.11)\\[1pt]
		gllvm & no & 2.853 & 7.16 & 0.073 & 0.100 & 0.127 & 58.3 & 0.730 & 0.0088 & 0.80 & $-$0.16\\
		 &  & (0.031) & (0.04) & (0.021) & (0.003) & (0.036) & (1.3) & (0.009) &  &  & (0.28)\\[1pt]
		MARS & no & 2.853 & 7.60 & 0.297 & 0.115 & 0.306 & 57.7 & 0.731 & $-$0.0006 & $-$0.02 & $-$0.58\\
		 &  & (0.031) & (0.18) & (0.002) & (0.004) & (0.036) & (1.2) & (0.008) &  &  & (0.07)\\[1pt]
		Random forest & no & 2.856 & 7.62 & 0.309 & 0.109 & 0.304 & 59.4 & 0.709 & $-$0.0008 & $-$0.01 & $-$0.46\\
		 &  & (0.029) & (0.18) & (0.006) & (0.003) & (0.035) & (1.3) & (0.005) &  &  & (0.16)\\[1pt]
		Gradient nearest & no & 2.885 & 7.24 & 0.306 & 0.111 & 0.313 & 68.1 & 0.699 & $-$0.0004 & $-$0.01 & $-$0.23\\
		\quad neighbor &  & (0.036) & (0.12) & (0.007) & (0.003) & (0.031) & (1.9) & (0.003) &  &  & (0.07)\\[1pt]
		HMSC, latent factors & no & 2.894 & 7.27 & 0.061 & 0.106 & 0.117 & 59.3 & 0.694 & 0.0142 & 0.89 & $-$0.78\\
		 &  & (0.031) & (0.12) & (0.023) & (0.003) & (0.029) & (1.1) & (0.006) &  &  & (0.12)\\[1pt]
		\hline
	\end{tabular}
	}
\end{table}

\paragraph{GEM.} The settings are those of Table~\ref{tab:gem_settings} and Appendix~\ref{app:gem_settings}, with the four components, standardized on the training plots, as inputs. The learning rate is annealed from $10^{-3}$ to zero on a cosine schedule over the $4{,}000$ iterations, and the last iterate is used. Each fit gives $1{,}000$ draws per test plot, and the fit time of Table~\ref{tab:veg_main} includes them, with two fits sharing a GPU. The ensemble pools $200$ draws from each of the five fits, and the single-fit rows of the tables average each metric over the five fits.

\paragraph{Tuning.} The number of iterations and the learning rate were chosen on a $480/120$ division of the training half of the benchmark's own split; no test plot of that split was used. Each setting was fitted with one seed on the $480$ plots and scored by the energy score of $500$ draws on the $120$ held-out plots, with the marginal log score as tie-breaker. The grid crossed $4{,}000$, $8{,}000$, $16{,}000$ and $32{,}000$ iterations with learning rates of $10^{-4}$, $3\times10^{-4}$ and $10^{-3}$, and was extended by one step where the optimum lay on its edge, $14$ settings in all; it selected $4{,}000$ iterations at $10^{-3}$. The variant without the control variate in Table~\ref{tab:veg_gem} was tuned in the same way and selected $64{,}000$ iterations at $10^{-4}$. The selected settings were used unchanged on every split. The $120$ plots held out for tuning GEM come from the training half of the benchmark's own split, so they appear among the test plots of the four re-splits but not of the benchmark's split. On that split, the mean energy-score margin of GEM over the $14$ baselines fitted on every split is about $86\%$ of its mean margin on the re-splits.

\paragraph{Single fit and ensemble.} Only GEM was tuned, run with five seeds and ensembled; each baseline is a single fit at the benchmark's or its package's settings, although the random forest, the boosted trees and the GAM tune internally. The single-fit rows are therefore the like-for-like comparison. The single fit beats $13$ of the $15$ baselines on the energy score and $14$ on the richness CRPS on every split (Table~\ref{tab:veg_interp_full}), though its marginals trail those of the best single-species models. With $600$ training plots a single fit depends on its random seed, its energy score varying with a standard deviation of $0.007$ across seeds within a split (Table~\ref{tab:veg_gem}). For this reason the ensemble pools five fits.

\paragraph{The gradient estimator.} Table~\ref{tab:veg_gem} compares the estimator of Proposition~\ref{prop:grad} with and without the control variate of Section~\ref{sec:algorithm}. The measured variance reduction from the control variate is a factor of $242$ to $606$, from $4{,}000$ repetitions at each of three seeds, at initialization and after $1{,}000$ steps. The estimator stays unbiased, with the norm of the mean difference between the two estimators at $0.98$ to $1.08$ standard errors in all six cells of the study. With the control variate switched off, it reproduces the gradient of the unmodified algorithm to $1.2\times10^{-6}$ over the $382{,}939$ parameters of the network, which in this study took the components and their squares as inputs. Without the control variate, the same network, tuned in the same way, needs sixteen times as many iterations and still reaches an energy score of $2.916$ against $2.795$, an AUC of $0.65$ against $0.73$ and a richness interval error of $0.052$ against $0.039$ (Table~\ref{tab:veg_gem}).

\begin{table}[t]
	\centering
	\caption{GEM on the vegetation data: the estimator of Proposition~\ref{prop:grad} without and with the control variate. Mean (standard deviation) over the five splits; AUC rare and AUC common are over the species with training prevalence below $5\%$ and at least $10\%$. The seed SD is the standard deviation of the energy score across the five seeds within a split, averaged over splits. NLL is the marginal log score of Table~\ref{tab:veg_main}.}
	\label{tab:veg_gem}
	\footnotesize\setlength{\tabcolsep}{2pt}
	{%
	\begin{tabular}{lccccccccccc}
		variant & energy$\scriptstyle\downarrow$ & \begin{tabular}[b]{@{}c@{}}rich.\\CRPS$\scriptstyle\downarrow$\end{tabular} & \begin{tabular}[b]{@{}c@{}}rich.\\int.$\scriptstyle\downarrow$\end{tabular} & NLL$\scriptstyle\downarrow$ & AUC$\scriptstyle\uparrow$ & \begin{tabular}[b]{@{}c@{}}AUC\\rare$\scriptstyle\uparrow$\end{tabular} & \begin{tabular}[b]{@{}c@{}}AUC\\common$\scriptstyle\uparrow$\end{tabular} & \begin{tabular}[b]{@{}c@{}}dep.\\gain$\scriptstyle\uparrow$\end{tabular} & \begin{tabular}[b]{@{}c@{}}co-\\occ.$\scriptstyle\uparrow$\end{tabular} & \begin{tabular}[b]{@{}c@{}}seed\\SD$\scriptstyle\downarrow$\end{tabular} & \begin{tabular}[b]{@{}c@{}}fit\\time$\scriptstyle\downarrow$\end{tabular}\\
		\hline
		GEM without & 2.916 & 6.89 & 0.052 & 64.4 & 0.652 & 0.639 & 0.670 & 0.0043 & 0.50 & 0.0110 & 3 min\\
		\quad control variate & (0.048) & (0.11) & (0.030) & (2.1) &  &  &  &  &  &  & \\[1pt]
		GEM, single fit & 2.795 & 6.23 & 0.039 & 57.3 & 0.727 & 0.697 & 0.752 & 0.0056 & 0.77 & 0.0071 & 19 s\\
		 & (0.039) & (0.17) & (0.019) & (1.4) &  &  &  &  &  &  & \\[1pt]
		GEM, 5-seed ensemble & 2.782 & 6.17 & 0.031 & 56.1 & 0.738 & 0.717 & 0.755 & 0.0078 & 0.81 & --- & ---\\
		 & (0.038) & (0.16) & (0.022) & (1.2) &  &  &  &  &  &  & \\[1pt]
		\hline
	\end{tabular}
	}
\end{table}

\subsection{Metrics}

Every metric of a split is computed on its $n=600$ test plots from the $M=1{,}000$ draws $V_i^{(1)},\dots,V_i^{(M)}\in\{0,1\}^q$ of the model at plot $i$, where $q=242$, $y_i\in\{0,1\}^q$ is the observed community and $y_{ij}=1$ when species $j$ is present; the tables then average it over the splits. The AUC, the two interval errors, the nine measures of the overall score and the twelve of Table~\ref{tab:veg_benchmark} use a fixed random subset of $100$ of the draws, the same for every model, as the benchmark's protocol does; all other metrics use all $M$ draws.

\paragraph{Energy score and CRPS.} The energy score of plot $i$ is
\[
\frac1M\sum_{m=1}^M\|V_i^{(m)}-y_i\|-\frac1{2M(M-1)}\sum_{m\ne m'}\|V_i^{(m)}-V_i^{(m')}\|,
\]
the estimate from the draws of the energy score $\E\|V-y_i\|-\frac12\E\|V-V'\|$ of Section~\ref{sec:background}, averaged over the plots. The species are coded $0$ and $1$, so that $\|v-w\|$ is the square root of the number of species on which $v$ and $w$ differ; this is half the distance in the $\pm1$ coding of Table~\ref{tab:links} in which the model is fitted, and the reported energy score is half the score in that coding. The richness of $v\in\{0,1\}^q$ is $R(v)=\sum_{j=1}^q\mathbf 1\{v_j=1\}$, the number of species present. Its CRPS at plot $i$ is the energy score of the scalar $R$,
\[
\frac1M\sum_{m=1}^M\big|R_i^{(m)}-r_i\big|-\frac1{2M(M-1)}\sum_{m\ne m'}\big|R_i^{(m)}-R_i^{(m')}\big|,\qquad R_i^{(m)}=R(V_i^{(m)}),\quad r_i=R(y_i).
\]
For two plots with communities $u$ and $w$ that have $a=\sum_ju_jw_j$ species in common, the $\beta$-S{\o}rensen dissimilarity is
\[
\beta(u,w)=\frac{R(u)+R(w)-2a}{R(u)+R(w)},
\]
the number of species present at only one of the two plots divided by the sum of their richnesses \citep{baselga2010partitioning}, set to $0.995$ when both plots are empty, as in the benchmark's code. Its CRPS for a pair $\{i,i'\}$ of test plots is computed as for richness, from the $M$ values $\beta(V_i^{(m)},V_{i'}^{(m)})$ and the observed $\beta(y_i,y_{i'})$. It is averaged over $300$ pairs of distinct test plots drawn uniformly with replacement with a fixed seed, so that every model is scored on the same pairs, as in the benchmark.

\paragraph{Log score and AUC.} The marginal log score of plot $i$ is
\[
-\sum_{j=1}^q\log\hat p_{ij}(y_{ij}),\qquad \hat p_{ij}(1)=1-\hat p_{ij}(0)=\frac{M\bar V_{ij}+\frac12}{M+1},\quad \bar V_{ij}=\frac1M\sum_{m=1}^MV_{ij}^{(m)},
\]
where the smoothing keeps a species that is never drawn from receiving probability zero. Table~\ref{tab:veg_rare} reports it per species, averaged over the species of a prevalence class. Let $\tilde p_{ij}$ be the frequency of species $j$ at plot $i$ in the subset of $100$ draws, and $n_{j1}$ and $n_{j0}$ the numbers of test plots at which species $j$ is present and absent. The AUC of species $j$ is
\[
\frac1{n_{j1}n_{j0}}\sum_{i:\,y_{ij}=1}\ \sum_{i':\,y_{i'j}=0}\Big[\mathbf 1\{\tilde p_{ij}>\tilde p_{i'j}\}+\tfrac12\mathbf 1\{\tilde p_{ij}=\tilde p_{i'j}\}\Big],
\]
and the reported AUC is its mean over the species with $n_{j1}n_{j0}>0$. This is the plain AUC, not the benchmark's version discussed below.

\paragraph{Interval errors.} For the interval error of richness, let $F_i(r)$ and $F_i(r^-)$ be the proportions of the $100$ draws whose richness is at most $r$ and below $r$. The randomized probability integral transform $U_i$ is uniform on $[F_i(r_i^-),F_i(r_i)]$, and it is uniform on $[0,1]$ when the forecast is calibrated \citep{czado2009predictive}. The coverage of the central $50\%$ interval is $c=n^{-1}\sum_i\P(\frac14\le U_i\le\frac34)$, computed exactly, the probability being $\mathbf 1\{\frac14\le F_i(r_i)\le\frac34\}$ when $F_i(r_i^-)=F_i(r_i)$, and the interval error is $|\frac12-c|$. The $\beta$ interval error of Table~\ref{tab:veg_interp_full} is the same quantity for the dissimilarity over the $300$ pairs.

\paragraph{Dependence gain and co-occurrence correlation.} The dependence gain uses the set $\mathcal C$ of common species, those present on at least $10\%$ of the training plots of the split, $76$ to $79$ of the $242$ on the five splits. For $j,k\in\mathcal C$, let $N_{ijk}$ be the number of draws in which $(V_{ij}^{(m)},V_{ik}^{(m)})$ equals the observed pair $(y_{ij},y_{ik})$, and let $\bar p_{ij}(1)=1-\bar p_{ij}(0)=\bar V_{ij}$. The model gives the observed pair the probability $(N_{ijk}+\frac12)/(M+2)$, and the product of its own marginals gives it $\{M\bar p_{ij}(y_{ij})\bar p_{ik}(y_{ik})+\frac12\}/(M+2)$, both smoothed as counts over the four cells of the pair. The dependence gain, in nats, is the mean log ratio of the two,
\[
\frac1n\sum_{i=1}^n\binom{|\mathcal C|}{2}^{-1}\sum_{j<k,\ j,k\in\mathcal C}\log\frac{N_{ijk}+\frac12}{M\bar p_{ij}(y_{ij})\bar p_{ik}(y_{ik})+\frac12}.
\]
For a model whose draws are independent across species it is zero up to Monte Carlo error, so such models show values near $-0.0006$. The co-occurrence correlation of Tables~\ref{tab:veg_interp_full} and~\ref{tab:veg_gem} is the Pearson correlation, over the pairs $j<k$ in $\mathcal C$, between the predicted residual co-occurrence $n^{-1}\sum_i(M^{-1}\sum_mV_{ij}^{(m)}V_{ik}^{(m)}-\bar V_{ij}\bar V_{ik})$ and the observed $n^{-1}\sum_i(y_{ij}-\pi_{ij})(y_{ik}-\pi_{ik})$. Here $\pi_{ij}$ is the occurrence probability that the GAM fitted to the same split gives species $j$ at plot $i$, a reference shared by all models.

\paragraph{The overall score.} The overall score of Table~\ref{tab:veg_main} has the benchmark's form. Each of nine measures is standardized across the models of the comparison fitted on the split, as $z=(x-\bar x)/s$, where $\bar x$ and $s$ are the mean and the standard deviation over those models, the latter with divisor the number of models. The standardized value $z$ is clipped to $[-2,2]$ and signed so that larger is better, and the nine are averaged. At the occurrence level the measures are the mean of $|\tilde p_{ij}-y_{ij}|$, the AUC, and the reliability: the plots are ordered by $\tilde p_{ij}$ and cut into ten bins of equal size, and the absolute difference between the means of $\tilde p_{ij}$ and of $y_{ij}$ in a bin is averaged over the bins and the species. At the richness level they are the root mean squared error of the predictive mean richness, the Spearman correlation between the predictive mean and $r_i$, and the interval error. At the composition level they are the same three for the S{\o}rensen dissimilarity and its turnover and nestedness components \citep{baselga2010partitioning} over the $300$ pairs; the three indices are standardized separately and their $z$ averaged into one measure.

\paragraph{The benchmark's measures.} We ported the benchmark's evaluation code and checked the port against the authors' own \textsf{R} functions on identical predictions and identical plot pairs, to agreement within $2\times10^{-15}$ on all $20$ values. Table~\ref{tab:veg_benchmark} gives its twelve measures as its code computes them. Four of them differ from the text that describes them: the accuracy of richness and of the dissimilarity indices is the mean absolute error of individual draws rather than the root mean squared error of the predictive mean, and the two rank correlations are computed per draw and averaged. Both reward predictions that are too narrow. Writing $X_1,\dots,X_M$ for the draws, $M^{-1}\sum_m|X_m-y| = \mathrm{CRPS}(\hat F,y) + \tfrac12\E|X-X'|$, so the coded accuracy is a proper score plus a penalty on predictive spread, minimized by a point mass at the conditional median. A rank correlation computed draw by draw is attenuated by exactly the sampling variation that a correctly dispersed predictive distribution is supposed to have. On richness the penalty reverses the comparison. GEM scores $12.72$ against $9.56$ for the random forest with coordinates, of which $6.17$ and $7.33$ are CRPS and the remainder is spread, whereas on the root mean squared error of the predictive mean it is ahead, $11.44$ against $11.80$, on all five splits. We therefore compute those four measures from the predictive distribution, and replace the benchmark's occurrence AUC, which follows \texttt{pROC} in flipping a species scored below $0.5$ to its complement, by the plain AUC, its occurrence calibration by the mean absolute difference between predicted and observed frequency over deciles of the predicted probability, and the plain coverage of a central interval by the randomized probability integral transform. If only the four draw-based measures are corrected, and the AUC, the occurrence calibration and the two interval measures are kept as the benchmark codes them, the same seven models stay at the top, but the classifier chain with a spatial basis comes first, at $1.37$ against GEM's $1.10$. Those four measures therefore carry most, but not all, of the correction. On the species level the two versions agree: GEM's AUC is $0.739$ as coded and $0.738$ plain, its occurrence calibration is $0.032$ against a range of $0.032$ to $0.034$ over all $17$ models, and its sharpness is $0.192$ against a range of $0.151$ to $0.213$.

\begin{table}[t]
	\centering
	\caption{The benchmark's twelve measures, exactly as its code computes them, on the vegetation data, averaged over the five splits (the benchmark's split only for HMSC with spatial latent factors). Accuracy, discrimination, calibration and precision at the three levels of \citet{norberg2019comprehensive}; the composition columns average the three $\beta$-diversity indices. Accuracy at the richness and composition levels is the mean absolute error of individual draws and discrimination there is a per-draw rank correlation, so models that draw species independently score well on them; see the text. Precision, a predictive standard deviation, has no arrow because neither direction is better on its own.}
	\label{tab:veg_benchmark}
	\footnotesize\setlength{\tabcolsep}{3pt}
	{%
	\begin{tabular}{lcccccccccccc}
		model & \multicolumn{4}{c}{occurrence} & \multicolumn{4}{c}{richness} & \multicolumn{4}{c}{composition}\\
		 & acc.$\scriptstyle\downarrow$ & AUC$\scriptstyle\uparrow$ & calib.$\scriptstyle\downarrow$ & prec. & acc.$\scriptstyle\downarrow$ & disc.$\scriptstyle\uparrow$ & calib.$\scriptstyle\downarrow$ & prec. & acc.$\scriptstyle\downarrow$ & disc.$\scriptstyle\uparrow$ & calib.$\scriptstyle\downarrow$ & prec.\\
		\hline
		GEM, 5-seed ensemble & 0.135 & 0.739 & 0.032 & 0.192 & 12.72 & 0.461 & 0.060 & 11.75 & 0.191 & 0.209 & 0.053 & 0.183\\
		GEM, single fit & 0.135 & 0.728 & 0.032 & 0.188 & 12.69 & 0.463 & 0.059 & 11.62 & 0.191 & 0.206 & 0.054 & 0.181\\
		Random forest + coord. & 0.136 & 0.734 & 0.032 & 0.193 & 9.56 & 0.601 & 0.263 & 3.96 & 0.147 & 0.315 & 0.269 & 0.074\\
		GAM + spatial smooth & 0.137 & 0.739 & 0.033 & 0.193 & 9.66 & 0.588 & 0.256 & 3.98 & 0.152 & 0.270 & 0.277 & 0.075\\
		GAM & 0.140 & 0.728 & 0.033 & 0.200 & 9.74 & 0.571 & 0.254 & 4.03 & 0.153 & 0.258 & 0.283 & 0.076\\
		HMSC, covariates only & 0.142 & 0.746 & 0.033 & 0.204 & 9.79 & 0.590 & 0.261 & 4.09 & 0.154 & 0.256 & 0.266 & 0.077\\
		HMSC, latent factors & 0.145 & 0.695 & 0.032 & 0.213 & 15.20 & 0.328 & 0.078 & 14.24 & 0.199 & 0.102 & 0.057 & 0.177\\
		HMSC, spatial latent f. & 0.140 & 0.731 & 0.032 & 0.208 & 14.15 & 0.419 & 0.098 & 13.73 & 0.193 & 0.175 & 0.042 & 0.170\\
		gllvm & 0.140 & 0.731 & 0.034 & 0.189 & 15.04 & 0.399 & 0.092 & 14.19 & 0.189 & 0.170 & 0.056 & 0.159\\
		Classifier chain & 0.135 & 0.735 & 0.032 & 0.186 & 13.09 & 0.426 & 0.034 & 11.41 & 0.188 & 0.174 & 0.045 & 0.164\\
		Classifier chain + basis & 0.125 & 0.735 & 0.032 & 0.151 & 11.91 & 0.501 & 0.020 & 9.84 & 0.178 & 0.247 & 0.036 & 0.152\\
		\hline
	\end{tabular}
	}
\end{table}

\subsection{Further results}

\paragraph{Per-split differences.} Figure~\ref{fig:veg_paired} repeats the main comparison as paired per-split differences against the five-seed ensemble, which removes the split-to-split variation common to the two models. That common variation is most of the split-to-split variation. The energy score varies by about $4\%$ across the models and, for each model, with a standard deviation of $0.028$ to $0.039$ across the splits, but the paired difference between GEM and the random forest with coordinates has a standard deviation of only $0.012$ (mean $0.0104$), and $0.009$ to $0.015$ for the two chains and the GAM with a spatial smooth.

\begin{figure}[t]
\centering
\includegraphics[width=\textwidth]{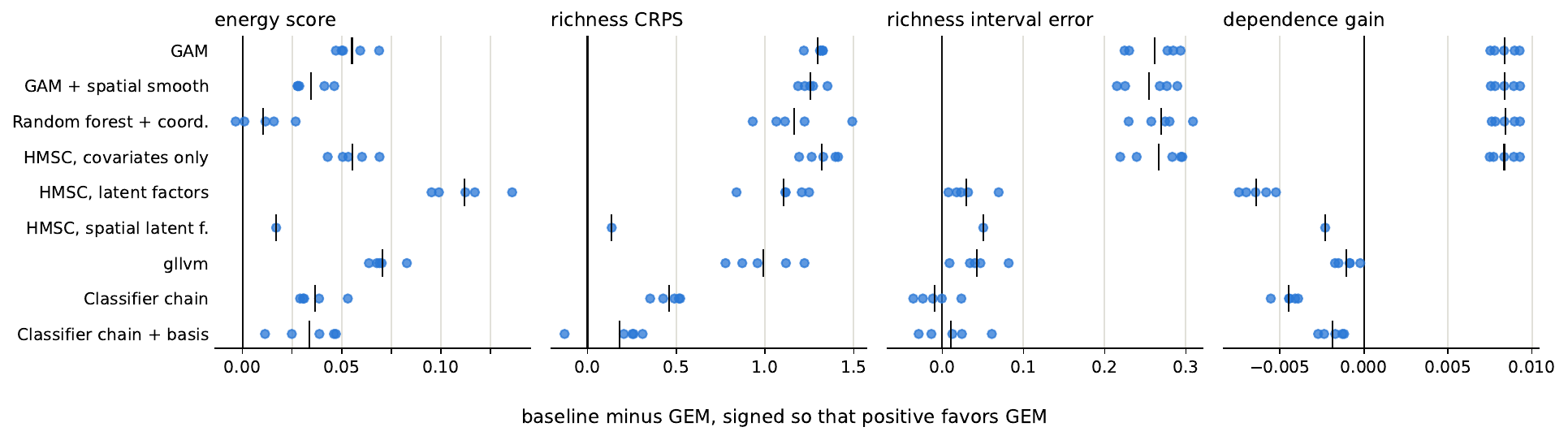}
\caption{Joint species distribution modeling on the vegetation data: paired per-split differences between nine of the fifteen baselines and the GEM ensemble on the same split, signed so that a positive value favors GEM. Each dot is one split; the vertical bar is the mean over splits. HMSC with spatial latent factors contributes one split.}\label{fig:veg_paired}
\end{figure}

\paragraph{Prevalence classes.} Table~\ref{tab:veg_rare} splits the species-level measures by training prevalence. HMSC without associations is the best model on the species below $5\%$ training prevalence, with an AUC of $0.746$ against $0.717$ for GEM, and is in the bottom third on the common ones, at $0.734$ against $0.755$; GEM is 8th of the $17$ models on the rare species and 3rd on the common ones.

\begin{table}[t]
	\centering
	\caption{Vegetation data: species-level accuracy by training prevalence, below $5\%$, $5$ to $10\%$ and at least $10\%$, averaged over the five splits (the benchmark's split only for HMSC with spatial latent factors). The classes are formed within each split and hold on average 120.6, 43.8 and 77.6 of the $242$ species. NLL is the marginal log score of Table~\ref{tab:veg_main}.}
	\label{tab:veg_rare}
	\footnotesize\setlength{\tabcolsep}{3pt}
	\begin{tabular}{lcccccc}
		model & \multicolumn{3}{c}{AUC$\scriptstyle\uparrow$} & \multicolumn{3}{c}{NLL$\scriptstyle\downarrow$}\\
		 & rare & uncommon & common & rare & uncommon & common\\
		\hline
		GEM, 5-seed ensemble & 0.717 & 0.767 & 0.755 & 0.0948 & 0.2100 & 0.4571\\
		GEM, single fit & 0.697 & 0.763 & 0.752 & 0.1002 & 0.2150 & 0.4608\\
		GAM + spatial smooth & 0.719 & 0.768 & 0.751 & 0.0948 & 0.2101 & 0.4619\\
		Random forest + coord. & 0.705 & 0.764 & 0.763 & 0.0997 & 0.2126 & 0.4539\\
		Boosted trees + coord. & 0.645 & 0.742 & 0.753 & 0.0980 & 0.2134 & 0.4595\\
		GAM & 0.711 & 0.756 & 0.734 & 0.0948 & 0.2125 & 0.4684\\
		Probit GLM & 0.724 & 0.759 & 0.733 & 0.0995 & 0.2149 & 0.4711\\
		MARS & 0.726 & 0.751 & 0.726 & 0.0965 & 0.2141 & 0.4725\\
		Random forest & 0.682 & 0.735 & 0.737 & 0.1058 & 0.2246 & 0.4744\\
		Boosted trees & 0.641 & 0.726 & 0.736 & 0.0986 & 0.2175 & 0.4682\\
		Gradient nearest neighbor & 0.665 & 0.727 & 0.735 & 0.1261 & 0.2744 & 0.5272\\
		HMSC, covariates only & 0.746 & 0.761 & 0.734 & 0.0917 & 0.2105 & 0.4691\\
		HMSC, spatial latent f. & 0.721 & 0.747 & 0.737 & 0.0931 & 0.2140 & 0.4712\\
		HMSC, latent factors & 0.687 & 0.710 & 0.694 & 0.0957 & 0.2207 & 0.4910\\
		gllvm & 0.720 & 0.756 & 0.731 & 0.1000 & 0.2160 & 0.4746\\
		Classifier chain + basis & 0.706 & 0.767 & 0.761 & 0.1051 & 0.2274 & 0.4749\\
		Classifier chain & 0.722 & 0.761 & 0.740 & 0.0944 & 0.2116 & 0.4668\\
		\hline
	\end{tabular}
\end{table}

\paragraph{Pairwise dependence.} The dependence gain, defined above, measures what a model's dependence structure adds to the log score of pairs of common species. GEM gains $0.0078$ nats per pair, less than each of the five joint baselines, which range from $0.0088$ for gllvm to $0.0142$ for HMSC with latent factors, and it is lower than each of them on every split (Table~\ref{tab:veg_interp_full}, Figure~\ref{fig:veg_paired}). The correlation between predicted and observed residual co-occurrence largely agrees. It is $0.81$ for GEM, level with $0.80$ for gllvm, against $0.89$ to $0.92$ for HMSC and the chains. The latent-factor models recover more of the pairwise association at the expense of the marginals (Figure~\ref{fig:veg_scatter}(b)).

\paragraph{Spatial variants.} A spatial basis leaves most of the classifier chain's dependence gain in place, $0.0096$ against $0.0123$ (Table~\ref{tab:veg_interp_full}). The coordinates matter more for the marginals of the tree models: they improve the energy score of the random forest from $2.856$ to $2.792$ and its AUC from $0.709$ to $0.734$, whereas a spatial smooth moves the GAM's AUC by $0.012$, from $0.726$ to $0.738$.

\clearpage
\section{Details for the health examination survey application}\label{app:nhanes}

\paragraph{Data.} Age is top-coded at $80$. No participant is excluded for a missing covariate; a missing education level, marital state or smoking state sets all indicators of that variable to zero. Systolic and diastolic blood pressure are the means of up to four readings. Body mass index, waist circumference, HbA1c and HDL cholesterol are log-transformed, and the seven continuous outcomes are standardized with the means and standard deviations of the training set of each split. The binary outcomes are diagnosed diabetes (NHANES item DIQ010, a borderline answer counting as no), diagnosed hypertension (BPQ020), blood-pressure medication (BPQ050A), diagnosed high cholesterol (BPQ080) and cholesterol medication (BPQ100D). An item that the questionnaire skips because of an earlier answer is coded as no. Self-rated health is item HUQ010. The $29$ covariate columns code ten variables. Age and the income-to-poverty ratio are standardized; the ratio is imputed by the training median and has an indicator of missingness ($8.9\%$ missing). The other variables are coded by indicators: female sex, five race and ethnicity groups, five education levels, three marital states, three smoking states (never, former, current), recreational physical activity, alcohol use and the six survey cycles. Alcohol use also has an indicator of missingness ($5.9\%$ missing; the question changed in the last cycle). The $20$ splits have $16{,}789$ training, $5{,}596$ validation and $5{,}598$ test participants.

\paragraph{Methods.} GEM is fitted to a $13$-dimensional coding of the outcome in which self-rated health is one coordinate equal to its level, with the ordinal link of Table~\ref{tab:links}. For scoring, its draws are converted to the $17$-dimensional coding of Section~\ref{sec:nhanes}, which the baselines use throughout. Its settings are those of Table~\ref{tab:gem_settings}; on the continuous coordinates, $\sigma$ as well as $g$ receives pathwise derivatives through the joint norm. Every network of the baselines is trained by Adam at learning rate $10^{-4}$ on the full batch for $5{,}000$ iterations, keeping the parameters with the lowest validation loss, computed every $50$ iterations and after the last. The stacked marginals fit the continuous block by the model \eqref{eq:gem_model} with the identity link, with $g$ of two hidden layers of $256$ units, a $7$-dimensional $\varepsilon$ entering at the input only and a common scale. They fit each binary label by a network $29\to256\to256\to1$ trained by the binary cross-entropy on standardized covariates, and self-rated health by a network $29\to256\to256\to5$ with a softmax output trained by the cross-entropy. Their draws are independent across blocks and labels. The Gaussian copula couples the same fits through the correlation matrix $R$ of normal scores computed on the training set. For a continuous coordinate, the normal score is based on the rank of its residual (the observation minus the mean of $200$ draws) among the training residuals. For a binary label, and for self-rated health taken as one coordinate, it is based on a randomized probability integral transform. A draw maps $\Phi(Z)$, $Z\sim\mathcal N(0,R)$, coordinate by coordinate through the fitted marginals: the empirical quantile function of $500$ draws for a continuous coordinate and the fitted probabilities for the others. DRF is the \textsf{R} package \texttt{drf}, version 1.3.1, at its defaults, which grow $3{,}000$ honest trees on half-samples with the MMD splitting rule, a minimum node size of $15$ and response scaling. It is fitted to the training set with the $17$-dimensional outcome, and its draws resample training outcome vectors with the forest's weights at $x$.

\paragraph{Metrics and timing.} All metrics use $500$ draws per test point, in the $17$-dimensional coding. The second term of the energy score averages $25$ disjoint pairs from a random subset of $50$ draws, the CRPS uses all pairs, and the ranked probability score uses the smoothed frequencies $(M_k+1)/(M+5)$ of the $M_k$ draws at level $k$. Fit times were measured on shared hardware. Two GEM fits shared a GPU, and DRF ran on $12$ CPU threads alongside other jobs. The fit time of GEM covers the $5{,}000$ iterations and the validation evaluations, that of the stacked marginals its seven networks, and that of DRF the whole call to \textsf{R}, sampling included.

\begin{table}[t]
\centering
\caption{Absolute values behind Figure~\ref{fig:nhanes}, as mean (standard error) over
the $20$ splits: joint energy score, CRPS averaged over the seven continuous
coordinates on the standardized scale, Bernoulli log-loss averaged over the five binary
labels, ranked probability score of self-rated health, and fit time in seconds
(GEM and the stacked marginals on one GPU, DRF on CPU). The standard
errors are larger than those of the paired differences because most of the
split-to-split variation is shared by all methods (see the text of this appendix).
The Gaussian copula is fitted on top of the stacked marginals, so its total cost is
that of the stacked marginals plus its own fit, which takes under a second and is the
time shown.
Bold: the best mean in each of the four metric columns, ties at the printed precision included.}
\label{tab:nhanes_absolute}
{\small\setlength{\tabcolsep}{4pt}
\begin{tabular}{l c c c c c}
\hline
 & joint & continuous & binary & ordinal & \\
method & energy score~$\downarrow$ & CRPS~$\downarrow$ & log-loss~$\downarrow$ & RPS~$\downarrow$ & fit time (s)~$\downarrow$ \\
\hline
GEM & \textbf{2.0880} & \textbf{0.4963} & \textbf{0.4534} & 0.4980 & 62.095 \\
 & (0.0021) & (0.0007) & (0.0008) & (0.0010) & (0.312) \\[2pt]
DRF & 2.1034 & 0.5001 & 0.4549 & 0.5057 & 243.8 \\
 & (0.0023) & (0.0008) & (0.0009) & (0.0010) & (16.2) \\[2pt]
Gaussian copula & 2.0951 & 0.4982 & 0.4567 & \textbf{0.4954} & 0.4278 \\
 & (0.0022) & (0.0007) & (0.0009) & (0.0010) & (0.0213) \\[2pt]
Stacked marginals & 2.0969 & 0.4973 & 0.4567 & \textbf{0.4954} & 55.306 \\
 & (0.0021) & (0.0007) & (0.0009) & (0.0010) & (1.617) \\
\hline
\end{tabular}}
\end{table}

\begin{table}[t]
\centering
\caption{Marginal CRPS of the seven continuous outcomes, in original units. The
``GEM score'' column is the CRPS of GEM averaged over the 20
splits; each baseline column is the paired difference baseline $-$ GEM (mean $\pm$
standard error over the splits), in the same units. Positive means GEM has the lower CRPS. $^{*}$ marks $|\Delta| > 2\,\mathrm{SE}$.}
\label{tab:nhanes_paired_crps}
{\small\setlength{\tabcolsep}{4pt}
\begin{tabular}{l r r r r }
\hline
outcome & GEM score & DRF & Copula & Stacked \\
\hline
BMI (kg/m$^2$) & 3.428 & +0.078$\,\pm\,$0.003$^{*}$ & +0.014$\,\pm\,$0.002$^{*}$ & +0.008$\,\pm\,$0.002$^{*}$ \\
waist (cm) & 8.271 & +0.189$\,\pm\,$0.008$^{*}$ & +0.035$\,\pm\,$0.005$^{*}$ & +0.021$\,\pm\,$0.004$^{*}$ \\
systolic BP (mmHg) & 8.545 & +0.004$\,\pm\,$0.004 & +0.033$\,\pm\,$0.005$^{*}$ & +0.018$\,\pm\,$0.005$^{*}$ \\
diastolic BP (mmHg) & 5.968 & +0.034$\,\pm\,$0.004$^{*}$ & +0.022$\,\pm\,$0.004$^{*}$ & +0.010$\,\pm\,$0.004$^{*}$ \\
HbA1c (\%) & 0.4172 & $-$0.0105$\,\pm\,$0.0004$^{*}$ & +0.0016$\,\pm\,$0.0003$^{*}$ & +0.0009$\,\pm\,$0.0003$^{*}$ \\
total cholesterol (mg/dL) & 21.61 & +0.07$\,\pm\,$0.02$^{*}$ & +0.08$\,\pm\,$0.01$^{*}$ & +0.03$\,\pm\,$0.01$^{*}$ \\
HDL cholesterol (mg/dL) & 7.939 & +0.139$\,\pm\,$0.006$^{*}$ & +0.033$\,\pm\,$0.004$^{*}$ & +0.017$\,\pm\,$0.003$^{*}$ \\
\hline
\end{tabular}}
\end{table}

\begin{table}[t]
\centering
\caption{Bernoulli log-loss of the five binary outcomes, from the smoothed marginal
probabilities $\hat p_j(x)=(n_+ + 1)/(M+2)$ of the $M=500$ draws. Columns as in
Table~\ref{tab:nhanes_paired_crps}: a positive entry means GEM
has the lower log-loss; $^{*}$ marks $|\Delta| > 2\,\mathrm{SE}$.}
\label{tab:nhanes_paired_logloss}
{\small\setlength{\tabcolsep}{4pt}
\begin{tabular}{l r r r r }
\hline
outcome & GEM score & DRF & Copula & Stacked \\
\hline
diabetes diagnosis & 0.3316 & +0.0011$\,\pm\,$0.0007 & +0.0027$\,\pm\,$0.0007$^{*}$ & +0.0026$\,\pm\,$0.0006$^{*}$ \\
hypertension diagnosis & 0.5403 & +0.0033$\,\pm\,$0.0006$^{*}$ & +0.0019$\,\pm\,$0.0004$^{*}$ & +0.0018$\,\pm\,$0.0005$^{*}$ \\
BP medication & 0.4476 & +0.0037$\,\pm\,$0.0006$^{*}$ & +0.0039$\,\pm\,$0.0005$^{*}$ & +0.0040$\,\pm\,$0.0005$^{*}$ \\
high-cholesterol diagnosis & 0.5587 & +0.0005$\,\pm\,$0.0005 & +0.0041$\,\pm\,$0.0006$^{*}$ & +0.0042$\,\pm\,$0.0007$^{*}$ \\
cholesterol medication & 0.3887 & $-$0.0011$\,\pm\,$0.0006 & +0.0040$\,\pm\,$0.0008$^{*}$ & +0.0043$\,\pm\,$0.0007$^{*}$ \\
\hline
\end{tabular}}
\end{table}

Table~\ref{tab:nhanes_absolute} gives the absolute values of the four metrics in
Figure~\ref{fig:nhanes}, together with the fit times. Tables~\ref{tab:nhanes_paired_crps} and~\ref{tab:nhanes_paired_logloss} break the
marginal comparison down by outcome. Because all methods are evaluated on the same
20 train/validation/test splits, we report paired differences. The split-to-split variation in the absolute scores is shared by all methods and cancels in the difference, so the standard errors of $\Delta$ are about two to eight times smaller than those of the scores. Two caveats apply to the stars: the splits re-partition the same participants, so the standard errors somewhat understate the uncertainty, and with 36 comparisons some stars are expected by chance. We therefore rely on the sign of the differences and its consistency across outcomes rather than on single entries.

On the continuous outcomes, GEM has a lower CRPS than every baseline for every variable except HbA1c, where DRF is better by $2.5\%$ of the score. DRF resamples observed outcome vectors and so reproduces the long right tail of HbA1c among diagnosed diabetics, which GEM does not fully capture. Elsewhere the margin over DRF reaches
$2.3\%$ (BMI, waist) and $1.8\%$ (HDL). Against the copula and stacked marginals
the gains are small, $0.1$--$0.4\%$ of the score, but consistent across all seven
outcomes. On the binary outcomes, GEM has the lower
log-loss for every label against the copula and stacked marginals, whose binary
marginals are fitted by independent classifiers. Against DRF it is better on four labels, on two of them by more than two standard errors, while cholesterol medication is slightly in DRF's favor, by less than two standard errors. Against
the marginal baselines the improvements are largest for the medication indicators
and diagnosed diabetes ($0.8$--$1.1\%$ of the score), where the labels depend strongly
on one another (blood-pressure medication implies a hypertension diagnosis) or, for diagnosed diabetes, on the measured HbA1c.

\end{document}